\documentclass{article}

\usepackage{amsmath}
\usepackage{amssymb}
\usepackage{amsfonts}
\usepackage{algpseudocode}
\usepackage{authblk}

\usepackage{graphicx}
\usepackage{textcomp}
\usepackage{xcolor}

\usepackage{tikz}
 \usetikzlibrary{trees}
 \usetikzlibrary{shapes}
 \usetikzlibrary{fit}
 \usetikzlibrary{shadows}
 \usetikzlibrary{backgrounds}
 \usetikzlibrary{arrows,automata}

\tikzset{
   n/.style= {circle,fill,inner sep=1.5pt,node distance=2cm}
  ,acc/.style={circle,draw,inner sep=3pt,node distance=2cm}
  ,phantom/.style={circle},
  ,arr/.style={->, >=stealth, semithick, shorten <= 3pt, shorten >= 3pt}
}

\usepackage{xspace}
\usepackage{paralist}
\usepackage{hyperref}
\usepackage[linesnumbered,ruled,vlined]{algorithm2e}

\SetKw{Break}{break}
\SetCommentSty{mycommfont}

\def\BibTeX{{\rm B\kern-.05em{\sc i\kern-.025em b}\kern-.08em
    T\kern-.1667em\lower.7ex\hbox{E}\kern-.125emX}}

\usepackage[appendix=inline]{apxproof} % leave in place

\newcommand{\pair}[1]{\ensuremath{\langle {#1} \rangle}\xspace}
\newcommand{\set}[1]{\ensuremath{\{ {#1} \}}\xspace}
\newcommand{\games}[4]{\ensuremath{\mathit{#1}({#2},{#3},{#4})}\xspace}
\newcommand{\streett}[3]{\games{Streett}{#1}{#2}{#3}}
\newcommand{\rabin}[3]{\games{Rabin}{#1}{#2}{#3}}
\newcommand{\parity}[3]{\games{Parity}{#1}{#2}{#3}}
\newcommand{\el}[3]{\games{EL}{#1}{#2}{#3}}
\newcommand{\attr}{\ensuremath{{\mathsf{Attr}}}\xspace}
\newcommand{\cpre}{\ensuremath{{\mathsf{CPre}}}\xspace}
\newcommand{\Win}{\ensuremath{{\mathit{Win}}}\xspace}

\newcommand{\incr}{\ensuremath{{\mathit{upd}}}\xspace}
\newcommand{\ind}{\ensuremath{{\mathit{ind}}}\xspace}

\newcommand{\child}{\ensuremath{{\mathit{child}}}\xspace}
\newcommand{\leaf}{\ensuremath{{\mathit{leaf}}}\xspace}
\newcommand{\grandchild}{\ensuremath{{\mathit{grandchild}}}\xspace}
\newcommand{\inft}{\mathit{inf}}
\newcommand{\subarena}{{S\!A}}

\SetKwComment{Comment}{$\triangleright$\ }{}

\newtheoremrep{theorem}{Theorem}
\newtheoremrep{lemma}[theorem]{Lemma}
\newtheoremrep{corollary}[theorem]{Corollary}
\newtheoremrep{definition}[theorem]{Definition}

\newtheoremrep{example}[theorem]{Example}
\newtheoremrep{remark}[theorem]{Remark}

\newlength{\defaulttextfloatsep}
\newlength{\defaultfloatsep}
\title{Solving Streett and Emerson-Lei Games \\ with 
Universal Trees
}

\author[1]{Daniel Hausmann}
\author[2]{Marcin Jurdzinski}
\author[3]{Nir Piterman}

\affil[1]{University of Liverpool, Liverpool, UK}
\affil[2]{University Warwick, Coventry, UK}
\affil[3]{University of Gothenburg and Chalmers University of Technology, Gothenburg, Sweden}

\begin{document}

\maketitle

\begin{abstract}
    Nearly a decade ago, Calude et al.~showed that parity games can be solved in quasi-polynomial time.
    This result is now understood in terms of \emph{universal trees}.
    By reduction to parity games, the quasi-polymonial result can benefit all $\omega$-regular games. 
    However, beyond such reductions—and with the exception of Rabin games—our understanding of the role of universal trees in direct solutions is still quite limited.
    In this work, we refute the common view that universal trees are relevant only for games that admit memoryless winning strategies. We contribute a full understanding of how universal trees interact with 
    Zielonka trees for the solution of Streett and Emerson-Lei games.

    As a consequence, we show that winning regions and strategies in Streett games with $n$ vertices, $m$ edges, and $k$ pairs can be computed in time $O(mk\log(k)k!|U(n,k)|)$, where $U(n,k)$ is a universal tree for $n$ leaves and depth $k$.
    This improves upon the best previously known complexity result for Streett games, which relied on reduction to parity games and their quasi-polynomial solution.
    
    Furthermore, we show that winning regions and strategies for Emerson-Lei games with $n$ vertices, $m$ edges, and $c$ colors can be computed in time $O(mc\log(c)c!|U(n,c/2)|)$, again improving over reductions to parity games. 
    Notably, our approach yields memory-optimal strategies, in contrast to those obtained via reductions to parity games.
    Finally, we show how universal trees can be used to bound the recursion tree of the Zielonka-McNaughton algorithm for Emerson-Lei games. This leads to a symbolic algorithm that replaces the factor $n^c$ in the time complexity of existing symbolic approaches with $|U(n,c)|$.     
\end{abstract}

\section{Introduction}
\label{sec:intro}

Infinite duration two-player games have been established 
as a central formalism for the solution of decision problems for various temporal logics.
They encode the interaction of systems (player 0) with their environment (player 1),
and qualitative temporal properties are evaluated by classifying infinite sequences of interactions as winning or losing for
a player. Such qualitative properties subsume safety, reachability, or general liveness
properties, and are typically given in the form of $\omega$-regular winning conditions.

Parity conditions are a well-known example as they 
are relatively simple, yet expressive enough to encode any $\omega$-regular winning condition. In parity games, good and bad events are prioritized
relatively to each other, and the
system player aims to ensure that the highest priority event that occurs infinitely often is good.
Further typical examples include Rabin~\cite{DBLP:journals/jsyml/ElgotR66}, Streett~\cite{DBLP:conf/stoc/Streett81}, and Emerson-Lei~\cite{DBLP:journals/scp/EmersonL87}
conditions that allow to express
more involved combinations of events that should, or should not, occur infinitely often.
Streett conditions correspond to strong fairness and allow to produce strong fair strategies. 
In addition, counter strategies are an important tool to debug unrealizable specifications requiring to solve Streett conditions even if the main interest is a Rabin condition. 
Emerson-Lei conditions have recently attracted
particular interest~\cite{DBLP:conf/atva/RenkinDP20,DBLP:conf/fossacs/HausmannLP24,DBLP:conf/ijcai/AminofGRV25,KR2025-78} as they are general, succinct, and closed under disjunction and conjunction, which enables
modular transformations and constructions of automata and games.

Deciding games with $\omega$-regular winning conditions then amounts to computing the regions in which one of the players
has a strategy to win, that is, to ensure that their objective is satisfied by all interactions that can
arise using that strategy; a full solution includes also the computation of a witnessing winning strategy.

A prominent field in which two-player games are of central importance is \emph{reactive synthesis}~\cite{DBLP:conf/icalp/PnueliR89},
where the task is to decide the realizability of the specification of a desired input/output behaviour, typically
given as an LTL formula, and -- if possible -- to provide a witnessing (synthesized) system. 
This problem is of high interest as it enables the automatic fabrication of safe-by-construction systems. 
The most efficient translations of realizability to games result in Emerson-Lei games \cite{DBLP:journals/corr/abs-1709-02102,DBLP:conf/atva/MajorBSSZ19,DBLP:conf/cav/LiTFVZ22,DBLP:journals/fmsd/RenkinSDP22}, but can also be reduced to the solution of parity games~\cite{DBLP:conf/cav/MeyerSL18,DBLP:conf/atva/RenkinDP20}. 
Synthesized systems are obtained by extracting winning strategies.

In more detail, the winning condition in game reductions for reactive synthesis is derived from the input LTL specification.
The LTL specification is converted to a deterministic automaton, which is then typically
\emph{paritized}, using a latest-appearance-record (LAR) construction to assign priorities to sequences of events, thereby incurring blowup factorial in the number of events~\cite{DBLP:conf/stoc/GurevichH82}. 
An alternative reduction has recently been proposed~\cite{DBLP:journals/theoretics/CasaresCFL24}, using Zielonka trees to transform Muller automata to history-deterministic Rabin automata or
deterministic parity automata, but still incurring factorial blowup in the worst case.
%\np{I think that factorial in $c/2$ is still factorial?}
%\dan{yes}
Games with Streett or Emerson-Lei winning conditions are more succinct and in principle
sideline the paritization by enabling more direct evaluation.
While this direct game solution approach eliminates the need to go through parity (or Rabin)
conditions, it comes with the caveat that the analysis of Streett and Emerson-Lei games is more involved. 

Breakthrough progress on the solution of parity games has shown that 
they can be solved in time quasi-polynomial in the number of game vertices and the
number of priorities~\cite{DBLP:conf/stoc/CaludeJKL017}. 
Subsequently, this result has been consolidated (e.g.~\cite{DBLP:conf/lics/JurdzinskiL17,DBLP:journals/lmcs/LehtinenB20}), and the essential ingredient unifying all known quasi-polynomial methods has been identified and abstracted, leading to the notion of \emph{universal trees}~\cite{DBLP:conf/soda/CzerwinskiDFJLP19}; 
tight upper and lower quasi-polynomial bounds on the size of universal trees have been derived~\cite{DBLP:conf/soda/CzerwinskiDFJLP19}.
While recent work has shown how (colored) universal trees can be used to solve Rabin games~\cite{ DBLP:conf/tacas/MajumdarST24}, further progress has been hindered by the non-positionality of winning conditions that go beyond Rabin conditions: winning strategies in Streett and Emerson-Lei games may require memory as analyzed in detail in~\cite{DBLP:conf/lics/DziembowskiJW97}, using the concept of Zielonka trees. 
Up to now it has remained unclear whether universal trees can be used directly to solve such games.

In this context, our technical contributions are as follows:

\begin{compactitem}
\item[--] \emph{Streett games.}
We introduce a notion of rankings for Streett games based on universal
trees and employ these rankings within a rank-lifting algorithm to compute winning regions
and winning strategies. 
By using quasi-polynomial universal trees instead of exponentially large ones, 
we improve the known upper bound on time complexity for the computation of winning regions and strategies in Streett games
with $n$ vertices, $m$ edges, and $k$ pairs from 
$O(mk!)(nk!)^{1+o(1)}$ to $O\left(k\log(k)mk!\right) n^{\log{k}+O(1)}$.
The proposed algorithm is
\emph{asymmetric} (focused on one player) and \emph{enumerative} (operating on individual vertices).
\item[--] \emph{Emerson-Lei games.}
We extend this approach by introducing universal-tree-based rankings and a rank-lifting algorithm for the more general class of Emerson-Lei games.
This extension is technically more involved and relies on Zielonka trees
to guide the construction of rankings. 
The resulting algorithm again is asymmetric and enumerative,
and improves the upper bound on time complexity
for the computation of winning regions and strategies in Emerson-Lei games with $n$ vertices, $m$ edges, and $c$ colors from 
$O(mc!) (nc!)^{1+o(1)}$ to $O\left(c\log(c)mc!\right) n^{\log(c/2) + O(1)}$.
Importantly, extracted strategies use optimal memory. 
As far as we are aware, this is the first definition of rankings and a rank lifting algorithm for Emerson-Lei games.
\item[--] \emph{Zielonka-McNaughton algorithm.}
Finally, we show how universal trees can be used to bound the recursion tree of the
Zielonka-McNaughton algorithm for Emerson-Lei games. 
Using quasi-polynomial universal trees yields
a decision algorithm with time complexity
$O\left(mcc!\right) n^{2\log(c/2)+O(1)},$
improving over existing approaches.
This algorithm is \emph{symmetric} (computing winning regions of both players simultaneously) and \emph{symbolic} (operating on sets of vertices),
and it generalizes previous quasi-polynomial Zielonka-style algorithms for parity games.
As special cases, it also improves the symbolic computation for Streett and Rabin conditions. 
\end{compactitem}
Thus our algorithms yield factorial improvements in the time complexity for solving Streett (resp.~Emerson-Lei) games, breaking the quadratic factorial $(k!)^2$ (resp.~$(c!)^2$) dependence on the number of pairs (resp.~colors) in previous algorithms.

On a conceptual level, our results show how universal trees and Zielonka trees can be combined to analyze
Streett and Emerson-Lei games. Thus we lift the quasi-polynomial methods (and the associated improved 
worst case runtime guarantees) from parity and Rabin games
to the more general setting of Streett and Emerson-Lei games. Our results also provide additional insight into how (rankings for) these more involved
conditions are composed by a careful interleaving of (rankings of) individual memory-free conditions; 
Zielonka trees serve as precise recipes for this interleaving.

We note that our results apply also to Muller games, where the size of the Zielonka tree is proportional to the explicit representation of the winning condition. We leave the computation of exact bounds in this case as future work. 

\emph{Related work:} 
Progress measure and rank-lifting algorithms have been developed for various types of games, including
parity~\cite{DBLP:conf/stacs/Jurdzinski00}, Rabin~\cite{DBLP:phd/us/Klarlund90}, and Streett~\cite{DBLP:conf/lics/PitermanP06} games.
For parity and Rabin games, it has been shown that progress can be measured by leaves in (colored) universal trees~\cite{DBLP:conf/lics/JurdzinskiL17, DBLP:conf/tacas/MajumdarST24}.
For parity games, quasi-polynomial Zielonka-McNaughton-style algorithms have been proposed, bounding the recursion using
universal trees~\cite{DBLP:conf/mfcs/Parys19,DBLP:journals/lmcs/LehtinenPSW22}.
For Rabin, Streett and Muller games, Liang et al.~\cite{DBLP:conf/mfcs/LiangK025} trade
the factorial dependence on the number of pairs (or colors) for an exponential dependence on arena size,
\emph{deciding} Streett and colored Muller games in time $O(n(kn+m2^n))$ and $O(cm2^n)$, respectively.
In the extreme case $k = \Omega(n)$ (resp.~$c (\log c)^2 = \Omega(n/\log n)$),
this is asymptotically better than our upper bounds for Streett (resp.~Emerson-Lei) games.

Universal graphs were introduced as a unifying framework for games with qualitative and quantitative objectives, including mean-payoff games~\cite{DBLP:conf/mfcs/FijalkowGO20}. This line of work was subsequently extended to study memory requirements for more general objectives~\cite{DBLP:journals/lmcs/ColcombetFGO22,DBLP:journals/theoretics/Ohlmann23,DBLP:conf/icalp/CasaresO23}.

Universal trees and graphs have also been applied to compute nested fixpoints~\cite{DBLP:conf/tacas/HausmannS21, DBLP:conf/csl/ArnoldNP21}, extending 
quasi-polynomial techniques beyond classic parity games to settings with richer branching structures, covering quantitative aspects.

Calude et al. have also shown that assuming the Exponential Time Hypothesis, Emerson-Lei games cannot be solved in time 
% $O(c^c.n^{o(1)})$ 
$c^c \cdot n^{O(1)}$~\cite{DBLP:journals/siamcomp/CaludeJKLS22}. 
Our results are consistent with this lower bound as the size of the universal tree is poly-logarithmic in~$n$.

\emph{Structure:} Section~\ref{sec:prelim} recalls notions relating to games, universal trees, and Zielonka trees. 
We then focus on Streett games in Sections~\ref{sec:zielonka mcnaughton} and~\ref{sec:ranking},
introducing universal-tree-based rankings and a rank-lifting solution algorithm.
Building on this, universal-tree-based rankings and a rank-lifting algorithm for Emerson-Lei games are introduced in Sections~\ref{sec:el} and~\ref{sec:elranking}.
Finally, Section~\ref{sec:symbolic} shows how universal trees can be used to bound the recursion tree of the Zielonka-McNaughton algorithm for Emerson-Lei games.

\section{Preliminaries}
\label{sec:prelim}

\subsection{Arenas, Games, and Winning Conditions}
An \emph{arena} is $A=\pair{V,V_0,V_1,E}$, where $V=V_0\uplus V_1$ is a set of vertices partitioned into $V_0$, the set of vertices owned by player~0, and $V_1$, the set of vertices owned by player~1.
The set $E\subseteq V\times V$ is a \emph{total} set of edges, that is, 
for all $v\in V$, there is some $v'\in V$ such that $(v,v')\in E$. 
Let $E(v)=\{v' ~|~ (v,v') \in E\}$ and $E^{-1}(v)=\{v' ~|~ (v',v) \in E\}$.
A \emph{play} in $A$ is an infinite sequence $\rho=v_1,v_2,\ldots$ such that for all $i\geq 1$ we have $v_{i+1}\in E(v_i)$.
A play starts in $v$ if $v_1=v$.
The set $\inft(\rho)$ is the set of vertices appearing infinitely often in $\rho$.

A \emph{game} is $G=\pair{A,\alpha}$, where $A$ is an arena and $\alpha \subseteq V^\omega$ is a \emph{winning condition}.
The condition $\alpha$ is an \emph{Emerson-Lei} (EL) condition if there is a set $C$ of colors, a map $\gamma:V\rightarrow 2^C$, and a condition $\beta \in \mathbb{B}(GF(\mathbb{B}(C)))$, where $\mathbb{B}(T)$ is the Boolean closure of $T$, such that $w\in \alpha$ iff $\gamma(w)\models \beta$.  
Here, $\gamma(w) \models GF\, \delta$ if there are infinitely many $i$ such that $\gamma(w_i)\models \delta$, for $c\in C$ we have $\gamma(w_i)\models c$ if $c\in \gamma(v_i)$ and Boolean connectives are interpreted as expected.
Notice that we allow vertices to be colored by multiple colors and arbitrary Boolean combinations of colors nested in $GF$ and $FG$. This increases succinctness of EL conditions, is a better match to their symbolic nature (and usage for synthesis), and does not change the algorithms involved. 
The condition $\alpha$ is a \emph{Streett} condition if $C=\{g_1,r_1,\ldots, g_k,r_k\}$ and $\alpha = \bigwedge_{j} (GF\,r_j\to GF\,g_j)$.
We define $G_j=\gamma^{-1}(g_j)$ and $R_j=\gamma^{-1}(r_j)$ and denote the Streett condition by $\{\pair{R_1,G_1},\ldots, \pair{R_k,G_k}\}$.
\emph{Rabin} conditions are also defined over $k$ pairs of colors and \emph{parity} conditions are defined over $k$ colors. Their definition is standard and omitted. 

Given a Streett or Rabin condition over $k$ pairs or a parity condition over $k$ colors, we define $|\alpha|=k$.
We denote by $\el nmc$ the class of Emerson-Lei games whose arenas have $n$ vertices, $m$ edges, and whose winning conditions have $c$ colors. 
Similarly, we denote by \streett nmk, \rabin nmk, and \parity nmk the classes of Streett, Rabin, and parity games, respectively, whose arenas have $n$ vertices, $m$ edges, and whose winning conditions have size $k$.

A \emph{strategy} for player~0 is a function $\sigma:V^* \cdot V_0\rightarrow V$ such that for every sequence $wv\in V^*\cdot V_0$ we have $\sigma(wv)\in E(v)$.

We define the amount of memory that a strategy uses as the minimal size of a transition
system that encodes the memory updates in the strategy according to the colors that are visited in plays
(see, e.g., ~\cite{DBLP:conf/lics/DziembowskiJW97} for details). 
A play $\rho$ is \emph{compatible} with $\sigma$ if for every prefix $v_1,\ldots, v_n$ of $\rho$ such that $v_n\in V_0$ we have $v_{n+1} = \sigma(v_1,\ldots, v_n)$.
A strategy is \emph{winning} if every play compatible with $\sigma$ satisfies the winning condition.
Player~0 wins from vertex $v$ if there exists a strategy $\sigma$ such that every play starting in $v$ and compatible with $\sigma$ is winning.
Given a game $G$, we denote by $\Win_0(G)$ the set of vertices winning for player~0 in $G$.
The notions of strategies and winning for player~1 are defined dually. 
All games considered here are \emph{determined}, that is, we always have $\Win_0(G)\cap \Win_1(G)=\emptyset$ and $\Win_0(G)\cup \Win_1(G)=V$.
\emph{Deciding} a game is the computation of $\Win_0(G)$ or $\Win_1(G)$.
\emph{Solving} a game computes also a winning strategy for player~0.

\begin{example}\label{ex:el}
Let $C=\{a,b,c,d\}$ and consider the condition
$\alpha_1 = (GF\, a \to GF\,  b) \land (GF\, a \to GF\,  c)$, encoding a Streett condition with two pairs $\pair{R_1,G_1}$ and $\pair{R_2,G_2}$
(with $R_1=R_2=\gamma^{-1}(a)$, $G_1=\gamma^{-1}(b)$ and $G_2=\gamma^{-1}(c)$), and
the Emerson-Lei condition $\alpha_2 = \alpha_1 \lor (FG\, \neg d \,\land\, GF\, b)$
which is the disjunction $\alpha_1$ with a Rabin
condition consisting of the single pair $\pair{F_1,E_1}$ (with $F_1=\gamma^{-1}(d)$ and $E_1=\gamma^{-1}(b)$).
Then consider the following game arena (denoting player $0$ nodes by circles,
and player $1$ nodes by boxes).
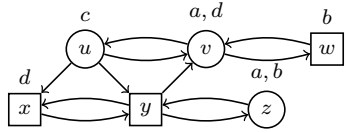
\begin{center}
\tikzset{every state/.style={minimum size=12pt}}
\begin{footnotesize}
  \begin{tikzpicture}[
    % Default arrow tip
    %-&gt;,&gt;=stealth',shorten &gt;=1pt,
		auto,
    % Default node distance
    node distance=0.8cm,
    % Edge stroke thickness: semithick, thick, thin
    semithick
    ]
     \node[state, label={above:$a,d$}] (0) {$v$};
     \node (yo1) [left of=0] {};
     \node (yo3) [right of=0] {};
     \node[state, label={above: $c$}] (1) [left of=yo1] {$u$};
     \node (yo2) [below of=1] {};
     \node[state, rectangle] (3) [right of=yo2] {$y$};
     \node[state, rectangle, label={above: $d$}] (4) [left of=yo2] {$x$};
     \node (yo) [right of=3] {};
     \node[state, label={above: $a,b$}] (2) [right of=yo]
     {$z$};
     \node[state, rectangle, label={above: $b$}] (5) [right of=yo3] {$w$};
     \path[->] (0) edge [bend right=15] node [pos=0.3,left] {} (1);
     \path[->] (1) edge [bend right=15] node [pos=0.3,left] {} (0);
     \path[->] (2) edge [bend right=15] node [pos=0.3,right] {} (3);
     \path[->] (3) edge [bend right=15] node [pos=0.3,right] {} (2);
     \path[->] (0) edge [bend right=15] node [pos=0.3,right] {} (5);
     \path[->] (5) edge [bend right=15] node [pos=0.3,right] {} (0);
   
     \path[->] (4) edge [bend right=15] node [pos=0.3,left] {} (3);
     \path[->] (1) edge node [pos=0.3,left] {} (4);
     \path[->] (3) edge [bend right=15] node [pos=0.3,left] {} (4);
     \path[->] (3) edge node [pos=0.3,left] {} (0);
     \path[->] (1) edge node [pos=0.3,left] {} (3);
     
  \end{tikzpicture}
\end{footnotesize}
\end{center}

\noindent For condition $\alpha_1$, player $0$ wins the top row
using the strategy $\sigma_1$ that always moves from $u$ to $v$, and alternatingly moves from $v$ to $w$
and from $v$ to $u$; this strategy uses \emph{memory}, that is, it is not positional.
Any play starting in the top row and compatible with $\sigma_1$ visits all colors infinitely often and hence satisfies $\alpha_1$. 
Player $1$ wins the bottom row using the (positional) strategy that always moves from $y$ to $z$, which ensures that
only the colors $a,b$ are visited infinitely often.

Regarding $\alpha_2$, player $0$ again wins the top row using strategy $\sigma_1$.
Player $1$ again wins the bottom row, but this time has to use a strategy with memory:
any play visiting just the color $d$ or just the colors $a,b$ infinitely often satisfies $\alpha_2$.
Let $\sigma_2$ be the strategy that alternates between moving from
$y$ to $x$ and from $y$ to $z$. Any play that starts in the bottom row and is compatible with this strategy
does not satisfy $\alpha_2$ since it
visits the colors $a$ and $d$ infinitely often, but does not visit colors $b$ or $c$ infinitely often.
Notice that for Emerson-Lei conditions (such as $\alpha_2$), both players may require memory to win.
\end{example}

We formally define universal trees and Zielonka trees later. 
For now, we cite and state the following results:

\begin{theoremrep}
Let $u(n,k)$ be the size of the smallest universal tree of depth $\lceil k\rceil$ for $n$ leaves.
Let $z_\alpha$ denote the number of leaves of the Zielonka tree of a winning condition $\alpha$.
\begin{enumerate}
    \item
    We have $u(n,k) \leq {{\lfloor \log(n)\rfloor+k-1}\choose{k-1}}$, which is $O(n^{\log(k)})$ \cite{DBLP:conf/soda/CzerwinskiDFJLP19}. If $k$ is $o(\log(n))$, then $u(n,k)$ is $O(n^{1+o(1)})$ \cite{DBLP:conf/lics/JurdzinskiL17}.
    \item
    \parity nm{2k} games can be solved in time $O(m \cdot u(n,k))$ \cite{DBLP:conf/stoc/CaludeJKL017,DBLP:conf/lics/JurdzinskiL17}
    and
	\rabin nmk games can be solved in time $O(mk!\cdot u(n,k))$ \cite{DBLP:conf/tacas/MajumdarST24}.
    \item
    Deciding Streett games is co-\textsc{NP}-complete \cite{DBLP:conf/focs/EmersonJ88,DBLP:journals/siamcomp/EmersonJ99}.
    \item 
    \streett nmk games can be decided in time $O(mk!\cdot u(n,k))$ \cite{DBLP:conf/tacas/MajumdarST24}
and solved in time $O(mk!\cdot u(nk!,k))$ \cite{DBLP:conf/stoc/Safra92,DBLP:conf/lics/JurdzinskiL17}, which is $O(mk!(nk!)^{1+o(1)})$.
    \item
    Deciding EL games is \textsc{PSpace}-complete \cite{DBLP:conf/mfcs/HunterD05}.
    \item 
    For EL condition $\alpha$ over $c$ colors, we have $z_\alpha \leq c!$ \cite{DBLP:conf/lics/DziembowskiJW97}. 
    \item
    \el nmc games can be solved in time $O(mz_\alpha\cdot u(nz_\alpha,c/2))$, which is $O(m z_\alpha(nz_\alpha)^{1+o(1)})$ when $z_\alpha$ is exponential in $c$ \cite{DBLP:conf/stoc/GurevichH82,DBLP:conf/lics/JurdzinskiL17}.
    The resulting strategy uses memory size at most $z_\alpha$ \cite{DBLP:journals/theoretics/CasaresCFL24}.
 \end{enumerate}
 \label{theorem:existing results}
\end{theoremrep}
We include further explanations for (4) and (7) below. 

\begin{proof}
    For (4), a \streett nmk game can be decided by deciding the dual Rabin game with (3) and
    solved %For (6), a \streett nmk game can be solved 
    by constructing a \parity{nk!}{mk!}{2k} game using the index appearance record \cite{DBLP:conf/stoc/Safra92} and applying (2).
    For (7), an \el nmc game can be solved by constructing an equivalent \parity {nc!}{mc!}{c} game using a variant of the later appearance record \cite{DBLP:conf/stoc/GurevichH82} and applying (2). 
    Using the leaves of the Zielonka tree as in \cite{DBLP:journals/theoretics/CasaresCFL24} is always more efficient and replaces $c!$ by $z_\alpha$. 
    Another alternative for (7) is to convert an \el nmc game to a \rabin {nm^+_\alpha}{mm^+_\alpha}{\lceil c/2\rceil} game, where $m^+_\alpha$ is the worst case memory requirement for the condition $\alpha$ and is always bounded by $z_\alpha$. 
    This is done by using a history-deterministic Rabin automaton for the winning condition \cite{DBLP:journals/theoretics/CasaresCFL24}. Using the approach in (2) to solve the Rabin game would yield the time complexity $O(mm^+_\alpha(\lceil c/2\rceil)!u(nm^+_\alpha,c/2))$. 
\end{proof}

Our main technical result is replacing $u(nk!,k)$ by $u(n,k)$ in (4) and replacing $u(nz_\alpha,c/2)$ by $u(n,c/2)$ in (7),
thereby improving the solution time of both Streett and EL games (by a factor of $k!$ and $c!$).%
\footnote{Even if we consider a combination of \cite{DBLP:journals/theoretics/CasaresCFL24,DBLP:conf/tacas/MajumdarST24} a non-linear factor of $m^+_\alpha$ is removed, recalling that $m^+_\alpha$ could be $(c/2)!$.}
Furthermore, the memory used by the strategy that we produce is optimal rather than $z_\alpha$ as in the reduction to parity games as in (7).

\subsection{Trees, Tree Embeddings, and Universal Trees}

A tree $T$ is a prefix-closed subset of $\mathbb{N}^*$.
That is, if $wv\in T$ for $w\in \mathbb{N}^*$ and $v\in \mathbb{N}$, then $w\in T$ as well.
We call $\epsilon$ the root of the tree, $wv$ is a child of $w$ and given $w'\in \mathbb{N}^+$, $ww'$ is a descendant of $w$.
An element $w\in T$ is a node and if there is no $v$ such that $wv\in T$ then $w$ is a leaf. 
Given a tree $T$, let $\leaf(T)$ denote the set of leaves of $T$.
An $A$-labeled tree is $\pair{T,\tau}$, where $T$ is a tree and $\tau:T\rightarrow A$.
Notice that trees have \emph{nodes} and games have \emph{vertices}.

We order the nodes in trees and whole trees using lexicographic order as follows. 
Let $t=t_1\cdots t_k$ and $s=s_1\cdots s_k$ be two nodes.
We write $t<s$ if there is a $j$ such that for every $i\leq j$ we have $t_i=t_i$ and $t_{j+1}<s_{j+1}$ or $t_{j+1}$ does not exist (i.e., $t$ is an ancestor of $s$).
We have $t\leq s$ if $t<s$ or $t=s$.
We use prefix comparisons between nodes using the notation $<_j$ and $\leq_j$ as follows.
We say $t<_j s$ if $t_1\cdots t_j < s_1 \cdots  s_j$ and similarly $t \leq_j s$.
Notice that for all nodes, we have $t\leq_0 s$ and $s \leq_0 t$.
Indeed $\leq_0$ compares the 0 length prefix of both, which is $\epsilon$. 
Given trees $T_1, T_2$ we write $T_1<T_2$ if for all $l_1\in T_1$, $l_2\in T_2$ we have $l_1<_1 l_2$. 
We do not use such comparisons if $|T_1|<2$ or $|T_2|<2$. 

A \emph{depth-$k$ tree} ($k$-tree) is a tree $T$ such that $\mathit{leaf}(T)\subseteq \mathbb{N}^k$.
An \emph{$n$-leaf depth-$k$ tree} ($(n,k)$-tree) is a $k$-tree with at most $n$ leaves.
Note that the only $0$-tree is a root, i.e., a $(1,0)$-tree. 
We freely use the notation $k$-tree for non-integer $k$ to mean $\lceil k\rceil$-tree.
Given two $k$-trees $T_1$ and $T_2$, $T_1$ \emph{embeds} into $T_2$ if there is an order respecting embedding $e:\leaf(T_1) \rightarrow \leaf(T_2)$.
That is, for every $l_1,l_2\in \leaf(T_1)$ and every index $i\leq k$ we have $l_1<_il_2$ iff $e(l_1)<_ie(l_2)$.

\begin{definition}
  A $k$-tree is $(n,k)$-universal if it embeds every $(n,k)$-tree.
\end{definition}

It is well known that %One can show that the set 
$\{1,\ldots, n\}^{\leq k}$ is an $(n,k)$-universal tree; it is of exponential size.
Following the seminal result of \cite{DBLP:conf/stoc/CaludeJKL017}, Jurdzinski and Lazic \cite{DBLP:conf/lics/JurdzinskiL17} constructed a quasi-polynomial universal tree \cite{DBLP:conf/soda/CzerwinskiDFJLP19}.

\begin{theorem}[{\rm \hspace{-0.1pt}\cite{DBLP:conf/lics/JurdzinskiL17}}]
For every $n$ and $k$, there exists an $(n,k)$-universal tree $U(n,k)$ such
that $|U(n,k)|$ and $|\leaf(U(n,k))|$ are $O(n^{\log(k)})$. 
\end{theorem}

For a tree $T$ we define $T{\Downarrow_1} = \set{n_1 ~|~ n_1\cdots n_k \in T}$ and $|T|_1=|T{\Downarrow_1}|$.
If $|T_1|_1=1$ we put $T{\Uparrow_1}= \{n_2\cdots n_k ~|~ n_1\cdot n_2\cdots n_k \in T\}$. 

Given a tree $U$, let $U_1,\ldots U_r \subseteq U$ be such that $\bigcup_{i}U_i=U$, and for all $i< j$ we have that $U_i$ is a tree, $|U_i|_1=1$, $|U_i \cup U_j|_1=2$, and $U_i<U_j$.
That is, if $n_1,\ldots, n_r$ are the children of the root, and $S_i$ is the subtree of $U$ rooted at $n_i$, then $U_i=\{\epsilon\} \cup S_i$. 

We base our technical development on the following central concept of
hierarchical trees.

\begin{definition}
  An $k$-tree $U$ is \emph{$(n,k)$-hierarchical}, if $k=1$ and it has at least $n$ leaves or, for $k>1$, 
  if for every partition $n_1$, $\ldots$, $n_r$ of $n$, i.e., $\sum_i n_i=n$, there exist $i_1,\ldots, i_r$ such that $U_{i_j}\Downarrow_1$ is $(n_j,k-1)$-hierarchical and $U_{i_j}<U_{i_{j'}}$ for $j<j'$.
\end{definition}

The hierarchical property above expresses the existence of embeddings in a recursive way. Thus, it can be seen as an alternative, slightly more strict formalization of universality:

\begin{lemma}\label{lemma:hierarchicalisuniversal}
    Every $(n,k)$-hierarchical tree is $(n,k)$-universal.
\end{lemma}
\begin{proof}
Let $U$ be an $(n,k)$-hierarchical tree. We have to show that every $(n,k)$-tree embeds into $U$. 
Let $T$ be an $(n,k)$-tree. 
Define the required embedding recursively as follows. 
Map the root of $T$ to the root of $U$. 
If $k=1$ and $U$ has at least $n$ leaves, then we are done. 
Otherwise, let $t_1, .., t_r$ be the children of the root of $T$. 
For each $i$, let $n_i$ denote the number of leaves below $t_i$. Then the sum of all $n_i$ is $n$ (since $T$ has $n$ leaves). 
Since $U$ is $(n,k)$-hierarchical, there are children $u_1, .., u_r$ of the root of $U$ (ordered in the same way as the children $t_1, .., t_r$ are ordered in $T$) such that for all $i$, $u_i$ is the root of an $(n_i,k-1)$-hierarchical tree. 
For all $i$, continue the recursive construction of the embedding with (the subtrees rooted at) $t_i$ and $u_i$.
\end{proof}

We recall the construction of 
the universal tree $U(n,k)$ of Jurdzinski and Lazic, and show that this tree is hierarchical. 
Simple trees for $n=1$ or $k=1$ are constructed directly.
For larger trees, we take inductively two copies of $U(\lfloor n/2\rfloor,k)$ and one copy of $U(n,k-1)$.
We attach them at the root putting the first copy of $U(\lfloor n/2\rfloor,k)$ on the left, adding to the right of it a child of the root under which we put $U(n,k-1)$, and finally putting the second copy of $U(\lfloor n/2\rfloor,k)$ to the right of this new child. Formally, we have the following. 

Given a node $n=n_1\cdots n_k$ we denote by $n+i$ the node $(n_1+i)\cdot n_2\cdots n_k$.
That is, increment the first value in $n$ by $i$. 
Given a tree $T$ we denote by $T+i$ the tree $\{n+i ~|~ n\in T\}$.
That is, we increment the first value in all nodes in $T$ by $i$.
We denote by $i\cdot n$ the node $i\cdot n_1\cdots n_k$ and by $i\cdot T$ the tree $\{i\cdot n ~|~ n\in T\}$.
Notice that if $T$ is a $k$-tree then $i\cdot T$ is a $(k+1)$-tree.

Let $U(1,k)={1^k}$ and $U(n,1)=\{1,\ldots, n\}$.
For $n>1$ and $k>1$, the $(n,k)$-universal tree $U(n,k)$ is obtained from combining the trees $U(\lfloor n/2 \rfloor,k)$ and $U(n,k-1)$ as follows.
Put $U_l = U(\lfloor n/2\rfloor, k)$, $U_m = (|U_l|_1+1)\cdot U(n,k-1)$, and $U_r = U(\lfloor n/2 \rfloor, k)+|U_l|_1+1$.
Finally, $U(n,k)=U_l \cup U_m \cup U_r$. 
Clearly, $U_l<U_m<U_r$.%, where $T_1<T_2$ if for all $l_1\in T_1$, $l_2\in T_2$ and $i\leq k$ we have $l_1<_i l_2$. 
    
\begin{lemma}\label{lemma:hierarchical}
    The tree $U(n,k)$ is $(n,k)$-hierarchical.
\end{lemma}

\begin{proof}
    The proof is by induction on $(k,n)$, ordered lexicographically.
    If $k=1$, then there is nothing to show. % since all $(n,1)$-trees are $(n,1)$-hierarchical. 
    For the inductive step, assume that $U(n',k-1)$ is $(n',k-1)$-hierarchical for all $n'\leq n$.
    Consider a partition $n_1,\ldots, n_r$ of $n$. Let $j$ be the maximal index such that $\sum_{i<j}n_i\leq \lfloor n/2\rfloor$.
    By assumption, $U_l$ and $U_r$ are $(\lfloor n/2 \rfloor, k)$-hierarchical trees; this hierarchical
    property of $U_l$ and $U_r$ yields, for each $i\neq j$, a hierarchical tree $U_{i}$ such that $U_{i}$ is $(n_i,k-1)$-hierarchical.
    For $U_{j}$, we use $U(n,k-1)$ which is $(n,k-1)$-hierarchical (hence $(n_j,k-1)$-hierarchical) by assumption.
    By construction, all the subtrees identified under $U_r$ are to the left of the subtrees identified under $U_m$ and $U_l$ and those under $U_m$ is to the left of those under $U_r$. 
\end{proof}

We note that other universal trees, namely $\{1,\ldots, n\}^{\leq k}$ and Parys' universal tree \cite{DBLP:conf/mfcs/Parys19}, are also hierarchical.
This follows directly from their recursive definition.
One simple universal tree that is not necessarily hierarchical (but very large) is the tree obtained by taking all $(n,k)$-trees and merging them at the root.
The resulting tree is clearly universal but not necessarily hierarchical. 
We conjecture that every ``reasonably sized'' universal tree is hierarchical. 

We also note that universal trees as presented above, are slightly different from the definition of Jurdzinski and Lazic~\cite{DBLP:conf/lics/JurdzinskiL17} in that the tree $U(n,1)$ in their definition is obtained by a recursive decomposition from $U(1,0)$, which is the tree with a single node, namely, the root. 
In the worst case, where $n=2^r$, our $U(2^r,1)$ above has $2^r$ leaves while the tree in~\cite{DBLP:conf/lics/JurdzinskiL17} has $2^{r+1}-1$ leaves. 
This does not change the estimate on the total size of the tree. 

\subsection{Zielonka Trees of Emerson-Lei conditions}

Let $\alpha$ be an EL condition over set $C$ of colors.
We say that $C'\subseteq C$ satisfies $\alpha$ if for every play $\rho$ such that $\inft(\rho)=C'$, we have $\rho\in \alpha$, and denote this situation by $C'\models \alpha$.
We say that $C'\subseteq C$ does not satisfy $\alpha$ if for every play $\rho$ such that $\inft(\rho)=C'$, we have $\rho\not\in \alpha$, denoted $C'\not\models \alpha$.
For a set $C'$ such that $C'\models \alpha$, define 
$$\textstyle \max_1(C')=\left \{
C'' \left | 
\begin{array}{l}
C'' \subsetneq C', C''\not\models \alpha,\mbox{ and}\\
\forall D~.~ C''\subsetneq D \subsetneq C' \mbox{ implies } D \models \alpha
\end{array}
\right . \right \}.$$

For a set $C'$ such that $C'\not\models\alpha$, define
$$
\textstyle \max_0(C')=\left \{ C'' 
\left | 
\begin{array}{l}
C'' \subsetneq C', C''\models \alpha, \mbox{ and}\\
\forall D~.~ C'' \subsetneq D\subsetneq C' \mbox{ implies } D \not\models \alpha
\end{array}
\right .
\right \}.$$

Given an EL condition we define its \emph{Zielonka tree} \cite{DBLP:conf/lics/DziembowskiJW97} to be the minimal labeled tree $Z(\alpha)=\pair{T,\tau}$, where $\tau:T\rightarrow 2^C$ satisfies $\tau(\epsilon)=C$ and for every node $n\in T$:
\begin{compactitem}
    \item[--] if $\tau(n)\models \alpha$ and $\max_1(\tau(n))=\{C_1,\ldots, C_m\}$ then $n\cdot 1,\ldots, n\cdot m \in T$ and $\tau(n\cdot i)=C_i$;
    \item[--] if $\tau(n)\not\models \alpha$ and $\max_0(\tau(n))=\{C_1,\ldots, C_m\}$ then $n\cdot 1,\ldots, n\cdot m \in T$ and $\tau(n\cdot i)=C_i$.
\end{compactitem} 
For a node $t\in T$ we denote by $\child(t)$ the set of children of $t$ and by $\grandchild(t)$ the set of grandchildren of $t$, 
that is, $\grandchild(t)=\bigcup_{t'\in \child(t)}\child(t')$.
Notice that if $n\cdot(i+1)$ is a node in $T$, then $n\cdot i$ is also a node in $T$. We denote by $z_\alpha$ the number of leaves of $Z(\alpha)$, $d_\alpha$ for its maximal depth, and
refer to nodes $t$ such that $\tau(t)\models\alpha$ as \emph{winning}, and to the remaining nodes as \emph{losing}.
We note that the size of $Z(\alpha)$, its number of leaves ($z_\alpha$), and its depth ($d_\alpha$) are independent of the order chosen for children. 
Similarly, correctness and upper bounds for algorithms computations do not depend on this order. 

In the special case that $\alpha$ is a Streett condition, $Z(\alpha)$ does not branch at losing nodes:
assuming that $G_i$ is visited finitely often, player 0 can only win if they eventually avoid $R_i$ forever.

\begin{example}\label{ex:ziel}
The Zielonka trees for the objectives $\alpha_1$ and $\alpha_2$ from Example~\ref{ex:el} are as follows.
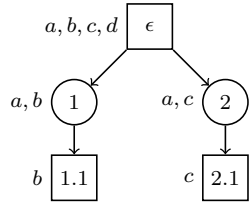
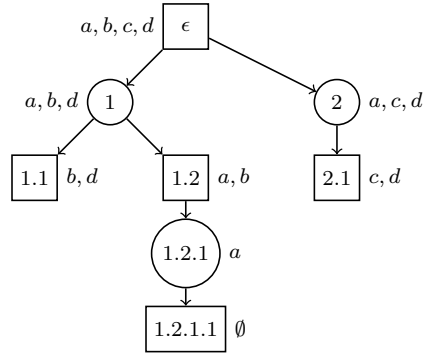
\begin{center}
\tikzset{every state/.style={minimum size=17pt}}
\begin{footnotesize}
 \begin{tikzpicture}[
    % Default arrow tip
    %-&gt;,&gt;=stealth',shorten &gt;=1pt,
		auto,
    % Default node distance
    node distance=1cm,
    % Edge stroke thickness: semithick, thick, thin
    semithick
    ]
     \node[state, rectangle, label={left: $a,b,c,d$}] (0) {$\epsilon$};
     \node (yoy) [left of=0] {};
     \node (yoi) [above of=yoy] {};
     \node (yoz) [left of=yoi] {\small{$Z(\alpha_1):$}};
     \node (yo) [below of=0] {};
     \node[state, label={left: $a,b$}] (1) [left of=yo] {$1$};
     \node[state, label={left: $a,c$}] (2) [right of=yo] {$2$};
     \node[state, rectangle, label={left: $b$}] (3) [below of=1] {$1.1$};
     \node[state, rectangle, label={left: $c$}] (4) [below of=2] {$2.1$};
     \node (yo1) [below of=3] {};
     \node (yo2) [below of=yo1] {};

     \path[->] (0) edge node [pos=0.3,left] {} (1);
     \path[->] (0) edge node [pos=0.3,left] {} (2);
     \path[->] (1) edge node [pos=0.3,right] {} (3);
     \path[->] (2) edge node [pos=0.3,right] {} (4);

     \node (yoyb) [right of=0] {};
     \node (yoyb1) [right of=yoyb] {};
     \node (yoyb2) [right of=yoyb1] {};
     \node (yoyb3) [right of=yoyb2] {};
     \node (yozar) [above of=yoyb2] {\small{$Z(\alpha_2):$}};
     \node[state, rectangle, label={left: $a,b,c,d$}] (0a) [right of=yoyb3] {$\epsilon$};
     \node (yoa) [below of=0a] {};
     \node (yo2a) [right of=yoa] {};
     \node[state, label={left: $a,b,d$}] (1a) [left of=yoa] {$1$};
     \node[state, label={right: $a,c,d$}] (2a) [right of=yo2a]
     {$2$};
     \node (yo1a) [below of=1a] {};
     \node[state, rectangle, label={right: $a,b$}] (3a) [right of=yo1a] {$1.2$};
     \node[state, rectangle, label={right: $b,d$}] (4a) [left of=yo1a] {$1.1$};
     \node[state, rectangle, label={right: $c,d$}] (5a) [below of=2a] {$2.1$};
     \node[state, label={right: $a$}] (6a) [below of=3a] {$1.2.1$};
     \node[state, rectangle, label={right: $\emptyset$}] (7a) [below of=6a] {$1.2.1.1$};
     \path[->] (0a) edge node [pos=0.3,left] {} (1a);
     \path[->] (0a) edge node [pos=0.3,right] {} (2a);
     \path[->] (1a) edge node [pos=0.3,left] {} (3a);
     \path[->] (1a) edge node [pos=0.3,left] {} (4a);
     \path[->] (2a) edge node [pos=0.3,left] {} (5a);
     \path[->] (3a) edge node [pos=0.3,left] {} (6a);
     \path[->] (6a) edge node [pos=0.3,left] {} (7a);
    
  \end{tikzpicture}
\end{footnotesize}
\end{center}
Nodes that satisfy the respective objective are denoted by boxes and
the remaining nodes by circles; the labels
are given next to nodes, e.g. in $Z(\alpha_2)$ we have $\tau(1)=\{a,b,d\}$ so that $\tau(1)\not\models \alpha_2$.
The strategy $\sigma_1$ for condition $\alpha_2$ from Example~\ref{ex:el} intuitively uses
the nodes $1$ and $2$ as memory values, encoding that the current objective
is to next visit color $c$ or $b$, respectively; by alternating between
both memory values, $\sigma_1$ ensures that
all colors $a,b,c,d$ (satisfying $\alpha_2$) are visited infinitely often.
Similarly, strategy $\sigma_2$ uses nodes $1.1$ and $1.2$ as memory values.
We point out that $Z(\alpha_2)$ has leaves at different depths.
Notice that we diverge from the normal representation of Zielonka trees in the literature (e.g., \cite{DBLP:journals/theoretics/CasaresCFL24}) in that we choose to denote nodes satisfying the objective by boxes and those not satisfying it by circles. 
The intuition for this choice is that 
for nodes satisfying the objective, player $0$ needs to win according to the sub-objectives encoded by \emph{all} subtrees of the node.
Dually, in a node not satisfying the objective player $0$ needs to win according the sub-objective encoded by \emph{some} subtree.
\end{example}

The Zielonka tree captures the maximal memory required to win a game with winning condition $\alpha$ and is intertwined with the solution of EL games as we discuss below. 
Regarding memory, we follow~\cite{DBLP:conf/lics/DziembowskiJW97} and match each node $t$ in the Zielonka tree $Z_\alpha=\pair{T,\tau}$ to the maximal memory for both players required to win a game with winning condition $\alpha$ colored only by the set of colors $\tau(t)$.
For a leaf $l$ we put $m^+_\alpha(l)=m^-_\alpha(l)=1$. 
For a winning node $t$ (such that $\tau(t)\models \alpha$) we put $m^+_\alpha(t)=\sum_{t' \in \child(t)} m^+_\alpha(t')$ and $m^-_\alpha(t) = \max_{t'\in \child(t)}(m^-_\alpha(t'))$.
For a losing node $t$ (such that $\tau(t)\not\models\alpha$) we put $m^+_\alpha(t)=\max_{t'\in \child(t)}(m^+_\alpha(t'))$ and $m^-_\alpha(t) = \sum_{t'\in\child(t)}m^-_\alpha(t)$.
We denote $m^+_\alpha(\epsilon)$ by $m^+_\alpha$ and similarly for $m^-_\alpha$.

\begin{theorem}[\rm \hspace{-0.1pt}\cite{DBLP:conf/lics/DziembowskiJW97}]
Given an Emerson-Lei game $G$ with objective $\alpha$, 
there is a strategy for player 0 that uses 
memory $m^+_\alpha$ and wins every game vertex in $\Win_0(G)$, and
there is a strategy for player 1 that uses
memory $m^-_\alpha$ and wins every game vertex in $\Win_1(G)$.
Furthermore, we have $\max\{m^+_\alpha,m^-_\alpha\}\leq z_\alpha$, 
and there is a series of EL games in which winning strategies
for player 0 (resp.~1) require at least memory $m^+_\alpha$ (resp. $m^-_\alpha$).
\end{theorem}

The numbers of nodes and leaves in Zielonka trees are at most factorial in the number of colors
of their objective,
while their height is directly bounded by the number of colors:

\begin{lemmarep}\label{lem:zielonkaSize}
    For an EL condition $\alpha$ over $C$ with Zielonka tree $Z(\alpha)=\pair{T,\tau}$, we have %the following
    $|T|\leq e|C|!$ \cite{DBLP:conf/fossacs/HausmannLP24}, $z_\alpha\leq |C|!$, and $d_\alpha\leq |C|$, where $e$ is Euler's number.
\end{lemmarep}

\begin{proof}
    It is simple to see that $z_\alpha \leq |C|!$ as every leaf corresponds to a 
    sequence of containments $C=C_0\subsetneq C_1 \subsetneq \cdots \subsetneq C_d'$.
    Every order putting colors in $C_i\setminus C_{i+1}$ before all colors in $C_{i+1}$ is consistent only with that leaf.
    Also $d_\alpha\leq |C|$ follows from the decrease in the sets along every edge in the tree. 
\end{proof}

We note that these are the worst case estimates that are guaranteed to suffice for all conditions that use $|C|$ colors. 
In many cases (as evidenced, for instance, in Example~\ref{ex:ziel}), the size of the Zielonka tree of a condition is much smaller than the worst case
(cf. Remark~\ref{rem:complexity}).
 
\section{Streett Games and Hierarchical Trees}
\label{sec:zielonka mcnaughton}

We start our technical development by recalling the Zielonka-McNaughton algorithm for solving Streett games. Based on this algorithm, we then analyze the structure of winning regions in Streett games.
We use this structural analysis to encode the winning regions using mappings to appropriate hierarchical trees. 

\subsubsection*{Zielonka-McNaughton Algorithm for Solving Streett Games}
We consider Algorithm~\ref{alg:Zielonka}, Horn's version of the Zielonka-McNaughton algorithm~\cite{DBLP:journals/apal/McNaughton93,DBLP:journals/tcs/Zielonka98} for solving Streett games \cite{DBLP:journals/ipl/Horn07}.
The algorithm works on a set of vertices $\subarena$ that constitute the sub-arena to work with and a set $\beta$ of pairs representing part of the Streett condition.
Initially, these parameters are the full set of vertices $V$ and the full winning condition $\alpha$.
Given a Streett game with $k$ pairs, the algorithm intuitively computes, for each pair $\pair{R_i,G_i}$, regions where player $0$ can
eventually force a visit to $G_i$ (line 3), or otherwise avoid visits to $R_i$ (line 6). Regions for the latter case contain neither
$R_i$ nor $G_i$, so they are solved recursively by solving a Streett game with $k-1$ pairs (line 7).
A winning strategy can be constructed from the results of the computation of the algorithm by keeping a record of intermediate computation results. 
Formally, let $\cpre_i^{W}(Z)$ compute the controllable (w.r.t.~player $i\in\{0,1\}$) predecessors of $Z\subseteq V$ within the sub-arena induced by $W\subseteq V$:
$$
\begin{array}{l c l}
\cpre_i^{W}(Z) & =& \{v\in V_i \cap W ~|~ E(v)\cap W \cap Z \neq \emptyset\} \,\cup\\ &&\{ v\in V_{1-i} \cap W ~|~ E(v) \cap W \subseteq Z \}.
\end{array}
$$
Building on $\cpre$, we use $\attr^{W}_i(Z)$ to denote the \emph{attractor} of player~$i$ to the set $Z$ of vertices in the sub-arena induced by the set $W\subseteq V$;
formally, we define $\attr^{W}_i(Z)$  to be the least fixpoint of the function that maps $X \subseteq W$ to $Z\cup \cpre_i^{W}(X)$.

\begin{algorithm}[bt]
\caption{\label{alg:Zielonka}$\\
\textsc{SolveStreett}(\subarena,\beta = \{\pair{R_1,G_1},\ldots, \pair{R_d,G_d}\})$}
\DontPrintSemicolon
\ForAll{$(i \in \{1,2,\ldots, d\})$}{
   $W_{{-}1}=\subarena$\;
   $W_0=\attr^{\subarena}_0(G_i)$\;
   $\mathsf{Residue} \gets \subarena \setminus W_0$\;
   \For{$(j\gets 0  ~;~ W_{2j-1}\neq \emptyset ~;~ j{+}{+})$}{
       $\mathsf{No}_{R_i} \gets \mathsf{Residue} \setminus \attr^{\mathsf{Residue}}_1(R_i)$ \;
       $W_{2j+1} \gets \mbox{\sc SolveStreett}(\mathsf{No}_{R_i},\beta \setminus \{\pair{R_i,G_i}\})$\;
       $W_{2j+2} \gets \attr^{\mathsf{Residue}}_0(W_{2j+1})$\;
       $\mathsf{Residue} \gets \mathsf{Residue} \setminus W_{2j+2}$\;
   }
   \If{$(\mathsf{Residue}\neq \emptyset)$}{
       \Return{$\mbox{\sc SolveStreett}(\subarena\setminus \attr^{\subarena}_1(\mathsf{Residue}),\beta)$}
    }
}
\Return{$\subarena$}
\end{algorithm}

\begin{theorem}[{\rm \hspace{-0.1pt}\cite{DBLP:journals/apal/McNaughton93,DBLP:journals/tcs/Zielonka98,DBLP:journals/ipl/Horn07}}]
 Algorithm~\ref{alg:Zielonka} solves Streett games.
\end{theorem}

When Algorithm~\ref{alg:Zielonka} terminates it returns through line~12. It then outputs the winning region for player $0$ in the input Streett game, which is also the sub-arena last analyzed.
We restrict attention to cases where the Algorithm returns though line~12.
In such a cases, the algorithm computes for the entire sub-arena and for each pair $i\in \{1,\ldots, d\}$, the sets $W_0,W_2,\ldots, W_{2l}$,
which form a partition of this sub-arena (and the winning region).
This is illustrated in Figure~\ref{fig:illustration}.
When a set $W_{2j}$ is the result of an attractor computation, namely, $W_{2j}=\attr^W_i(Z)$ for some $W$, $i$ and $Z$, then we assume that we also have access to the partition of $W_{2j}$ according to the distance from $Z$.
That is, $W_{2j}=W_{2j}^0\cup W_{2j}^1 \cup \cdots \cup W_{2j}^l$, where $W_{2j}^0=Z$ and 
$W_{2j}^{p+1}=\cpre_i^{W}(\textstyle\bigcup_{q\leq p} W_{2j}^{q})\setminus \textstyle\bigcup_{q\leq p} W_{2j}^{q}$.
Clearly, $W_0^0=G_i\cap \subarena$. 
The partition satisfies the following properties.

\begin{figure}[bt]

\begin{center}

\begin{tikzpicture}
  % Parameters
  \def\H{.5}   % height of each large rectangle
  \def\W{6.1}   % total width (4 + 1 + 1 + 1)

  % Partition x-coordinates
  \def\xA{4}
  \def\xB{4.7}
  \def\xC{5.4}

  % Draw 5 stacked rectangles with shared borders
  \foreach \i in {0,...,4} {

    % Vertical position
    \pgfmathsetmacro{\y}{-\i*\H}

    % --- Shading ---
    % Light gray shading for first wide cell of the top row only
    \ifnum\i=0
      \fill[gray!15] (0,\y) rectangle (\xA,\y+\H);
    \fi

    % Gray shading for the three narrow cells (all rows)
    \fill[gray!30] (\xA,\y) rectangle (\xB,\y+\H);
    \fill[gray!30] (\xB,\y) rectangle (\xC,\y+\H);
    \fill[gray!30] (\xC,\y) rectangle (\W,\y+\H);

    % --- Borders ---
    % Outer rectangle
    \draw (0,\y) rectangle (\W,\y+\H);

    % Internal partitions
    \draw (\xA,\y) -- (\xA,\y+\H);
    \draw (\xB,\y) -- (\xB,\y+\H);
    \draw (\xC,\y) -- (\xC,\y+\H);

  }
    \node at ({0.6},{0.23}) {\tiny{$G_{i}=W^0_0$}};
    \node at ({1.8},{0.23-0.5}) {\tiny{$W_{1}=W_2^0=\Win(\alpha\setminus\{G_i,R_i\})$}};
    \node at ({1.8},{0.23-1}) {\tiny{$W_{3}=W_4^0=\Win(\alpha\setminus\{G_i,R_i\})$}};
    \node at ({1.8},{0.23-1.5}) {\tiny{$\ldots$}};
    \node at ({2},{0.23-2}) {\tiny{$W_{2l-1}=W_{2l}^0=\Win(\alpha\setminus\{G_i,R_i\})$}};

    \node at ({4.35},{0.23}) {\tiny{$W^1_{0}$}};
    \node at ({4.35},{0.23-0.5}) {\tiny{$W^1_{2}$}};
    \node at ({4.35},{0.23-1}) {\tiny{$W^1_{4}$}};
    \node at ({4.35},{0.23-1.5}) {\tiny{$\ldots$}};
    \node at ({4.35},{0.23-2}) {\tiny{$W^1_{2l}$}};

    \node at ({5.05},{0.23}) {\tiny{$W^2_{0}$}};
    \node at ({5.05},{0.23-0.5}) {\tiny{$W^2_{2}$}};
    \node at ({5.05},{0.23-1}) {\tiny{$W^2_{4}$}};
    \node at ({5.05},{0.23-1.5}) {\tiny{$\ldots$}};
    \node at ({5.05},{0.23-2}) {\tiny{$W^2_{2l}$}};

    \node at ({5.75},{0.23}) {\tiny{$\ldots$}};
    \node at ({5.75},{0.23-0.5}) {\tiny{$\ldots$}};
    \node at ({5.75},{0.23-1}) {\tiny{$\ldots$}};
    \node at ({5.75},{0.23-1.5}) {\tiny{$\ldots$}};
    \node at ({5.75},{0.23-2}) {\tiny{$\ldots$}};

    \node[right] at (\W,0.25) {\}};
    \node[right] at (\W+0.17,0.251) {\tiny{$W_{0}=\attr_0(G_i)$}};
    \node[right] at (\W,0.25-0.5) {\}};
    \node[right] at (\W+0.17,0.25-0.52) {\tiny{$W_{2}=\attr_0(W_1)$}};
    \node[right] at (\W,0.25-1) {\}};
    \node[right] at (\W+0.17,0.25-1.03) {\tiny{$W_{4}=\attr_0(W_3)$}};
    \node[right] at (\W,0.25-2) {\}};
    \node[right] at (\W+0.17,0.25-2.05) {\tiny{$W_{2l}=\attr_0(W_{2l-1})$}};

\end{tikzpicture}

\end{center}
\vspace*{-10pt}
\caption{\label{fig:illustration}\textit{Winning region partition with respect to a single pair $\pair{R_i,G_i}$ as computed by Algorithm~\ref{alg:Zielonka}. The
recursively computed white winning regions do not contain the colors $R_i$ and $G_i$.}}
\vspace*{-15pt}
\end{figure}

\begin{lemmarep}\label{lem:partition_properties}
When Algorithm~\ref{alg:Zielonka} returns from line~12 for sub-arena $\subarena$ and $\beta$ such that $|\beta|=d$, then for all $i\in \{1,\ldots, d\}$, the partition $W_0,\ldots, W_{2l}$ satisfies:
\begin{itemize}
\item[--]
    for every $v\in W^0_0\cap V_0$ we have $v\in G_i$ and there exists $v'\in E(v)$ such that $v'\in \bigcup_{j\geq 0} W_{2j}$;
\item[--]
    for every $v\in W^0_0\cap V_1$ we have $v \in G_i$ and for every $v'\in E(v)\cap \subarena$ we have $v'\in\bigcup_{j\geq 0}W_{2j}$.
\item[--]
    for $p > 0$ and every $v\in W_{2j}^p\cap V_0$ there is $v'\in E(v)$ such that $v'\in W_{2j}^{p-1}$;
\item[--] 
    for $p > 0$ and every $v\in W_{2j}^p\cap V_1$ we have that every $v'\in E(v)\cap\subarena$ satisfies $v'\in \left (\bigcup_{j'<2j} W_{j'} \right )\cup \left (\bigcup_{p'<p} W_{2j}^{p'} \right )$;
\item[--]
    for every $v\in W_{2j+1}\cap V_0=W_{2j+2}^0\cap V_0$ there exists $v'\in E(v)$ such that $v'\in  W_{2j+1}$;
\item[--] 
    for every $v\in W_{2j+1}\cap V_1=W_{2j+2}^0\cap V_1$ we have that every $v'\in E(v)\cap \subarena$ satisfies $v'\in \bigcup_{j'\leq 2j+1} W_{j'}$;
\item[--]
    $\sum_{j} |W_{2j}|=|\subarena|$.
\end{itemize}
\end{lemmarep}

\begin{proof}
    We first show that if the first call of the algorithm starts with $V$ then every call of the algorithm works on a sub-arena. That is, in $\subarena$ every vertex $v$ has at least one successor in the sub-arena. 
    This clearly holds initially for $V$.
    The set $\mathsf{Residue}$ computed in lines~4 and 9 is the result of removing an attractor of player~0.
    It follows that all $V_0$ nodes in $\mathsf{Residue}$ have all their outgoing edges in $\mathsf{Residue}$ and all $V_1$ nodes in $\mathsf{Residue}$ have at least one edge in $\mathsf{Residue}$.
    Similarly, the recursive calls in lines~7 and 11 are applied to a sub-arena minus an attractor for player~1.
    Hence, all recursive calls are applied to a sub-arena.

    As, by assumption, the call to the algorithm returns through line~12, we have $\mathsf{Residue}=\emptyset$ and hence $\subarena=\bigcup_{j\geq 0}W_{2j}$.

    \begin{itemize}
    \item[--] In line~3 we set $W_0^0$ as $G_i\cap \subarena$. As $\subarena$ is a sub-arena, every $v\in W_0^0\cap V_0$ has some successor in the sub-arena. 
    \item[--] As $\subarena$ is partitioned to $W_0,\ldots,W_{2l}$, this follows from restricting attention to $E(v)\cap \subarena$. 
    \item[--] As $W_{2j}^p$ is a layer of an attractor, every $v\in W_{2j}^p\cap V_0$ for $p>0$ has some successor in $W_{2j}^{p-1}$.
    \item[--] As $W_{2j}$ is computed inside the sub-arena, no edges from $W_{2j}^p\cap V_1$ for $p>0$ can go to the $\mathsf{Residue}$. It follows that all edges inside the sub-arena must go to $W_{j'}$ for $j'<2j$ or to $W_{2j}^{p'}$ for $p'<p$. 
    \item[--] In line~7, $W_{2j+1}$ is computed in a sub-arena and, recursively, a similar partition of it is built. It follows that for every $v\in W_{2j+1}\cap V_0$ there exists some successor inside the same sub-arena.
    \item[--]
    Similarly, no edges from $W_{2j+1}\cap V_1$ can go to $\mathsf{Residue}$. It follows that all edges inside the sub-arena must go to $W_{j'}$ for $j'\leq {2j+1}$.
    \item[--] 
    We established that $\subarena = \bigcup_{j}W_{2j}$. 
    For all $j<j'$ we have $W_{2j}\cap W_{2j'}=\emptyset$.
    Indeed, $W_{2j'}$ is computed over a residue that does not contain $W_{2j}$. 
    Hence, $\sum_{j}|W_{2j}|=|\subarena|$.
    \end{itemize}
\end{proof}

Notice that whenever Algorithm~\ref{alg:Zielonka} is called with $0$ pairs, it immediately returns through line 12 and returns the entire sub-arena. In this case, the lemma holds vacuously.

It follows from Lemma~\ref{lem:partition_properties} that
in the dark gray areas in Figure~\ref{fig:illustration}, player 0 always can move one block to the left
while all moves of player 1 go either at least one block to the left or at least one row upwards. Thus player 0 can attract to a white area (or to $G_i$)
and then use a recursively computed winning strategy. Any resulting play either eventually avoids $R_i$ forever,
or recurringly visits $G_i$.
\subsubsection*{Signatures for Streett Games}
Based on the partition that is established by the Zielonka-McNaughton algorithm,
we analyze the internal structure of winning regions in terms of universal trees.

The \emph{signature} of a winning region in a Streett game with $n$ vertices and $k$ pairs
is a collection of $k$ different mappings into $U(n,k)$. 
Each mapping maps some individual vertices into (sets containing single) leaves of $U(n,k)$ and it maps recursively computed winning regions into smaller hierarchical trees, which are subsets of $U(n,k)$. 
Recursively, each of these winning regions is associated with its own signature, leading to $(k-1)!$ mappings into smaller hierarchical trees. This intuitive explanation is made formal in the remainder of this section. 

We use $w_{2j+1}$ and $w_{2j}^p$ to denote $|W_{2j+1}|$ and $|W_{2j}^p|$, respectively.
For $v\in W_{2j}^p$ and $v'\in W_{2j'}^{q}$, we write $v\prec v'$ if $j<j'$ or $j=j'$ and $p<q$.

\begin{lemmarep}
    When Algorithm~\ref{alg:Zielonka} returns from line 12 for sub-arena $\subarena$ and $\beta$ such that $|\beta|=d$ and given a $(|\subarena|,d)$-hierarchical tree $U$,
    there exist $d$ mappings $s_i:\subarena \rightarrow 2^{U}$ ($1\leq i\leq d$) such that:
    \begin{enumerate}
        \item 
        There is a leaf $u_0^0 \in U$ such that 
        for all $v\in W_0^0$ we have $s_i(v)=\{u_0^0\}$.
        \item 
        For all $j$ and all $p>0$ there is a leaf $u_{2j}^p\in U$ such that for all $v\in W_{2j}^p$ we have $s_i(v)=\{u_{2j}^p\}$. 
        \item For each $j>0$ there is a subtree $U_{2j}^0\subseteq U$ that is $(w_{2j}^0,d-1)$-hierarchical and for all $v\in W_{2j}^0$ we have $s_i(v)=U_{2j}^0$. 
        \item For all $v,v'$, $v\prec v'$ implies $s_i(v)<_1 s_i(v')$.
    \end{enumerate}
    \label{lem:streett signature properties}
\end{lemmarep}

Notice that from 4. it follows that the leaves and the trees used above are all disjoint. As before, when $d=0$ the lemma holds vacuously. 

\begin{proof}
    We build the mapping $s_i$. 
    By construction, $\sum_{j\geq 0}\sum_{p\geq 0}w_{2j}^p=|\subarena|$.
    It follows that these values form a partition of $|\subarena|$.
    As $U$ is $(|\subarena|,d)$-hierarchical, 
    for every $j$ and $p$ we can associate a subtree $U_{2j}^p$ such that $U_{2j}^p$ is $(w_{2j}^p,d-1)$-hierarchical, and if $j<j'$ or $j=j'$ and $p<p'$, then we have $U_{2j}^p<U_{2j'}^{p'}$.
   
    In the case that $j=0$ and $p=0$, in which case $W_0^0=G_i$, we associate with every $v\in W_0^0$ the minimal leaf $u_0^0$ in $U_0^0$.
    If $j>0$ and $p=0$, in which case $W_{2j}^0=W_{2j-1}$, we associate with every $v\in W_{2j}^0$ the subtree $U_{2j}^0$.
    If $p>0$, we choose the minimal leaf $u_{2j}^p$ in $U_{2j}^p$ and associate it with every $v\in W_{2j}^p$.

    We establish the items in the lemma.
    We note that the construction associates a different part of the tree with every different $W^p_{2j}$.
    In particular, for every two $v_1,v_2$ such that $v_1 \prec v_2$ we have $s_i(v_1)\cap s_i(v_2)=\emptyset$.
    \begin{enumerate}
        \item This is established by choosing a single leaf for the set $W^0_0$. 
        \item This is established by choosing a single leaf for such a set $W^p_{2j}$.
        \item This is established by the association of the entire tree $U_{2j}^0$ with a set of the form $W_{2j}^0$.
        \item This is established by $U_{2j}^p<U_{2j'}^{p'}$ for $j<j'$ or $j=j'$ and $p<p'$.
    \end{enumerate}
\end{proof}

\section{Solving Streett Games with Streett Rankings}
\label{sec:ranking}

We recall the definition of Streett rankings \cite{DBLP:conf/lics/PitermanP06} and establish that the mappings constructed in the previous section are indeed a Streett ranking.
By the results of Piterman and Pnueli, the direct computation of a Streett ranking in a rank lifting algorithm solves Streett games in time that is proportional to the size of the ranking. 
Our presentation slightly modifies the ranking concept from \cite{DBLP:conf/lics/PitermanP06} to adjust it to our notation and to the usage of universal trees.  

Fix a Streett condition $\alpha = \set{\pair{R_1,G_1},\ldots, \pair{R_k,G_k}}$.
Given a $k$-tree $U$, we denote by $U^\infty$ the set $U\cup \{\infty\}$. 
Recall that for $u_1,u_2\in U$, we denote by $<_j$ and $=_j$ the comparisons of the $j$-length prefixes of $u_1$ and $u_2$, respectively. 
Clearly, $u_1<_j u_2$ implies $u_1<_{j+1}u_2$. 
For all $u\in U$ we define $u<\infty$. 
Concretely, we will later use $U=U(n,k)$.
However, to reuse the results of Piterman and Pnueli, we keep notations generic in the employed tree,
which also highlights that their results (and ours) are more general than using one fixed tree.

\subsubsection*{Memory update for Streett rankings}
To win a Streett game, player 0 has to ensure that for each $i$ such that $R_i$ is visited
infinitely often, $G_i$ is visited infinitely often as well. To incorporate this requirement into rankings,
we use permutations over the pairs as memory, intuitively encoding
the relative importance of the pairs.
Namely, if a pair $\pair{R_i,G_i}$ appears before $\pair{R_j,G_j}$ in a permutation, then
a visit to $G_i$ is more important than a visit to $G_j$, and 
visits to $R_i$ are not possible while pursuing $G_j$. 
Once a visit to $G_i$ occurs, the focus changes to another pair 
that is currently less important than $\pair{R_i,G_i}$, changing the permutation and making the new pair the most important
among those that appear after $\pair{R_i,G_i}$ in the permutation.
At the same time, the previous order between all the pairs that were not more important
than $\pair{R_i,G_i}$ is forgotten and reset to 
the order of their appearance in the winning condition $\alpha$. 
We formalize this using functions $\ind$ and $\incr$ that compute the most important pair
whose $G$-component is visited by a vertex $v$
and shuffle the suffix of the permutation from that pair on,
as explained. Formally, we have the following. 

Let $\Pi(k)$ denote the set of permutations over $\{1,\ldots, k\}$.
Let $\pi=j_1,\ldots ,j_k\in \Pi(k)$ be some permutation.
Given a vertex $v$, let $\ind(v,\pi)$ denote the minimal index $l$ such that $v\in G_{j_l}$ and $k+1$ if no such $l$ exists.
Let $\incr_l(\pi)$ be the permutation $j_1,\ldots ,j_{l-1},j'_l,\ldots, j'_k$, where $j'_l$ is set to $\min\set{j_l,\ldots, j_k}$ if $j_l=\max\set{j_l,\ldots, j_k}$ and $j'_l$ is set to $\min\{j_{l'}>j_l ~|~ l' > l\}$ otherwise. 
Also, $j'_{l+1},...,j'_k$ are the elements in $\set{j_l,\ldots, j_k}\setminus \set{j'_l}$ ordered in increasing order. 
This establishes a cycle between all the pairs $\{j_l,\ldots, j_j\}$ and a strategy will try to ensure that all $G_i$ for all these pairs is visited infinitely often, logging the visits in this cyclic order. 
Thus we have $\incr_k(\pi)=\pi$; we put $\incr_{k+1}(\pi)=\pi$ as well. 
For short, we write $\incr_v(\pi)$ for $\incr_{\ind(v,\pi)}(\pi)$. We point out that 
\incr corresponds to the function $\mathit{incr}$ in \cite{DBLP:conf/lics/PitermanP06}.

\subsubsection*{Progress in Streett Games}
We define a relation between values $u_1,u_2$ in the tree $U$, measuring progress towards various $G$-goals 
according to their importance order in a permutation.
Generally, the $l$-th level in the tree is relevant for the $l$-th pair in the permutation.
Then progress is made between $u_1$ and $u_2$ at level $l$ if the $l{-}1$-prefix of $u_2$ is less
than or equal to the $l{-}1$-prefix of $u_1$.
Strict progress is made at level $l$ if the $l$-prefix of $u_2$ is strictly less
than the $l$-prefix of $u_1$.
Together with a permutation, the sets $G_i$ and $R_i$ that are visited by a game vertex $v$ identify the level $l$ at 
which progress should be made; a visit to $G_{j_l}$ allows progress on level $l$, while strict progress at level $l$ is required if $v$ visits $R_{j_l}$ but not $G_{j_l}$. 
This intuitive explanation is made formal in what follows.

We define, for a given permutation $\pi=j_1,\ldots, j_k$, a vertex $v\in V$, and two values $u_1,u_2\in U$, 
that $u_1>_{v,\pi}u_2$ iff $u_1\geq_{l-1} u_2$,
where $l$ is the minimal index such that $v\in G_{j_l}\cup R_{j_l}$ (if no such index exists, then $l=k+1$). If $l\leq k$ and $v\in R_{j_l}\setminus G_{j_l}$ then we also require $u_1>_l u_2$. 

Finally, put $\infty >_{v,\pi} u$ for every $u\in U$.
The value $\infty$ indicates that one needs to make infinite progress in order to win (which
is not possible).

\begin{definition}
A \emph{Streett ranking} over $V$ and $k$-tree $U$ is a function $r:V \times \Pi(k) \rightarrow U^\infty$.
A Streett ranking is \emph{good} if for every vertex $v$ and $\pi\in \Pi(k)$ such that $r(v,\pi)\neq \infty$ we have the following:
\begin{compactitem}
    \item[--]
        if $v\in V_0$, there is some $v'\in E(v)$ such that $r(v,\pi) >_{v,\pi} r(v',\incr_v(\pi))$;
    \item [--]
        if $v\in V_1$, for all $v' \in E(v)$ we have $r(v,\pi) >_{v,\pi} r(v',\incr_v(\pi))$.
\end{compactitem}
\end{definition}

Intuitively, a Streett ranking is good if it witnesses a strategy for player $0$ that always makes (strict) progress.
For a good ranking $r$, $v\in V_0$ and $\pi$ such that $r(v,\pi)\neq \infty$, player $0$ has a move to some $v'\in E(v)$
such that (strict) progress is made between the value that $r$ assigns to $v$ under memory value $\pi$ and the value
that $r$ assigns to $v'$ under the new memory value $\incr_v(\pi)$ that is obtained by updating $\pi$ according to the
colors visited by $v$. The level at which this progress is required to happen is determined by the colors visited by $v$ together
with $\pi$.
For $v\in V_1$, progress is required for all successors of $v$ instead.
\begin{example}\label{ex:strettrank}
For the objective $\alpha_1$ in Example~\ref{ex:el}, the strategy $\sigma_1$ is obtained from the following
Streett ranking $r$, using $T^\infty$, where $T$ is a tree with two levels, with the root having three children, referred to as 
$u_1,u_2,u_3$ (for readability of this example, we use $T$ in place of a full universal tree). Also let $\pi_1$ and $\pi_2$ denote the permutations $(2,1),(1,2)\in \Pi(2)$, respectively, corresponding to the respective
nodes $1.1$ and $2.1$ in $Z(\alpha_1)$ from Example~\ref{ex:ziel}; the relation between rankings and Zielonka trees will be explained in more detail in Section~\ref{sec:elranking}.
\begin{align*}
r(u,\pi_1)&=u_1 & r(v,\pi_1)&=u_2 & r(w,\pi_1)&=u_3 \\
r(u,\pi_2)&=u_3 & r(v,\pi_2)&=u_2 & r(w,\pi_2)&=u_1
\end{align*}
The ranking $r$ assigns $\infty$ to all other pairs. This ranking is good; for instance for $(u,\pi_1)$,
we have $u\in V_0$, $v\in E(u)$ and
\begin{align*}
r(u,\pi_1)=u_1>_{u,\pi_1} u_2 = r(v,\pi_2)= r(v,\incr_u(\pi_1)),
\end{align*}
pointing out that $\incr_u(\pi_1)=\pi_2$ since $u$ visits color $c$, corresponding to a visit to $G_2$, the
granting set of the Streett pair that is considered most relevant in $\pi_1$. Hence $\incr$
rearranges the order of the pairs to $\pi_2$ when applied to $u$ and $\pi_1$. Furthermore, 
$u_1>_{u,\pi_1} u_2$ since $u$ is colored with $c$ but not with $a$ (corresponding to 
visiting $G_2$ but not $R_2$), and $u_1\geq_0 u_2$ so that progress is made
between $u_1$ and $u_2$ at level $0$.

On the other hand, we cannot assign any leaf from $T$ to, e.g., $(y,\pi_1)$ since
from $y$, player $1$ can indefinitely force color $a$ and avoid color $c$ (corresponding to visiting $R_2$ and avoiding $G_2$, respectively),
which would necessitate an infinitely decreasing chain of leafs assigned by such a ranking,
hence preventing it from being good.
\end{example}

Piterman and Pnueli show that finding a good Streett ranking is a sound way to establish both a winning region and a witnessing winning strategy in a Streett game.
Indeed, player~0 uses a permutation $\pi\in \Pi(k)$ as memory in this strategy.
As long as the memory value is $\pi=j_1 \ldots j_k$, player~0 uses the ranking $r(\cdot,\pi)$ to determine her next move,
intuitively trying to minimize the rank $r(\cdot,\pi)$. 
Whenever a node $v$ such that $v\in G_{j_l}$ (for minimal $l$) is visited, she updates her memory to $\incr_v(\pi)$ and moves to a vertex $v'$ such that $r(v',\incr_v(\pi))$ is not $\infty$ (we say that she updates her memory even if $\ind(v,\pi)\geq k$, which implies that $\incr_v(\pi)=\pi$).
Consider an infinite play where player~0 uses this strategy.
Then none of the sets $G_{j_l}$ are visited as long as the memory does not change.
If the memory stabilizes and remains constant over an infinite suffix of the play, then for all $l$, the set $R_{j_l}$ is visited finitely often. Indeed, every time $R_{j_l}$ is visited there is some decrease in $r(\cdot,\pi)$, and it is never incremented as the only way to increment the rank is if the memory is changed.
Otherwise, there is some $l$ such that the memory changes around position $l$ infinitely often. It follows that all sets $R_{j_l'}$ for $l'<l$ are visited finitely often and all sets $G_{j_l''}$ for $l'' \geq l$ are visited infinitely often. 
Piterman and Pnueli state the following:

\begin{lemma}[{\rm\hspace{-0.1pt}\cite{DBLP:conf/lics/PitermanP06}}]
    Given a good Streett ranking $r$, player~0 wins the Streett game from every state $v$ such that for some permutation $\pi$ we have $r(v,\pi)\neq \infty$. 
    \label{lem:soundness streett ranking}
\end{lemma}

\begin{corollary}[{\rm\hspace{-0.1pt}\cite{DBLP:conf/lics/PitermanP06}}]
    A good Streett ranking induces a winning strategy using at most $k!$ memory values.
    \label{cor:strategy memory streett}
\end{corollary}

Piterman and Pnueli implicitly use rankings over the universal tree $U=\{1,\ldots, n\}^{\leq k}$.
While their proof of Lemma~\ref{lem:soundness streett ranking} is tailored to this tree, it directly carries over to the more general case as presented here.
It also follows as a special case of the proof of Lemma~\ref{lem:soundness el ranking} for Emerson-Lei games below.

We now use the insights from Section~\ref{sec:zielonka mcnaughton} to improve the results from~\cite{DBLP:conf/lics/PitermanP06}
and show that the set $U(n,k)^\infty$ (i.e., $U(n,k)\cup \{\infty\}$) is complete for solving Streett games.

\begin{lemmarep}
    For every Streett game with $n$ vertices and $k$ pairs there exists a good Streett ranking over $U(n,k)$ such that for every vertex $v$ that is winning for player~0 there is some permutation $\pi$ such that $r(v,\pi)\neq \infty$.
    \label{lem:completeness streett ranking}
\end{lemmarep}

\begin{proofsketch}
    The ranking is obtained by combinations of the signatures identified in Section~\ref{sec:zielonka mcnaughton}.
    Given a permutation, every prefix of the permutation identifies a selection of the pairs, which leads to the choice of a signature.
    At some point, the prefix identifies a single leaf rather than a subtree. This is the value associated with all the possible suffixes of this prefix. 
    The properties in Lemma~\ref{lem:streett signature properties} translate directly to the properties of the ranking. 
\end{proofsketch}
\begin{proof}
    We show that the combination of signatures for Streett games established in Section~\ref{sec:zielonka mcnaughton} are, in fact, the required ranking.

    Consider a permutation $\pi=j_1,\ldots, j_k\in \Pi(k)$ and let $\pi|_i$ denote the prefix of $\pi$ of length $i$, namely $j_1,\ldots, j_i$.
    For every permutation $\pi$ and every prefix $\pi|_i$, let $s_{\pi|_i}$ denote the signature that is defined for the choices of pairs as defined by $\pi|_i$ with the expected propagation of hierarchical trees.
    That is, $s_{\pi|_1}$ is the signature obtained for $j_1$ as the first pair chosen in line~1 of Algorithm~\ref{alg:Zielonka} with $U(n,k)$ as the hierarchical tree.
    Then, for every recursively computed region $W_{2j}^0$, $s_{\pi|_2}$ is the signature obtained when the top level call chooses $j_1$ in line~1, and the first recursive call chooses $j_2$ in line~1 using the $(|W_{2j}^0|,d-1)$-hirarchical tree $U_{2j}^0$ identified by $s_{\pi|_1}$.
    Similarly, $s_{\pi|_i}$ corresponds to the chain of recursive calls choosing $j_1$ in line~1 identifying a tree $U_1$, then $j_2$ in line~1 with tree $U_2$, and so on, until finally $j_i$ is chosen in line~1 in the $i$-th call with $U_i$.
    Notice that, as a final case, if the $k$-th call relies on calling {\sc SolveStreett} for some sub-arena with $0$ pairs, then it also identifies a unique leaf (equivalently a $(1,0)$-tree) for that sub-arena. It follows that in any case, the $k$-th call (if reached) \emph{always} associates leaves with all vertices regardless of whether they are part of an attractor or a recursively computed winning region. 

    Consider a node $v$ in the winning region and a permutation $\pi=j_1,\ldots, j_k$ and
    the series of calls to {\sc SolveStreett} with pairs $j_1$, $j_2$, and so on. 
    In every such call, either $v$ is included in a direct attractor computation, in which case the signature associates with it a single leaf, or it is included in a recursively computed winning region, in which case the signature associates with it a subtree (or a leaf at level $k$).
    If $v$ is included in a recursively computed winning region for $j_i$, then it is included in the call for $j_{i+1}$. 
    Dually, if $v$ is included in an attractor then it is not included in the sub-arena of further recursive calls.
    Let $i$ be the last call in which $v$ is included in a sub-arena, but at most $k$. 
    It follows that $v$ is included as part of an attractor (or $i=k$).  
    Let $\pi|_i$ be the corresponding prefix.
    It follows that $v$ is assigned a value in $s_{\pi|_i}$ but not in $s_{\pi|_{i+1}}$.
    This is notably the case, where $v$ is found in $W_{2j}^p$ for some $p>0$ and associated with a single value in $U(n,k)$, or where $i=k$ and the entire sub-arena sent to the $k+1$-th call of {\sc SolveStreett} is associated with a single value in $U(n,k)$. 
    Then, for every permutation $\pi'$ such that $\pi|_i=\pi'|_i$ we put $r(v,\pi')=s_{\pi|_{i}}(v)$.
    For simplicity, we denote this unique value by $s_{\pi|_{i}}(v)$, though, formally, $s_{\pi|_i}(v)$ is a set containing this unique value. 
    Clearly, for every $v\in \Win_0$ and every permutation $\pi\in \Pi(k)$ we have $r(v,\pi)\neq \infty$.

    In order to prove the progress of the ranking, we unpack the definition of $u_1 >_{v,\pi} u_2$.
    Let $\pi=j_1,\ldots, j_k$. Consider a vertex $v$ and a successor $v'$. 
    By definition, we identify the minimal index $l$ such that $v\in G_{j_l}\cup R_{j_l}$. We then require $u_1 \geq_{l-1} u_2$ and if $v\in R_{j_l}$ also $u_1 >_{l} u_2$ (or $k+1$). 
    Consider the indices, ``walking'' from $j_1$ to $j_l$ until ``meeting'' this significant index $l$.
    Starting from $j_1$, we require that $u_1\geq_0 u_2$, which always holds. If $v\in G_{j_1}$ we are done.
    If $v\in R_{j_1}\setminus G_{j_1}$ then, we require in addition that $u_1>_1 u_2$. That is, that the first entry in $u_1$ is larger than the first entry in $u_2$.
    If $1$ is not the significant index, then consider some $i>1$ that is still at most the significant index.
    For $i$, if $v\in G_{j_i}\cup R_{j_i}$, we require that $u_1\geq_{i-1} u_2$ which forces the first $i-1$ entries of $u_1$ to be bounded from below by the first $i-1$ entries of $u_2$.
    If $v\in G_{j_i}$, we are done. If $v\in R_{j_i}\setminus G_{j_i}$, we require, in addition,
    that $u_1>_{i}u_2$, which means that one of the first $i$ entries in $u_1$ is strictly larger than the same entry in $u_2$.
    Again, if $i$ is not the significant index, we continue to $i+1$.
    If we reach the end of the permutation, we found that $v$ is in no $G_i$ and no $R_i$, in that case, we just require that $u_1\geq_k u_2$, that is, that the entire entry $u_1$ is bounded from below by $u_2$. Indeed, they could be equivalent.
    Because $u_1 \geq_{l'} u_2$ implies that $u_1 \geq_{l''} u_2$ for every $l''<l'$,
    as we go along the permutation with the index $i$ ranging over $\{1,\ldots, k+1\}$, we need to enforce that $u_1 \geq_{i-1} u_2$, and then make further checks regarding $G_{j_i}$ and $R_{j_i}$, and then either stop or continue to $i+1$. 
    
    We now establish that:
    \begin{itemize}
    \item 
        For every $v\in V_0\cap \Win_0$ and all $\pi\in\Pi(k)$ there is $v'\in E(v)$ such that $r(v,\pi)>_{v,\pi} r(v',\incr_v(\pi))$.

        Let $\pi=j_1,\ldots, j_k$. 
        We iterate $i$ over $1,\ldots, k+1$.
        We establish, as we iterate over $i$, that the sub-arena $\subarena$ handled at stage $i$ ensures that for every $v' \in E(v)$ such that $v' \in \subarena$ and every permutation $\pi'$ such that $\pi'|_{i-1}=\pi|_{i-1}$ we have $r(v,\pi) \geq_{i-1} r(v',\pi')$.
        In particular, this holds if $\pi'=\incr_v(\pi)$ as getting to iteration $i$ ensures that $v\notin G_{j_{i'}} \cup R_{j_{i'}}$ for every $i'<i$. 
        Initially, {\sc SolveStreett} is called with the entire winning region as the sub-arena, and hence, $E(v)\cap \Win_0\neq\emptyset$ and it is indeed the case that for every $v'\in E(v)\cap \Win_0$ and every permutation $\pi'$ we have $r(v,\pi')\neq \infty$.
        In particular, $r(v,\pi) \geq_0 r(v',\incr_v(\pi))$.

        Consider a call of {\sc SolveStreett} handling the $i$-th pair. In particular $i\leq k$. 
        
        If $v\in W_0^0$ then $v\in G_{j_i}$. 
        Iterating over $i'<i$ we established that, for every $v'\in E(v)$ that is in the sub-arena sent to {\sc SolveStreett}, we have $r(v,\pi)\geq_{i-1} r(v',\pi')$ for every $\pi'$ such that $\pi'|_{i-1}=\pi|_{i-1}$.
        As $v\in G_{j_i}$ we have that $\incr_v(\pi)$ is an updated value. But clearly, $\pi|_{i-1}=\incr_v(\pi)|_{i-1}$. Furthermore, this is sufficient to establish that $r(v,\pi) >_{v,\pi} r(v',\incr_v(\pi))$ for every $v'\in E(v)\cap \subarena$.
        In this case we do not need to continue the iteration.

        Consider the case that $v\in W_{2j}^p$ for some $j$ and $p>0$.
        
        By the definition of attractors, there is some $v'\in E(v)\cap \subarena$ such that $v'$ is in the lower set $W_{2j}^{p-1}$.
        Furthermore, we have $v\notin W_0^0$ so that $v\notin G_{j_1}$. 
        It follows that the shared prefix between $\incr_v(\pi)$ and $\pi$ includes  $j_i$.
        Then, for every permutation $\pi'$ that has the same $i$-prefix as $\pi$ we have that $r(v,\pi)$ and $r(v',\pi')$ is set using the same signature in such a way that $r(v,\pi)>_i r(v',\pi')$. This is particularly the case if $\pi'$ is $\incr_v(\pi)$.
        It implies that $r(v,\pi)>_{v,\pi} r(v',\incr_v(\pi))$, regardless of the index $l$ such that $v\in G_{j_l} \cup R_{j_l}$ (or if $k+1$). As before, we stop the iteration.

        Consider the case that $v\in W_{2j}^0$ for $j>0$. It again follows that $v\notin G_{j_i}$.
        As $v$ is in {\tt No\_R$_{j_1}$}, we conclude $v\notin R_{j_1}$ and $v\notin\attr^{\tt Residue}_1(R_{j_i})$ (so $v$ is not a dead end in the sub-arena defined in line~6 of the top level call of {\sc SolveStreett}).
        Consider now the recursive call of {\sc SolveStreett}.
        For all nodes $v'$ that are computed by {\sc SolveStreett}, we have that their rank is in the same subtree of $U(n,k)$ such that $r(v,\pi)=_i r(v',\pi')$ for all $\pi'$ such that $\pi|_i=\pi'|_i$.
        Hence, in order to establish that 
        $r(v,\pi)>_{v,\pi} r(v',\incr_v(\pi))$ we need to check the rest of the permutation.

        The requirement for the iteration $i+1$ is that $r(v,\pi)=_i r(v',\pi')$, which we have just established. 
        We continue in this iteration until either we get to a value of $i$ where the iteration stops and establishes that $r(v,\pi)>_{v,\pi} r(v',\incr_v(\pi))$ or we get to $i=k+1$.
        In this final case, $v\notin G_j \cup R_{j}$ for all $j$ and $\incr_v(\pi)=\pi$. 
        We have established that for every $v'\in E(v)\cap \subarena$ we have $r(v,\pi) \geq_{k} r(v',\incr_v(\pi))$, which shows that $r(v,\pi)>_{v,\pi} r(v,\incr_v(\pi))$.

        The case of $v\in V_1$, below, requires to consider also the case of choice of transitions outside the sub-arena.
    \item 
        For every $v\in V_1\cap \Win_0$, every $\pi\in \Pi(k)$, and every $v'\in E(v)$ we have $r(v,\pi)>_{v,\pi}r(v',\incr_v(\pi))$.

        Let $\pi=j_1,\ldots, j_k$. 
        We iterate $i$ over $1,\ldots, k+1$.
        We establish, as we iterate over $i$, that the sub-arena $\subarena$ handled at stage $i$ ensures that for every permutation $\pi'$ such that $\pi'|_{i-1}=\pi|_{i=1}$ and all successors $v'$ of $v$ we have either (a) $v'$ is outside the sub-arena and $r(v,\pi) >_{i-1} r(v',\pi')$ or (b) $v'$ is inside the sub-arena and $r(v,\pi) \geq_{i-1} r(v',\pi')$. In particular, this holds when $\pi'$ is $\incr_v(\pi)$. 
        Initially, {\sc SolveStreett} is called with the entire winning region as the sub-arena.
        Hence, $E(v)\subseteq \Win_0$ and for every $v'\in E(v)$ and every $\pi'$ we have $r(v',\pi')\neq \infty$.
        Then, $r(v,\pi)\geq_0 r(v',\incr_v(\pi))$ establishing this requirement. 

        Consider a call of {\sc SolveStreett} handling the $i$-th pair. In particular $i\leq k$. 
        
        Consider a node $v$ and a successor $v'\in E(v)$.
        Iterating over $i'<i$ we established that for every $v'\in E(v)$ that is not in the sub-arena we have $r(v,\pi)>_{i-1} r(v',\incr_v(\pi))$, which is sufficient.

        So, in the rest of this proof, we care only about the case that $v'$ is in the sub-arena. 
        In this case, we established that for every $v'\in E(v)$ that is in the sub-arena we have $r(v,\pi)\geq_{i-1} r(v',\pi')$ for every $\pi'$ such that $\pi'|_{i-1}=\pi|_{i-1}$.

        If $v\in W_0^0$ then $v\in G_{j_1}$. Then, $\incr_v(\pi)$ is an updated value. However, $r(v,\pi)\geq_{i-1}r(v',\incr_v(\pi))$ is sufficient to establish that $r(v,\pi)>_{v,\pi}r(v',\incr_v(\pi))$ and we are done.  We stop the iteration. 

        Consider the case that $v\in W_{2j}^p$ for some $j$ and $p>0$.

        By the definition of attractors, for every $v'\in E(v)$ that is in the sub-arena we have that $v'$ is in a lower set $W_{2j'}^{p'}$ for some $j'<j$ or $j=j'$ and $p'<p$.
        In any case, $r(v,\pi)>_i r(v',\pi')$ for every permutation $\pi'$ that agrees with $\pi$ on the $i$ prefix. As $v\notin G_{j_i}$, we have that $\incr_v(\pi)$ agrees with $\pi$ on the $i$-prefix.
        This implies that $r(v,\pi) >_{v,\pi} r(v',\incr_v(\pi))$ regardless of the index $l$ such that $v\in G_{j_l} \cup R_{j_l}$ (or if $l=k+1$). We stop the iteration. 

        Consider the case that $v\in W_{2j}^0$ for $j>0$.
        We have $v\notin G_{j_i}$.
        The attractor for player~1 to $R_{j_i}$ is removed just before the recursive call to {\sc SolveStreett}. Thus we have
        $v\notin \attr^{\tt Residue}_1(R_{j_i})$.
        In particular, it follows that $v\notin R_{j_i}$.
        Consider a vertex $v'\in E(v)$ in the sub-arena.
        It can be the case that $v' \in W_{2j'}^p$ for $j'<j$ and some $p$.
        In this case $r(v,\pi)>_i r(v',\pi')$ for all permutations $\pi'$ agreeing with $\pi$ on the $i$ prefix. Particularly, for $\incr_v(\pi)$, implying $r(v,\pi)>_{v,\pi}r(v',\incr_v(\pi))$.
        These cases include all the nodes $v'\in E(v)\cap \subarena$ such that $v'\notin W_{2j}^0$.
        It follows that when we call {\sc SolveStreett} on $W_{2j}^0$, for every successor $v'$ of $v$ that is not in $W_{2j}^0$ we have $r(v,\pi)>_i r(v',\pi')$.
        For every successor $v'$ of $v$ that is in $W_{2j}^0$ and all permutations agreeing with $\pi$ on the $i$-th prefix, they are going to be allocated in the same subtree $U_{2j}^0$.
        It follows that they have $r(v,\pi)=_i r(v',\pi')$ for all the permutations $\pi'$ we care about.
        This establishes that when we iterate to $i+1$, the required conditions hold. 
        We also note that $\pi|_i = \incr_v(\pi)|_i$ as $G_{j_1}$ cannot be visited in this region.
        Furthermore, in order to establish that $r(v,\pi)>_{v,\pi} r(v',\incr_v(\pi))$ we need to check the rest of the permutation. 

        The requirement for the iteration $i+1$ was just established. 
        We continue this iteration until either we get to a value of $i$ where the iteration stops and establishes that $r(v,\pi) >_{v,\pi} r(v,\incr_v(\pi))$ or we get to $i=k+1$.
        In this final case $v\notin G_j\cup R_j$ for all $j$ and $\incr_v(\pi)=\pi$. We have established that for every $v'\in E(v)\cap \subarena$ we have $r(v,\pi)\geq_{k} r(v',\incr_v(\pi))$, which shows that $r(v,\pi)>_{v,\pi} r(v',\incr_v(\pi))$. 
    \end{itemize} 
\end{proof}

\subsubsection*{Improved rank-lifting algorithm for Streett games.}
Finally, Piterman and Pnueli show that Streett games can be solved by a direct rank-lifting algorithm that runs in time $O(kmk!|U|)$, where $U$ is the tree used in the ranking. 
By Lemma~\ref{lem:completeness streett ranking}, we can use the quasipolynomial tree $U(n,k)$ for $U$.

\begin{theoremrep}
\label{thm:streett-complexity}
    Streett games with $n$ states, $m$ edges, and $k$ pairs can be solved in time $O(mk\log(k)k!|U(n,k)|)$ and space $O(nk!\log{(n)}\log{(k)})$. 
\end{theoremrep}

To prove Theorem~\ref{thm:streett-complexity}, we propose a rank-lifting algorithm for
Streett games, show it to be correct and analyze its complexity.

 We define a lifting function that updates Streett rankings. For a Streett ranking $r$, $v\in V$ and $\pi\in \Pi(k)$, $\mathsf{Lift}(r,v,\pi)$ is the function that is defined by putting $(\mathsf{Lift}(r,v,\pi))(u,\pi')=r(u,\pi')$ for $u\neq v$ or $\pi'\neq \pi$, and 
    $$
    \begin{array}{l}
    (\mathsf{Lift}(r,v,\pi))(v,\pi) = \textstyle\max_{\leq} \{ r(v,\pi), \mathsf{lift} (r,v,\pi) \}, \mbox{ where}\\[4pt]
     \mathsf{lift}(r,v,\pi)  =
      \left \{
      \begin{array}{l l}
      \textstyle\min_\leq \big \{ \mathsf{update}(r,v,\pi,v') \, \big | \, v'\in E(v) \big \} & v \in V_0  \\
      \textstyle\max_\leq \big \{ \mathsf{update}(r,v,\pi,v') \, \big | \,v'\in E(v) \big \} & v\in V_1 
      \end{array}
    \right .
\\
    \mathsf{update}(r,v,\pi,v') = \\
    \hfill 
    \textstyle\min_\leq \big ( 
      \{\infty\} \cup 
      \{u \in U(n,k) \mid u>_{v,\pi}r(v',\incr_v(\pi))\} 
        \big  ) 
    \end{array}
    $$   

    Crucially, the definition of $\mathsf{Lift}$ encodes the property that is required for a Streett ranking to be good at $v$ and $\pi$. 
    Hence any Streett ranking $r$ that is a fixpoint of $\mathsf{Lift}$ for every $v$ and $\pi$, that is, for which we have $\mathsf{Lift}(r,v,\pi)=r$ for every $v$ and $\pi$, is a good Streett ranking: Let $r$ be a Streett ranking with $\mathsf{Lift}(r,v,\pi)=r$ for every $v$ and $\pi$, and let $v\in V_0$ and $\pi\in \Pi(k)$ such that $r(v,\pi)\neq\infty$. 
    Then we have $r(v,\pi)=\textstyle\min_\leq\{u\in U(n,k)\mid u>_{v,\pi}r(v',\incr_v(\pi))\}$ for some suitable $v'\in E(v)$, showing that
    $r(v,\pi)>_{v,\pi}r(v',\incr_v(\pi))$, as required. 
    The case for $v\in V_1$ is analogous.

   \setlength{\textfloatsep}{1pt}
    
    \begin{algorithm}[bt]
    \caption{\label{alg:rank lifting streett}$ \textsc{RankLifting}$}
    \DontPrintSemicolon
    $r=r_{min}$\;
    \While{{\rm (}$\exists v,\pi ~.~ \mathsf{Lift}(r,v,\pi)\neq r${\rm )}}{
        $r \gets \mathsf{Lift}(r,v,\pi)$\;
    }
    \Return{$r$}
    \end{algorithm}

\setlength{\textfloatsep}{\defaulttextfloatsep}

    The rank-lifting algorithm for Streett games 
    is presented as Algorithm~\ref{alg:rank lifting streett}. The algorithm 
    starts by initializing $r$ to be the minimal ranking $r_{\min}$, assigning the minimal leaf in $U(n,k)$ to all pairs 
    $(v,\pi)\in V\times\Pi(k)$, and then repeatedly applies the lifting function to $r$ for some $v$ and $\pi$ until the ranking stabilizes (cf.~\cite{DBLP:conf/stacs/Jurdzinski00,DBLP:conf/lics/PitermanP06}).
    We point out that once $r(v,\pi)=\infty$ for some $\pi$, then we can put $r(v,\pi')=\infty$ for all permutations $\pi'$. 

\begin{lemma}
    \label{lemma:rank-lifting-correct-streett}
    Algorithm~\ref{alg:rank lifting streett} can compute a good Streett ranking such that for every state $v$ winning for player~0 there is some permutation $\pi$ such that $r(v,\pi)\neq \infty$.
\end{lemma}

\begin{proof}
    Given two Streett rankings $r,r'$, we define 
    \begin{itemize}
      \item[--] $r\sqsubseteq r'$ iff for all $v\in V$ and $\pi\in \Pi(k)$, $r(v,\pi)\leq r'(v,\pi)$,
      \item[--] $(r\sqcap r')(v,\pi)=\min_\leq(r(v,\pi),r'(v,\pi))$ for $v\in V$ and $\pi\in \Pi(k)$,
      \item[--] $(r\sqcup r')(v,\pi)=\max_\leq(r(v,\pi),r'(v,\pi))$ for $v\in V$ and $\pi\in \Pi(k)$.
    \end{itemize}
    Recall that $r_{\min}$ is the ranking that assigns the minimal element in $U$ to every $v$ and $\pi$ and define $r_{\max}$ to be the ranking that assigns $\infty$ to every $v$ and $\pi$.  
    Then the set of Streett rankings, partially ordered by $\sqsubseteq$, forms a complete lattice with join $\sqcup$ and meet $\sqcap$
    and the according minimal and maximal elements $r_{\min}$ and $r_{\max}$, respectively.

    Recall the definition of $\mathsf{Lift}$ appearing above. % in Section~\ref{sec:ranking}.
    Then for every $v$ and $\pi$, $\mathsf{Lift}$ with respect to $v$ and $\pi$ is a monotone expansive operator, that is, we have $\mathsf{Lift}(r,v,\pi)\sqsubseteq\mathsf{Lift}(r',v,\pi)$ for any two
    Streett rankings $r,r'$ such that $r\sqsubseteq r'$ and $r\sqsubseteq \mathsf{Lift}(r,v,\pi)$. 
    Define the function $\mathsf{LIFT}$ on Streett-ranking by
    $\mathsf{LIFT}(r)=\textstyle\bigsqcup_{v,\pi} \mathsf{Lift}(r,v,\pi)$, which is also monotone
    and hence has a least fixpoint. 

    The least fixpoint of $\mathsf{LIFT}$ can be computed by arbitrary iteration of the different $\mathsf{Lift}(r,v,\pi)$ operators:
    Clearly, every fixpoint of $\mathsf{LIFT}$ is a fixpoint for $\mathsf{Lift}(\cdot,v,\pi)$ for every $v$ and $\pi$ and vice versa. 
    Furthermore, the height of the lattice is finite and every $\mathsf{Lift}$ is monotone and expansive.
    Hence, every arbitrary iteration of the different operators terminates.
    Let $r_{min}=r_0,\ldots$ denote such arbitrary iteration order with $r_\infty$ its limit 
    and let $r_\mu$ denote the least fixpoint of $\mathsf{LIFT}$.
    Clearly, $r_{min}\sqsubseteq r_\mu$. 
    By induction, if $r_i\sqsubseteq r_\mu$, then, for every $v$ and $\pi$ we have
    $\mathsf{Lift}(r_i,v,\pi) \sqsubseteq \mathsf{Lift}(r_\mu,v,\pi)=r_\mu$. 
    Hence, $r_\infty\sqsubseteq r_\mu$ and as both are fixpoints of $\mathsf{LIFT}$ they are equivalent. 
 
    Thus, Algorithm~\ref{alg:rank lifting streett} computes the least fixpoint of $\mathsf{LIFT}$. 
    In particular, for any good Streett ranking $r$ and every $v$ and $\pi$ we have $\mathsf{Lift}(r,v,\pi) \sqsubseteq r$. 
\end{proof}

This proves correctness. We now turn to efficient implementation of Algorithm~\ref{alg:rank lifting streett} and its complexity analysis. 

\begin{lemma}
    \label{lem:rank-lifting-complexity-streett}
    Algorithm~\ref{alg:rank lifting streett} runs in time $O(mk\log(k)k!|U(n,k)|)$ and in space $O(nk!\log{(n)}\log{(k)})$.
\end{lemma}

\begin{proof}
    It is standard that the number of possible applications of $\mathsf{Lift}$ can be at most the number of possible ranks per node and per permutation.
    Thus, there can be at most $|U(n,k)|$ applications of $\mathsf{Lift}$ per node and per permutation.
    In order to keep the search for the next vertex and permutation to lift efficient we need to keep additional information as follows \cite{DBLP:journals/siamcomp/EtessamiWS05}.
    For every pair $(v,\pi)$ we keep the rank of the ``best'' successors as well as the number of successors with this rank (the number is required only for nodes in $V_0$).
    We keep a list of pairs $(v,\pi)$ that need to be lifted.
    Initially, we add to this list all nodes and permutations $(v,\pi)$ such that $r(v,\pi)$ needs lifting as $v$ belongs to some $R_{j_l}$ for some relevant pair. 
    This can be analyzed in time proportional to $O(nk\log(k)k!)$ ($k\log(k)$ is required to explore the entire length of the permutation).
    We maintain in the list only nodes that definitely need lifting. 
    When lifting a pair $(v,\pi)$, we first explore all successors of $v$, compute the new rank $r(v,\pi)$, and update the number of successors it depends on. 
    We then explore all predecessors $^\backprime v$ of $v$ and permutations $^\backprime \pi$ such that $r(^\backprime v,^\backprime \pi)$ depends on $r(v,\pi)$ in order to remain the same (i.e., $\pi=\incr_{\backprime v}(^\backprime \pi)$), update their counts, and if they need a lift due to reliance on $(v,\pi)$ they are added to the list. 
    A complication is that given a predecessor $^\backprime v$ of $v$ the number of permutations $^\backprime \pi$ such that $r(^\backprime v,^\backprime \pi)$ depends on $r(v,\pi)$ may be very large.
    For every permutation $\psi$, the node $^\backprime v$ identifies a level $l$ that is the minimal value such that $^\backprime v\notin R_{j_{l'}}$ for any $l'<l$ and either $l=k+1$ or $^\backprime v\in G_{j_l}$. It follows that all pairs $(^\backprime v,\psi)$ with a permutation $\psi$ that agrees with $^\backprime \pi$ on the $l$-length prefix depend on the pair $(v,\pi)$.
    Indeed, for all such permutations $\incr_{\backprime v}(\psi)=\pi$. 
    Thus, when updating $r(v,\pi)$ we have to check $(^\backprime v,\psi)$ for all such permutations $\psi$. There are $(k-l)!$ such permutations, which suggests that the complexity of checking a single predecessor of a single node may lead to checking on $k!$ pairs.
    Contrarily, we show that the total number of such updates corresponding to the edge $(^\backprime v,v)$ in the game and \emph{all} permutations $\pi$ (namely, when incrementing $r(v,\pi)$) is bounded by $k!$.\footnote{We note that this issue was glanced over by Piterman and Pnueli \cite{DBLP:conf/lics/PitermanP06}, who just stated that the total amount of work per node is proportional to its out-degree and in-degree leading to overall $O(mk!)$ for all nodes and all permutations.}
    
    Consider a fixed node $^\backprime v$ and all possible permutations over $\{1,\ldots,k\}$ arranged in a tree.
    That is, the permutation $i_1,i_2,\ldots, i_k$ identifies the node $i_1,i_2,\ldots, i_k$ in the tree.
    The nodes in the tree that correspond to the significant levels $l$ with respect to $^\backprime v$ as defined above form a cut in this tree: there is a set $C$ of nodes in the tree such that each branch in the tree has exactly one node in $C$.
    A node in the cut in level $l$ in the tree identifies a dependency of $(k-l)!$ permutations on the same $\incr_{^{\backprime}v}$ permutation.
    Indeed, as $^{\backprime}v \in G_{i_l}$ every permutation that starts with $i_1,\ldots, i_l$ is changed by $\incr_{^{\backprime}v}$ to the same permutation $i_1,\ldots, i'_l,i'_{l+1},\ldots, i'_{k}$, where $i'_l$ is the least element in $i_{l+1},\ldots, i_{k}$ that is larger than $i_l$ or, if $i_l$ is the maximum of $i_{l},\ldots, i_k$, then $i'_l$ is the minimum of $i_{l+1},\ldots, i_k$, and $i'_{l+1},\ldots, i'_k$ are the rest of the elements in $i_l,\ldots, i_k$ organized in increasing order.
    However, the total number of such dependencies in the entire tree is still $k!$.
    Indeed, the total number of dependencies in a subtree rooted at $i_1,\ldots, i_l$ is bounded by $(k-l)!$.
    We show this by induction on the depth of the subtree. 
    For a subtree of depth $1$ or depth $0$, clearly the maximal number of dependencies is 1, as $\incr_{k}(\pi)=\incr_{k+1}(\pi)=\pi$ for every permutation $\pi$. 
    Consider a subtree of depth $l+1$ and a cut in it. 
    If the cut is at the root, then the cut contains exactly the root and corresponds to $(l+1)!$ permutations depending on the single permutation $\incr_{k-l-1}$ obtained from all $(l+1)!$ permutations in the subtree. 
    Otherwise, the cut defines a cut of each of the $l+1$ subtrees rooted at the children of the root of our subtree.
    By induction, the number of dependencies is bounded by $(l+1)l!=(l+1)!$.
    
    Given a node $v$ and a permutation $\pi$, a lift operation applied to $(v,\pi)$ explores all the outgoing edges of $v$ and all its incoming edges and the permutations that depend on $\pi$. 
    Overall, when iterating over all nodes and all permutations, the amount of work required is proportional to $O(mk!)$.
    Every comparison of the ranking takes time proportional to the length of (the representation of) the permutation and the rank, i.e., $O(k\log(k))$. 
    It follows that the total time complexity of the rank lifting algorithm is $O(mk\log(k)k!|U(n,k)|)$.

    The space requirement follows from the need to store $nk!$ ranking values.
    Each ranking value requires to store $k$ entries whose total length is $O(\log(n))$. 
    This can be done either by storing explicitly the separators by storing $\log(n)+k$ ternary values or by tagging every bit by which entry in the $k$-sequence it belongs to, requiring $\log(n)(1+\log(k))$ bits. 
    As extra information, each vertex and each permutation require the best rank of successors and their numbers. 
    Both fit in $O(\log(n)+k)$ or $O(\log(n)(1+\log(k))$.
    It follows that the total space requirement for the ranking is $O(nk!\min{\{\log(n)\log(k),\log(n)+k\}})$.
    In addition, the list of nodes and permutations that need lifting could be maintained by a dedicated bit per vertex and permutation, that is $O(nk!)$ bits. 
\end{proof}

\section{Emerson-Lei Games and Hierarchical Trees}
\label{sec:el}

As for Streett games, we analyze the structure of the winning regions in Emerson-Lei games and connect this structure to suitable hierarchical trees.
We start from a variant of the Zielonka-McNaughton algorithm that uses the Zielonka tree as a guide for solving Emerson-Lei games.

\subsubsection*{Zielonka-McNaughton Algorithm for Solving EL Games}
Recall that an EL game is $G=\pair{A,\alpha}$, where $A=\pair{V,V_0,V_1,E}$ is an arena, $\alpha=\pair{C,\gamma,\beta}$, and $C$ is a set of colors, $\gamma: V \rightarrow 2^C$, and $\beta \in \mathbb{B}(GF(\mathbb{B}(C)))$.
Let $Z(\alpha)=\pair{T,\tau}$ be the Zielonka tree for $\alpha$ and recall that $\tau$ is a function labeling nodes in $T$ with sets of colors
and that $d_\alpha$ is the maximal height of $Z(\alpha)$. Put $k=\lceil d_\alpha/2\rceil$.
%For a node $t\in T$ we denote by $\child(t)$ the set of children of $t$ and by $\grandchild(t)$ the set of grandchildren of $t$, 
%that is, $\grandchild(t)=\bigcup_{t'\in \child(t)}\child(t')$.
Recall that for a node $t\in T$, $\child(t)$ is the set of children of $t$ and $\grandchild(t)$ is the set of grandchildren of $t$. 
For presentation purposes, we assume, without loss of generality, that $C\models \alpha$. 
This can be achieved easily by adding a single new color $g$ to $C$ and changing $\beta$ to $\beta \vee GF \,g$, which replaces $|C|$ by $|C|+1$.
Algorithm~\ref{alg:Zielonka EL} solves EL games,
treating $C$, $\gamma$, $\beta$, and $Z(\alpha)$ as implicit.
An algorithm for the case of $C\not\models \alpha$ requires to reproduce the loop in lines 6-15 as a top level wrapper, which we avoid for space considerations. 
The algorithm works with a set of vertices $\subarena$ that constitute the sub-arena to work with and a winning node $t$ (such that $\tau(t)\models\alpha$) in the Zielonka tree. 
Initially, these parameters are the full set of vertices $V$ and the root of $Z(\alpha)$ which is labelled with $C$.

\begin{algorithm}[bt]
\caption{\label{alg:Zielonka EL}$\textsc{SolveEL}(\subarena,t)$}
\DontPrintSemicolon
\ForAll{$(t' \in \child(t))$}{
   $W_{{-}1}=\subarena$\;
   $\mathsf{Target} = \{v ~|~ \gamma(v)\cap (\tau(t) \setminus \tau(t')) \neq \emptyset \}$\; 
   $W_0=\attr^{\subarena}_0(\mathsf{Target})$ \tcp*[l]{Force visit to $\tau(t)\setminus \tau(t')$}
   $\mathsf{Residue} \gets \subarena \setminus W_0$\;
   \For{$(j\gets 0  ~;~ W_{2j-1}\neq \emptyset ~;~ j{+}{+})$}{
        $W_{2j+1}\gets \emptyset$\;
       \ForAll{$(t'' \in\child(t'))$}{
            $\mathsf{Avoid} = \{ v ~|~ \gamma(v) \cap (\tau(t') \setminus \tau(t'')) \neq \emptyset\}$\;       
            \tcp*[l]{Avoid visits to $\tau(t')\setminus \tau(t'')$}
            $\mathsf{Rem\_Avoid}\gets \mathsf{Residue} \setminus \attr^{\mathsf{Residue}}_1(\mathsf{Avoid})$ \;
            $W_{2j+1} \gets \mbox{\sc SolveEL}(\mathsf{Rem\_Avoid},t'')$\;
            \If{$(W_{2j+1}\neq\emptyset)$}{
            \Break\;
            }
       }
       $W_{2j+2} \gets \attr^{\mathsf{Residue}}_0(W_{2j+1})$\;
       $\mathsf{Residue} \gets \mathsf{Residue} \setminus W_{2j+2}$\;
   }
   \If{$(\mathsf{Residue}\neq \emptyset)$}{
       \Return{$\mbox{\sc SolveEL}(\subarena\setminus \attr^{\subarena}_1(\mathsf{Residue}),t)$}
    }
}
\Return{$\subarena$}
\end{algorithm}

 Intuitively, a call of the algorithm with parameters $\subarena$ and $t$ computes, for each $t'\in \child(t)$,
 regions in $\subarena$ where player~$0$ can eventually force a visit to a (positive) color from $\tau(t)\setminus \tau(t')$ (line 5),
 or there is some $t''\in \child(t')$ such that player~$0$ can avoid visiting (negative) colors
 from $\tau(t')\setminus \tau(t'')$ (line 10). All colors in regions for the latter case are contained in
 $\tau(t'')$; these regions are solved recursively according to $t''$ (line 11).
 
\begin{theorem}[{\rm\hspace{-0.1pt}\cite{DBLP:journals/apal/McNaughton93,DBLP:journals/tcs/Zielonka98}}]
 Algorithm~\ref{alg:Zielonka EL} solves Emerson-Lei games.
\end{theorem}
When Algorithm~\ref{alg:Zielonka EL} terminates it returns through line~18. It then outputs the winning region of player $0$ in the input Emerson-Lei game, which is also the sub-arena last analyzed. 
We restrict attention to cases where the Algorithm returns through line~18.
Just like for Streett games, a call of Algorithm~\ref{alg:Zielonka EL} with sub-arena $\subarena$ and  $t\in T$ as parameters establishes,
for each $t'\in \child(t)$ (where $\tau(t')\not\models \alpha$), the partition $W_0,W_2,\ldots, W_{2m}$ of $\subarena$ according
to the sub-objective encoded by $t$, where the attractor $W_{2j}$ has the additional partition $W_{2j}^1,\ldots, W_{2j}^{p}$. 
Notice that for each $j>0$ there is an associated node $t''\in \child(t')$ (where $\tau(t'')\models \alpha$) and that $W_{2j}^0=W_{2j-1}$.
The set $W_{2j-1}$ does not contain vertices labeled with colors from $\tau(t)\setminus \tau(t'')$.
The partition for $t$ and $t'$ has the following properties.

\begin{lemma}\label{lemma:partition}
When Algorithm~\ref{alg:Zielonka EL} returns from line~18 for sub-arena $\subarena$ and $t\in Z_\alpha$, then for every child $t'$ of $t$ the partition $W_0,\ldots, W_{2m}$ satisfies:
\begin{itemize}
\item[--] 
    for every $v\in W^0_0\cap V_0$ we have $\gamma(v) \cap (\tau(t)\setminus \tau(t'))\neq \emptyset$ and there exists $v'\in E(v)$ such that $v'\in \bigcup_{j\geq 0} W_{2j}$;
\item[--]
    for every $v\in W^0_0\cap V_1$ we have $\gamma(v) \cap (\tau(t)\setminus \tau(t'))\neq \emptyset$ and for every $v'\in E(v)\cap \subarena$ we have $v'\in \bigcup_{j\geq 0} W_{2j}$;
\item[--]
    for $p>0$ and every $v\in W_{2j}^p\cap V_0$ there is $v'\in E(v)$ such that $v'\in W_{2j}^{p-1}$;
\item[--] 
    for $p>0$ and every $v\in W_{2j}^p\cap V_1$ we have that every $v'\in E(v)\cap \subarena$ satisfies $v' \in \left (\bigcup_{j'<2j} W_{j'}\right ) \cup \left (\bigcup_{p'<p} W_{2j}^{p'}\right )$;
\item[--]
    for every $v\in W_{2j+1}\cap V_0=W_{2j+2}^0\cap V_0$ there exists $v'\in E(v)$ such that $v' \in  W_{2j+1}$;
\item[--] 
    for every $v\in W_{2j+1}\cap V_1=W_{2j+2}^0\cap V_1$ and all $v'\in E(v)\cap \subarena$ satisfies $v'\in  \bigcup_{j'\leq 2j+1} W_{j'}$;
\item[--]
    $\sum_{j} |W_{2j}|=|\subarena|$.
\end{itemize}
\end{lemma}

\begin{proof}
    As in the case of Streett games, if the first call of the algorithm starts with $V$ then every call of the algorithm works on a sub-arena. 
    This clearly holds initially for $V$.
    The set $\mathsf{Residue}$ computed in lines~5 and 15 is the result of removing an attractor of player~0.
    It follows that all $V_0$ nodes in $\mathsf{Residue}$ have all their outgoing edges in $\mathsf{Residue}$ and all $V_1$ nodes in $\mathsf{Residue}$ have at least one edge in $\mathsf{Residue}$.
    Similarly, the recursive calls in lines~11 and 17 are applied to a sub-arena minus an attractor for player~1.
    Hence, all recursive calls are applied to a sub-arena.

    As, by assumption, the call to the algorithm returns through line~18, we have $\mathsf{Residue}=\emptyset$ and hence $\subarena=\bigcup_{j\geq 0}W_{2j}$.

    \begin{itemize}
    \item[--] In line~4 we set $W_0^0$ as $\mathsf{Target}\cap \subarena$, where $\mathsf{Target}$ is exactly the set of nodes that have $\gamma(v)\cap (\tau(t)\setminus \tau(t'))\neq \emptyset$. 
    As $\subarena$ is a sub-arena, every $v\in W_0^0\cap V_0$ has some successor in the sub-arena. 
    \item[--] As $\subarena$ is partitioned to $W_0,\ldots,W_{2m}$, this follows from restricting attention to $E(v)\cap \subarena$.
    \item[--] As $W_{2j}^p$ is a layer of an attractor, every $v\in W_{2j}^p\cap V_0$ for $p>0$ has some successor in $W_{2j}^{p-1}$.
    \item[--] As $W_{2j}$ is computed inside the sub-arena, no edges from $W_{2j}^p\cap V_1$ for $p>0$ can go to the $\mathsf{Residue}$. It follows that all edges inside the sub-arena must go to $W_{j'}$ for $j'<2j$ or for $W_{2j}^{p'}$ for $p'<p$. 
    \item[--] In line~11, $W_{2j+1}$ is computed in a sub-arena and, recursively, a similar partition of it is built. It follows that for every $v\in W_{2j+1}\cap V_0$ there exists some successor inside the same sub-arena.
    \item[--]
    Similarly, no edges from $W_{2j+1}\cap V_1$ can go to $\mathsf{Residue}$. It follows that all edges inside the sub-arena must go to $W_{j'}$ for $j'\leq {2j+1}$.
    \item[--] 
    We established that $\subarena = \bigcup_{j}W_{2j}$. 
    For all $j<j'$ we have $W_{2j}\cap W_{2j'}=\emptyset$.
    Indeed, $W_{2j'}$ is computed over a residue that does not contain $W_{2j}$. 
    Hence, $\sum_{j}|W_{2j}|=|\subarena|$.
    \end{itemize}
\end{proof}

Notice that when Algorithm~\ref{alg:Zielonka EL} is called with a leaf $t$, it immediately returns through line 18 and returns the entire sub-arena. In this case, the lemma holds vacuously. 

\subsubsection*{Signatures for EL Games}
Based on the partition of the winning region established by Algorithm~\ref{alg:Zielonka EL} we analyze the winning region in terms of the universal tree.

We use the ordering between the children of every node in the Zielonka tree $Z(\alpha)=\pair{T,\tau}$
and the induced lexicographic ordering between any two nodes in the tree, denoted by $<$.
In particular, every non-leaf in $T$ has a minimal child.
Recall that $m^+_\alpha$ and $m^-_\alpha$ are the amounts of memory required to win by the respective players. 

The \emph{signature} of a winning region in an EL game (relative to the sub-objective encoded by $t$, where $\tau(t)\models\alpha$) is a collection of $|\child(t)|$ different mappings 
into $U(n,k)\times (\grandchild(t)\cup \{\bot\})$.
The mapping for $t'\in \child(t)$ maps individual game vertices into leaves of $U(n,k)$ marked with $\bot$, the latter signaling that these nodes do not depend on further recursive winning regions. 
It maps recursively-computed winning regions to smaller universal trees, which are subsets of $U(n,k)$, and to a unique element of $\grandchild(t)$.
Recursively, each of these winning regions is associated with its own signature, corresponding to at most $m^+_\alpha$ mappings into the smaller universal trees and at most $m^-_\alpha$ values. 
This intuitive explanation is made formal in the rest of this section. 

Again we write $v \prec v'$ if $v\in W_{2j}^p$, $v'\in W_{2j'}^{q}$, and $j<j'$ or $j=j'$ and $p<q$.

\begin{lemmarep}
    When Algorithm~\ref{alg:Zielonka EL} returns from line 18 for sub-arena $\subarena$ and node $t$ and given a $(|\subarena|,k-|t|/2)$-hierarchical tree $U$, then for each $t'\in \child(t)$ there are mappings $s^1_{t'}:\Win_0 \rightarrow 2^{U}$ and $s^2_{t'}:\Win_0 \rightarrow (\child(t') \cup \{\bot\})$ s.t.:
    \begin{enumerate}
        \item there is a leaf $u_0^0\in U$ such that for all $v\in W_0^0$ we have $s^1_{t'}(v)=\{u_0^0\}$.
        \item for each $j$ and $p>0$ there is a leaf $u_{2j}^p\in U$ such that for all $v \in W_{2j}^p$ we have $s^1_{t'}(v)=\{u_{2j}^p\}$.
        \item for each $j>0$ there is a subtree $U_{2j}^0\subseteq U$ that is $(w_{2j}^0,k-|t|/2-1)$-hierarchical and for all $v \in W_{2j}^0$ we have $s^1_{t'}(v)=U_{2j}^0$. 
        \item for all $v,v'$ we have $v \prec v'$ implies $s^1_{t'}(v)<_1 s^1_{t'}(v')$.
        \item for each $j$ and every $v_1,v_2\in W_{2j}$ we have 
        $s^2_{t'}(v_1)=s^2_{t'}(v_2)$. Furthermore, $s^2_{t'}(v_1)=\bot$ iff $v_1\in W_{0}$ or if $t'$ has no children.
    \end{enumerate}
    \label{lem:el signature properties}
\end{lemmarep}

Notice that from 4 it follows that the leaves and the trees used above are all disjoint.
When $t$ does not have children, the algorithm returns the whole $\subarena$ and the lemma holds vacuously.

\begin{proof}
    We use $w_{2j+1}$ and $w_{2j}^p$ to denote $|W_{2j+1}|$ and $|W_{2j}^p|$, respectively.

    Consider a child $t'$ of $t$. 
    We build the signatures $s^1_{t'}$ and $s^2_{t'}$. 
    By construction $\sum_{j\geq 0}\sum_{p\geq 0}w_{2j}^p = |\subarena|$.
    It follows that these different values form a partition of $|\subarena|$.
    As $U$ is an $(|\subarena|,k-|t|/2)$-hierarchical tree, for every $j$ and $p$ we can associate a subtree $U_{2j}^p$ such that $U_{2j}^p$ is $(w_{2j}^p,\lceil(d_\alpha-|t|-2)/2)\rceil)$-hierarchical and if $j<j'$ or $j=j'$ and $p<p'$ we have $U_{2j}^p<U_{2j'}^{p'}$.
    
    If $j=0$ and $p=0$, in which case $W^0_0=\{v ~|~ \gamma(v) \cap (\tau(t)\setminus \tau(t'))\neq \emptyset\}$, then we put $s^1_{t'}(v)=\{u_0^0\}$ for all $v\in W^0_0$, where $u_0^0$ is the the minimal leaf in $U_0^0$. We further put $s^2_{t'}(v)=\bot$. 
    For $v\in W_0^p$, we put $s^1_{t'}(v)=\{u_0^p\}$, where $u_0^p$ is the minimal leaf in $U_0^p$.
    We further put $s^2_{t'}(v)=\bot$.
    Notice that in case that $t'$ has no children, then $W_0=\subarena$. It then follows that for every $v\in\subarena$ we set $s^2_{t'}(v)=\bot$. 
    Consider the case that $j>0$.
    If $p=0$, in which case $W_{2j}^0=W_{2j-1}$, then there is a node $t'' \in \child(t')$ such that $W_{2j-1}$ is computed by the recursive call to {\sc SolveEL} with $t''$ as the node from $\child(t')$ (chosen in line~7).
    In this case, we put $s^1_{t'}(v)=U_{2j}^0$ and $s^2_{t'}=t''$ for all $v\in W_{2j}^0$.
    If $p>0$, then we choose the minimal leaf $u_{2j}^p$ in $U_{2j}^p$ and put $s^1_{t'}(v)=\{u_{2j}^p\}$ and $s^2_{t'}(v)=t''$ for all $v\in W_{2j}^p$.

    We establish the items in the lemma. We note that $s^1_{t'}$ associates a different part of $U(n,k)$ with every different $W_{2j}^p$. In particular, for any two $v_1, v_2$ such that $v_1 \prec v_2$ we have $s^1_{t'}(v)\cap s^1_{t'}(v)=\emptyset$.
    \begin{enumerate}
        \item This is established by choosing for $s^1_{t'}$ a single leaf for the set $W^0_0$.
        \item This is established by choosing a single leaf for such a set $W^p_{2j}$ for $s^1_{t'}$. 
        \item This is established by the association of the entire sub-tree $U^0_{2j}$ with a set of the form $W_{2j}^0$ for $s^1_{t'}$ according to which $W_{2j-1}$ is computed recursively.
        \item $U_{2j}^p<U_{2j'}^{p'}$ for $j<j'$ or $j=j'$ and $p<p'$. 
        \item This is established by putting $s^2_{t'}=\bot$ for all $v\in W_0$ (which is also the case if $t'$ is childless) and by setting $s^2_{t'}$ to the relevant child $t''$ for all $v\in W_{2j}$ for $j>0$. 
    \end{enumerate}
\end{proof}
\section{\hspace{-3pt}Solving Emerson-Lei Games with EL-Rankings}
\label{sec:elranking}

We generalize the definition of Streett rankings from Section~\ref{sec:ranking} 
to Emerson-Lei rankings and establish that the mappings constructed in the previous section indeed constitute an Emerson-Lei ranking.
As before, a direct rank lifting algorithm solves Emerson-Lei games in time proportional to the size of the ranking.
Strategies arising from our ranking use at most the memory that is required for
the winning condition.

We start by detailing the role that Zielonka trees play in Emerson-Lei rankings.
Fix an Emerson-Lei condition $\alpha$.
A leaf $l=d_1\cdots d_o$ in the Zielonka tree $Z(\alpha)$ identifies a \emph{branch}
$t_0, t_1, \ldots, t_o$, where $t_0=\epsilon$, 
for every $i\geq 1$ we have $t_i=d_1\cdots d_{i}$, $t_{i}\in\child(t_{i-1})$ and $\tau(t_{i-1})\supsetneq \tau(t_{i})$.
Recall that the number of leaves (and thus also the number of branches) of a Zielonka tree for an objective with set $C$ of colors is bounded by $|C|!$.

In the signatures of an EL game in Section~\ref{sec:el}, the odd directions in a branch tell us which signature $(s_1,s_2)$ to look at and the even directions are the values given by the mappings $s_2$. 
When we compose these signatures together to a ranking, the number of mappings that should be considered per vertex in the game (replacing permutations in the Streett case) corresponds to the possible number of combinations of odd directions in branches.
The values that each mapping stores augment a leaf of the universal tree with even directions.
The number of mappings corresponds to the memory required to win the game as established by the Zielonka tree \cite{DBLP:conf/lics/DziembowskiJW97} (see Corrolary~\ref{cor:strategy memory EL}).

Given a branch $t_0,t_1,\ldots, t_j$ in the Zielonka tree,
we have by assumption that $\tau(t_{2i})\models \alpha$ while $\tau(t_{2i+1})\not\models \alpha$.
We use both odd and even directions in branches to define notations that
are similar to those in the ranking for Streett games.
We separate the leaf $t_j=d_1\cdots d_j$ into its odd directions 
$o_1, o_2, \ldots o_{\lceil j/2\rceil }$, where $o_i=d_{2i-1}$, and its even directions
$e_1,\ldots, e_{\lfloor j/2\rfloor}$, where $e_i=d_{2i}$.
Clearly, $o_1 \in \child(\epsilon)$, $o_1e_1\cdots e_i\in \child(o_1e_1\cdots o_i)$ and $o_1e_1\cdots o_{i+1}\in \child(o_1e_1\cdots e_i)$, $\tau(o_1e_1\cdots e_i)\models \alpha$, and $\tau(o_1e_1\cdots o_i)\not\models \alpha$. 
Note that $\lceil j/2\rceil -1 \leq \lfloor j/2 \rfloor \leq \lceil j/2\rceil$.
Let $\Pi_\alpha$ be the set of all possible subsequences of odd directions and 
let $\Psi_\alpha$ be the set of all possible subsequences of even directions.

\subsubsection*{Memory update for EL rankings}
As for Streett games, events in the game change the memory. 
For EL games, we use $\Pi_\alpha$ as the set of memory values and $\Psi_\alpha$ as 
part of the values that vertices in the game are mapped to.
The update of a memory value $\pi\in\Pi_\alpha$ increases (cyclically) one direction in $\pi$
and resets all subsequent directions to the minimal value 1.
Together with a relevant even sequence, a single memory value determines a branch/leaf in the Zielonka tree.

Let $\pi=o_1,\ldots, o_d$ and $\psi=e_1,\ldots, e_{d'}$ be the odd and even directions, respectively, in a branch of the Zielonka tree.
Their combination defines the current relative importance of visits to parts of $C$:
A (positive) visit to $\tau(o_1e_1\cdots e_i)\setminus \tau(o_1e_1\cdots o_{i+1})$ is more important than a (negative) visit to $\tau(o_1e_1\cdots o_{i+1})\setminus \tau(o_1e_1\cdots e_{i+1})$ and a (negative) visit to $\tau(o_1e_1\cdots o_{i})\setminus \tau(o_1e_1\cdots e_i)$ is more important than a (positive) visit to $\tau(o_1e_1\cdots e_i)\setminus \tau(o_1e_1\cdots o_{i+1})$.
Concentrating on positive goals, there are $\lceil c/2 \rceil$ levels of goals to pursue. 
The $i$-th level goal is to visit $\tau(o_1e_1\cdots e_{i-1})\setminus \tau(o_1e_1\cdots o_{i})$.

Winning according to a node in the Zielonka tree at an even level is achieved by avoiding any colors not contained in the label of the node, while for each of the node's children, also visiting colors not contained in their label.
With given odd and even sequences, this translates to the $l$-th level goal that is encoded by $o_{l}$ directing at the child $o_1e_1\cdots o_{l}$ of $o_1e_1\cdots e_{l-1}$ and identifying $\tau(o_1e_1\cdots e_{l-1})\setminus \tau(o_1e_1\cdots o_{l})$ as a set of colors to pursue.
Formally, this means that while pursuing the $l$-th level goal, no visits to $C\setminus \tau(o_1e_1\cdots e_{l-1})$ are allowed to occur.
Once a visit to the $l$-th level goal occurs, we change our focus to another child of $o_1e_1\cdots e_{l-1}$, making it our new $l$-th level goal.
That is, we aim to show that the game is not stuck in colors appearing below that child.
Notice that for the current memory value, these colors are less important than level $l$.
When we move to a new child, we organize the rest of the colors according to the minimal descendant
of this new child. 

To this end, we define a function $\incr$ that takes an index $l$ as parameter and updates memory values at level $l$.
Formally, let $\incr_l(\pi,\psi)$ be the sequence $o_1,\ldots, o_{l-1},o'_l,1_{l+1},\ldots,1_{d'}$, where $o'_l$ is defined to be $o_l+1$ if $o_l<|\child(o_1e_1\cdots e_{l-1})|$, and $1$ otherwise. 
Furthermore, the sequence is padded with $1$s by adding indices $1_{l+1},...,1_{d'}$, where $d'$ is the length of the longest 
leaf below $o_1e_1\cdots e_{l-1}o'_l$.
It can be the case that $\pi =\incr_l(\pi,\psi)$ (only if $o_1e_1\cdots e_{l-1}$ has just a single child). We also put $\incr_{d''}(\pi,\psi)=\pi$
for $d''>d$. 
Notice that many different even sequences can complete one odd sequence to a valid leaf in the Zielonka tree.

Next, we define the level at which a memory update occurs, which depends on the colors labeling visited vertices in the game.
Consider a vertex $v\in V$ and two sequences $\pi=o_1,\ldots, o_d$ and $\psi=e_1,\ldots, e_{d'}$. 
If there is an index $l$ such that $\gamma(v)\subseteq \tau(o_1e_1\cdots e_{l-1})$ but $\gamma(v)\not\subseteq \tau(o_1e_1\cdots o_{l})$, then we put $\ind(v,\pi,\psi)=l$.
Otherwise, either 
\begin{inparaenum}
    \item[(a)] there is an index $l$ such that $\gamma(v)\subseteq \tau(o_1e_1\cdots o_l)$ and $\gamma(v)\not\subseteq \tau(o_1e_1\cdots e_l)$,
    \item[(b)] $d'<d$ and $\gamma(v)\subseteq \tau(o_1e_1\cdots o_d)$, or
    \item[(c)] $\gamma(v)\subseteq \tau(o_1e_1\cdots e_{d'})$.
\end{inparaenum}
In all these cases we put $\ind(v,\pi,\psi)=d+1$. 

Then $\incr_{\ind(v,\pi,\psi)}(\pi,\psi)$ either updates $\pi$ at level $\ind(v,\pi,\psi)$, where $\gamma(v)\cap (\tau(o_1e_1\cdots e_{l-1})\setminus \tau(o_1e_1\cdots o_l))\neq \emptyset$, or leaves $\pi$ unchanged. 
Notice that, as for Streett games, in the case that $\gamma(v)\cap (\tau(o_1e_1\cdots e_{l-1})\setminus \tau(o_1e_1\cdots o_l))\neq \emptyset$ and $o_1e_1\cdots e_{l-1}$ has only the child $o_l$ we still say that $\pi$ is updated. 
For short, we write $\incr_v(\pi,\psi)$ for $\incr_{\ind(v,\pi,\psi)}(\pi,\psi)$.

\subsubsection*{Progress in EL games}
We measure progress towards goals (determined according to a current branch of the Zielonka tree)
using pairs from $U\times \Psi_\alpha$, where $U$ is a $k$-tree and values from $\Psi_\alpha$ are sequences of even branching directions.
The $l$-th level in $U$ and the $l$-prefix of $\Psi_\alpha$ are relevant for the $l$-th goal as determined by the branch.
Then progress is made between $(u_1,\psi)$ and $(u_2,\psi')$ at level $l$ if the $l{-}1$-prefix of $u_2$ is smaller 
than the $l{-}1$-prefix of $u_1$, or these prefixes are identical but the branching direction of $\psi'$ at level $l-1$ is smaller
than the branching direction of $\psi$ at level $l-1$.
Progress is strict at level $l$ if the $l$-prefix of $u_2$ is smaller than the $l$ prefix of $u_1$ or, these prefixes are identical but the branching direction of $\psi'$ at level $l$ is smaller than the branching direction of $\psi$ at level $l$. 
Together with a permutation, the colors that are visited by a game vertex identify the level at 
which progress should be made, where progress is required to be strict if the most relevant colors are negative.
This intuitive explanation is made formal in what follows.

As before, for a $k$-tree $U$, let $U^\infty$ denote the set $U\cup \{\infty\}$. 
We use the prefix ordering of nodes in $U$ as before and establish prefix ordering on elements of $\Psi_\alpha$ and pairs from $U^\infty \times \Psi_\alpha$. 
Let $\psi=e_1,\ldots, e_d$ and $\psi'=e'_1,\ldots, e'_{d'}$ be two even sequences in $\Psi_\alpha$.
Then, $\psi >_i \psi'$ if there is a $j\leq i$ such that for all $j'<j$ we have $e_{j'}=e'_{j'}$ and $e_j>e'_{j'}$ or $e_j$ does not exist (i.e., $\psi'$ is a prefix of $\psi$). 
We write $\psi\geq_i \psi'$ if $\psi>_i\psi'$ or $e_1,\ldots, e_i=e'_1,\ldots, e'_i$. 
For the combination, given $u_1,u_2\in U$ and $\psi_1,\psi_2\in \Psi_\alpha$, we write $(u_1,\psi_1)>_i (u_2,\psi_2)$ if there is $k<i$ such that $u_1=_k u_2$, $\psi_1=_k \psi_2$, and either $u_1>_{k+1} u_2$ or $u_1=_{k+1} u_2$ and $\psi_1>_{k+1} \psi_2$. 
We write $(u_1,\psi_1)\geq_i (u_2,\psi_2)$ if $(u_1,\psi_1)>_i(u_2,\psi_2)$ or $u_1=_iu_2$ and $\psi_1=_i\psi_2$.
Finally, we write $\geq$ and $>$ (without index) to denote that the relation holds for some index $i$.
Consider a branch $t_0, t_1,t_2,\ldots, t_k$ in the Zielonka tree, the partition of the leaf $t_k$ to its odd sub-sequence $\pi=o_1,\ldots, o_d$ and its even sub-sequence $\psi=e_1,\ldots, e_{d'}$, and a vertex $v\in V$. Recall that $t_0=\epsilon$. 
Given two values $u_1,u_2\in U$ and an additional even sub-sequence $\psi'=e'_1,\ldots, e'_{d'}$, we define 
$(u_1,\psi)>_{v,\pi}(u_2,\psi')\text{ iff }(u_1,\psi)\geq_{l-1} (u_2,\psi')$ where
$l$ is the minimal index such that 
$\gamma(v)\subseteq \tau(o_1e_1\cdots e_{l-1})$
but $\gamma(v)\not\subseteq \tau(o_1e_1\cdots e_{l-1}o_l e_{l})$ (and $d+1$ if no such index exists).
If $l\leq d$ but $\gamma(v)\subseteq \tau(o_1e_1\cdots e_{l-1}o_l)$ 
then we additionally require $(u_1,\psi)>_{l} (u_2,\psi')$
as in this case, the most relevant colors for $v,\pi$ and $\psi$ are negative.
In addition, we define $\infty >_{v,\pi} (u,\psi)$ for all $u\in U$ and $\psi\in \Psi_\alpha$,
where $\infty$ again indicates that one needs to make infinite progress in order to win. 

\subsubsection*{Rankings for EL games}

Consider a function $r:V\times \Pi_\alpha \rightarrow (U^\infty \times \Psi_\alpha)$.
If $r(v,\pi)=(u,\psi)$, we write $r_1(v,\pi)$ for $u$ and $r_2(v,\pi)$ for $\psi$; 
if $r_1(v,\pi)=\infty$ then we write $r(v,\pi)=r_2(v,\pi)=\infty$. 
Such a function is \emph{Zielonka-tree consistent} (with $Z(\alpha)=\pair{T,\tau}$) if it satisfies the following conditions:
\begin{compactitem}
    \item[--] 
    For every $v\in V$ and every $\pi\in \Pi_\alpha$ such that $\psi=r_2(v,\pi)\neq \infty$ we have $\pi$ and $\psi$ interleave to a leaf/branch in $T$. 
    This implies that $|\psi| \leq |\pi| \leq |\psi|+1$.
    \item[--]
    For every $v\in V$ and every two odd sequences $\pi_1,\pi_2\in \Pi_\alpha$ such that $r(v,\pi_1)\neq \infty$ and $r(v,\pi_2)\neq \infty$, if $\pi_1|_l=\pi_2|_l$ then $r_1(v,\pi_1)|_l=r_1(v,\pi_2)|_l$ and $r_2(v,\pi_1)|_l=r_2(v,\pi_2)|_l$, where $\cdot|_l$ takes the $l$-length prefix of a sequence. 
\end{compactitem}

\begin{definition}
An Emerson-Lei ranking over $V$ (and $U$) is a Zielonka-tree consistent function $r:V \times \Pi_\alpha \rightarrow (U^\infty\times \Psi_\alpha)$.
An Emerson-Lei ranking is \emph{good} if for every state $v\in V$ and all $\pi\in \Pi_\alpha$ such that $r_1(v,\pi)\neq \infty$ we have
\begin{compactitem}
    \item[--]
        if $v\in V_0$, there is some $v'\in E(v)$ such that $r(v,\pi) >_{v,\pi} r(v',\incr_{v}(\pi,\psi))$;
    \item[--]
        if $v\in V_1$, then for each $v' \in E(v)$ we have $r(v,\pi) >_{v,\pi} r(v',\incr_{v}(\pi,\psi))$.
\end{compactitem}
\end{definition}

\begin{example}\label{ex:elrank}
Coming back to Example~\ref{ex:el}, we consider a ranking for objective $\alpha_2$, 
using the tree $T$ in which the root
has two children, each of which has a single child; $T$ has just two leaves $u_{1},u_2$.
The branches in the Zielonka tree $Z(\alpha_2)$ (Example~\ref{ex:ziel}) are encoded by the sequences $\underline{1}1$, $\underline{1}2\underline{1}1$, $\underline{2}1$
of branching directions (with odd branching directions underlined), so we have
$\Pi_{\alpha_2}=\{\underline{1},\underline{1}\underline{1},\underline{2}\}$ and
$\Psi_{\alpha_2}=\{1,21\}$; we point out that not all combinations
combine to a branch in $Z(\alpha_2)$.
Then $\sigma_1$ can be extracted from the following ranking:
\begin{align*}
r(u,\underline{1})&=(u_{1},1) & r(v,\underline{1})&=(u_{2},1) & r(w,\underline{1})&=(u_{2},1) \\
r(u,\underline{1}\underline{1})&=(u_{1},21) & r(v,\underline{1}\underline{1})&=(u_{2},21) & r(w,\underline{1}\underline{1})&=(u_{2},21)\\
r(u,\underline{2})&=(u_{2},1) & r(v,\underline{2})&=(u_{2},1) & r(w,\underline{2})&=(u_{1},1)
\end{align*}
The ranking $r$ assigns $\infty$ to all other pairs. This ranking is good; for instance for $(u,\underline{1}\underline{1})$,
we have that $r_2(u,\underline{1}\underline{1})=21$ completes $\underline{1}\underline{1}$ to the branch
$\underline{1}2\underline{1}1$ in $Z(\alpha_2)$; also we have
$u\in V_0$, $v\in E(u)$ and
$
r(u,\underline{1}\underline{1})=(u_{1},21)>_{u,\underline{1}\underline{1}} (u_{2},1) = r(v,\underline{2}),
$
pointing out that $\incr_u(\underline{1}\underline{1},21)=\underline{2}$ since $c\in\gamma(u)$ (so that the memory changes
around the first position), and that $(u_{1},21)>_{u,\underline{1}\underline{1}} (u_{2},1)$ since
$\ind(u,\underline{1}\underline{1},21)=1$, $c\in\gamma(u)$ and $u_1=_0 u_2$.
Contrarily, we cannot assign any leaf from $T$ to, e.g., $(y,\underline{1})$ as from $y$, player $1$ can avoid color $c$ forever,
but force colors $a,d$ infinitely often.
\end{example}

Next, we
show that the
computation of good Emerson-Lei rankings 
is a sound way to establish winning regions and extract witnessing winning strategies in Emerson-Lei games.

\begin{lemmarep}
    Given a good EL ranking, player~0 wins the EL game from every state $v$ such that for some 
    sequence $\pi$ we have $r(v,\pi)\neq \infty$. 
    \label{lem:soundness el ranking}    
\end{lemmarep}
\begin{proofsketch}
    Player~0 uses the sequence $\pi=o_1,\ldots, o_d\in \Pi_\alpha$ as memory.
    As long as the memory value is $\pi$, player~0 uses the two parts of the ranking $r(\cdot,\pi)$ to determine her next move,
    intuitively trying to minimize the rank $r(\cdot,\pi)$.
    Whenever she reaches a vertex $v$ such that $\gamma(v)\cap (\tau(o_1e_1\cdots e_{l-1})\setminus \tau(o_1e_1\cdots o_{l})\neq\emptyset$ she chooses a successor $v'$ such that $r(v',\incr_v(\pi,r_2(v,\pi)))$ agrees with $r(v,\pi)$ on the $l-1$-prefix (in particular, it is required 
    that $r(v',\incr_v(\pi,r_2(v,\pi)))\neq \infty$), she updates her memory to $\incr_v(\pi,r_2(v,\pi))$ and continues playing (we say that she updates her memory even if $\pi'=\pi$). 
    Consider an infinite play where player~0 uses this strategy.

    Consider the case that the memory eventually stabilizes. 
    Then we have that for all $l$ the situation that $\gamma(v)\cap (\tau(o_1e_1\cdots e_{l-1})\setminus \tau(o_1e_1\cdots o_l))\neq \emptyset$ does not happen anymore.
    Every time that $\gamma(v)\cap (\tau(o_1e_1\cdots {o_l})\setminus \tau({o_1e_1\cdots e_l}))\neq \emptyset$ there is a decrease in $r_1(\cdot,\pi)$ or $r_2(\cdot,\pi)$ and these values are never incremented as the only way to increment the rank is if the memory is changed.
    So, eventually the play gets stuck in some $\tau(o_1e_1\cdots e_l)$ and is won by player $0$.
    
    Otherwise, there is a minimal point $l$ for which the memory changes around point $l$ infinitely often. 
    It follows that eventually always $\gamma(v)\subseteq \tau(o_1e_1\cdots e_{l-1})$. Furthermore,
    for every $o \leq |\child(o_1e_1\cdots e_{l-1})|$ we have that some color from $\tau(o_1e_1\cdots e_{l-1})\setminus \tau(o_1e_1\cdots e_{l-1}o)$ is visited infinitely often. 
    Hence, the set of colors visited infinitely often satisfies $\alpha$ by definition of the children of a node in the Zielonka tree.
\end{proofsketch}

\begin{proof}
    Consider a play starting in $v_0$ and let $\pi_0\in\Pi_\alpha$ be such that $r(v_0,\pi_0)\neq \infty$.
    We construct a play $v_0,\ldots$ and a corresponding sequence of memory values $\pi_0,\ldots$ such that $r(v_i,\pi_i)\neq \infty$.
    Let $\psi_i=r_2(v_i,\pi_i)$.

    Consider the case that $v_i\in V_0$. By goodness of the ranking, there is some $v_{i+1}\in E(v_i)$ such that $r(v_i,\pi_i)>_{v_i,\pi_i} r(v_{i+1},\pi_{i+1})$, where $\pi_{i+1}=\incr_{v_i}(\pi_i,\psi_i)$. 
    As $r(v_i,\pi_i)\neq \infty$ it cannot be the case that $r(v_{i+1},\pi_{i+1})=\infty$. 

    Consider the case that $v_i\in V_1$ and consider a successor $v_{i+1}$ of $v_i$ chosen by player~1. 
    By goodness of the ranking we have $r(v_i,\pi_i)>_{v_i,\pi_i} r(v_{i+1},\pi_{i+1})$, where $\pi_{i+1}=\incr_{v_i}(v,\pi_1,\psi_i)$.
    As $r(v_i,\pi_i)\neq \infty$ it cannot be the case that $r(v_{i+1},\pi_{i+1})=\infty$. 

    We have established two infinite sequences $v_0,\ldots$ and $\pi_0,\ldots$ such that for every $i$ we have
    $r(v_i,\pi_i)>_{v_i,\pi_i} r(v_{i+1},\pi_{i+1})$.
    Let $l_i$ be the minimal index such that 
$\gamma(v_i)\subseteq \tau(o_1e_1\cdots e_{l_i-1})$
but $\gamma(v_i)\not\subseteq \tau(o_1e_1\cdots e_{l_i-1}o_{l_i} e_{l_i})$.
    Then, either 
    \begin{inparaenum}
    \item[(a)]
    $\gamma(v_i)\cap (\tau(o_1e_1\cdots e_{l_i-1}) \setminus \tau(o_1e_1\cdots o_{l_i}) \neq \emptyset$ and the $l_i$th element of $\pi_{i+1}$ is updated,
    \item[(b)]
    $r_1(v_i,\pi_i)>_{l_i} r_1(v_{i+1},\pi_{i+1})$ and $\pi_{i+1}|_{l_i}=\pi_i|_{l_i}$, or
    \item[(c)]
    $r_1(v_i,\pi_i)=_{l_i} r_1(v_{i+1},\pi_{i+1})$, $r_2(v_i,\pi_i)>_{l_i} r_2(v_{i+1},\pi_{i+1})$, and $\pi_{i+1}|_{l_i}=\pi|_{l_i}$.
    \end{inparaenum}

    Let $l$ be the minimal index in $l_0,\ldots$ that appears infinitely often.
    Let $i_0$ be the index such that for every $i'>i_0$ we have $l_{i'}\geq l$.
    It follows that for all $i'>i_0$ we have $r_1(v_{i'},\pi_{i'})=_{l-1}r_1(v_{{i'}+1},\pi_{{i'}+1})$.
    It is also the case that $\pi_{i'}|_{l-1}=\pi_{{i'}+1}|_{l-1}$ and $\psi_{i'}|_{l-1}=\psi_{{i'}+1}|_{l-1}$.
    Let $e_1,\ldots, e_{l-1}$ be $\psi_{i'}|_{l-1}$ and let $o_1,\ldots, o_{l-1}$ be $\pi_{i'}|_{l-1}$.
    It follows that for every $i'>i_0$ we have $\gamma(v_i)\subseteq \tau(o_1e_1\cdots e_j)\cap \tau(o_1e_1\cdots o_j)$ for all $j<l$ (it is also the case that $\tau(o_1e_1\cdots o_je_j)\supset \tau(o_1e_1\cdots o_j)$).
    Consider the infinitely many positions greater than $i_0$ such that either $r_1(v_i,\pi_i)>_{l}r_1(v_{i+1},\pi_{i+1})$, $r_2(v_i,\pi_i)>_{l}r_2(v_{i+1},\pi_{i+1})$, or $\gamma(v)\cap (\tau(o_1e_1\cdots e_{l_{i}-1})\setminus \tau(o_1e_1\cdots o_{l_i}))\neq \emptyset$. Clearly, we cannot have an infinite suffix with only the first two happening.
    Whenever the third happens, if $\pi_i$ has $o$ as the child number of $o_1e_1\cdots e_{l-1}$ then $\pi_{i+1}$ has the next child of $o_1e_1\cdots e_{l-1}$.
    It follows, that for every $o\leq |\child(o_1e_1\cdots e_{l-1})|$ we have $inf(\gamma(v_0),\ldots) \not\subseteq \tau(o_1e_1\cdots e_{l-1}o)$. On the other hand, we have for every $i'\geq i_0$ that $\gamma(v_{i'})\subseteq \tau(o_1e_1\cdots e_{l-1})$. It follows that the play is winning for player~0.
\end{proof}

Notably, unlike reductions to parity games, which embed the memory required for both players into the arena, our approach results in a strategy with optimal memory usage.
The following corollary follows from the winning strategy constructed in the proof of Lemma~\ref{lem:soundness el ranking}.

\begin{corollary}
    A good EL ranking induces a winning strategy using at most $m^+_\alpha$ memory.
    \label{cor:strategy memory EL}
\end{corollary}

\begin{proof}
    The strategy constructed in the proof of Lemma~\ref{lem:soundness el ranking} uses an entry $\pi\in \Pi_\alpha$ as memory and updates it.
    Consider how many memory values/odd sequences one needs to represent below a certain node $t$ of the Zielonka tree.

    For a node $t$ in the Zielonka tree, let $mem(t)$ denote the number of odd sequences that are extensions of the odd sequences that include this node.

    If $t\models \varphi$, then we need to represent all the odd sequences that start with all $t'$ for $t'\in \child(t)$. 
    So, $mem(t)=\sum_{t'\in \child(t)} mem(t')$.

    If $t\not\models \varphi$, then we need to represent the odd sequences that are below one of its children.
    So, $mem(t)=\max_{t'\in \child(t)} mem(t')$.

    If $t$ is a leaf, then one memory value is sufficient.

    It follows that the number of odd sequences is computed in exactly the same way as $m^+_\alpha$ and the two are equivalent.

    Notice that the number of even sequences is equivalent to $m^-(\alpha)$. 
\end{proof}

The results of Section~\ref{sec:el} can be used to show that EL rankings using the universal tree $U(n,d_\alpha/2)$ are also complete. 

\begin{lemmarep}
    For every EL game there exists a good EL ranking $r$ using $U(n,d_\alpha/2)$ such that for every state $v$ winning for player~0 there is some sequence $\pi$ such that $r(v,\pi)\neq \infty$.
    \label{lem:completeness el ranking}
\end{lemmarep}
\begin{proofsketch}
    The
    ranking is obtained by combining the signatures identified in Section~\ref{sec:el}.
    This time the signatures include two parts corresponding to values in $U(n,k)$ and to nodes in the Zielonka tree.
    As we go down the Zielonka tree, we collect the winning nodes in the Zielonka tree from the second part of the signatures.
    At some point we get to a location in the tree such that the first signature of the node identifies a single leaf in the universal tree.
    Then, we set the first ranking to this single leaf 
    and the second ranking to the even subsequence corresponding to the losing node in the Zielonka tree. 
    The recursive construction ensures that the ranking is Zielonka-tree consistent. 
    The properties in Lemma~\ref{lem:el signature properties} translate directly to the properties of the ranking. 
\end{proofsketch}

\begin{proof}
    We show that the combination of signatures for EL games established in Section~\ref{sec:el} are, in fact, the required ranking.

    We define the ranking by scanning the Zielonka tree from the root. 
    Consider a vertex $v$.
    Consider a child of the root $o_1$.
    Let $s^1_{o_1}$ and $s^2_{o_1}$ denote the two parts of the signature that is defined for $o_1$.
    Recall that $k=\lceil d_\alpha/2\rceil$ and let $U(n,k)$ be the hierarchical tree used by $s^1_{o_1}$.
    Extend $o_1$ to a branch in the Zielonka tree by induction.
    Suppose that $o_1,\ldots, o_l$ and $e_1,\ldots, e_{l-1}$ have been identified and for every $l'\leq l$ we have identified the signatures $s^1_{o_1,\ldots, o_{l'}}$ and $s^2_{o_1,\ldots, o_{l'}}$ such that $e_{l'}=s^2_{o_1,\ldots, o_{l'-1}}(v)$.

    Consider the case that $s^2_{o_1,\ldots, o_l}(v)=e_l\neq \bot$.
    In this case, $v$ is included in some recursively computed region $W_{2j}^0$ for $j>0$ or in the attractor $W_{2j}^p$ for $j\geq 0$ and $p>0$. 
    In the second case, $s^1_{o_1,\ldots, o_l}(v)$ is a unique leaf $u$ of $U(n,k)$. 
    Consider some leaf $t_k$ in the Zielonka tree such that $\pi=o_1,\ldots, o_d$ for some extension of $o_1,\ldots, o_l$ is its subsequence of odd directions and $\psi=e_1,\ldots, e_{d'}$ is its subsequence of even directions for some extension of $e_1,\ldots, e_{l}$.
    Then, for $\pi$ we set $r_1(v,\pi)=s^1_{o_1,\ldots, o_l}(v)$ (though, formally, $s^1_{o_1,\ldots, o_l}(v)$ is a set containing this unique value).
    We set $r_2(v,\pi)=\psi$.
    If $v\in W_{2j}^0$ (and in this case $j>0$), then $s^1_{o_1,\ldots, o_l}(v)$ is $U_{2j}^0$ for some $(|W_{2j}^0|,k-l-1)$-hierarchical tree. 
    Notice that in the case that $e_l$ has no children it is impossible that $s^2_{o_1,\ldots, o_l}(v)$ is $e_l$.
    Indeed, the loop in lines 8-13 iterates over an empty set of children in line 8 and uses $W_{2j+1}=\emptyset$ as set in line 7.
    Hence, in the case that $s^2_{o_1,\ldots, o_l}(v)=e_l$, for every child $o_1e_1\cdots e_lo_{l+1}$ of $o_1e_1\cdots e_l$, we consider the signature $s^1_{o_1,\ldots, o_{l+1}}$ using the $(|W_{2j}^0|,k-l-1)$-hierarchical tree $U_{2j}^0$.
    The recursive construction of the signature also carries $s^2_{o_1,\ldots, o_l}(v)=e_l$.

    Consider the case that $s^2_{o_1,\ldots, o_l}(v)=\bot$.
    This only happens when $v$ is in $W_0^0$ or when $o_1,e_1,\ldots, o_l$ does not have children. 
    In both cases, a single leaf in $U(n,k)$ is associated with $v$.
    Hence, $|s^1_{o_1,\ldots, o_l}(v)|=1$.
    Consider some leaf $t_k$ in the Zielonka tree such that $\pi=o_1,\ldots, o_d$ for some extension of $o_1,\ldots, o_l$ is its subsequence of odd directions and $\psi=e_1,\ldots, e_{d'}$ is its subsequence of even directions for some extension of $e_1,\ldots, e_{l-1}$.
    Then, for $\pi$ we set $r_1(v,\pi)=s^1_{o_1,\ldots, o_l}(v)$ (though, formally, $s^1_{o_1,\ldots, o_l}(v)$ is a set containing this unique value).
    Furthermore, we set $r_2(v,\pi)=\psi$. 
    We do the same for every leaf $t_k$ that extends $o_1e_1\cdots o_l$.
    This completes the definition of the ranking.
    Clearly, for every $v\in \Win_0$ and every odd sequence $o_1,\ldots, o_d$ we have $r(v,\pi)\neq \infty$. 

    Consider $v$ and $\pi_1$ and $\pi_2$ such that $r(v,\pi_1)\neq \infty$ and $r(v,\pi_2)\neq \infty$. 
    It is simple to see that $r_2(v,\pi_1)$ completes $\pi_1$ to a branch in the Zielonka tree.
    By the recursive construction of the ranking, we construct the prefixes of both parts of the ranking as we go down the tree.
    Thus, if $\pi_1|_l=\pi_2|_l$ then $r_1(v,\pi_1)|_l=r_1(v,\pi_2)|_l$ and $r_2(v,\pi_1)|_l=r_2(v,\pi_2)|_l$.
    Thus, the function is Zielonka-tree consistent. 
    
    We now establish that the ranking given by $r_1$ and $r_2$ is good.
    Fix $\pi=o_1,\ldots, o_d$ and $\psi=e_1,\ldots, e_{d'}$. 
    \begin{itemize}
        \item 
        We show that for every $v\in V_0\cap \Win_0$ such that $r_2(v,\pi)=\psi$ there is some edge $(v,v')$ such that $r(v,\pi)>_{v,\pi} r(v',\incr_v(\pi,\psi))$.

        Let $\pi=o_1,\ldots, o_d$ and let $\psi=e_1,\ldots, e_{d'}$.
        We iterate $i$ over $\{1,\ldots, d\}$.
        We establish, as we iterate over $i$, that the sub-arena $\subarena$ handled at stage $i$ ensures that for every $v'\in E(v)$ such that $v'\in \subarena$ and every odd sequence $\pi'$ such that $\pi'|_{i-1}=\pi|_{i-1}$ we have $r(v,\pi)\geq_{i-1} r(v',\pi')$.
        In particular, this holds for $\incr_v(\pi,\psi)$ as getting to iteration $i$ ensures that $\gamma(v) \subseteq \tau(o_1e_1\cdots o_i)$, which implies that $\incr_v(\pi,\psi)|_{i-1}=\pi_{i-1}$.
        Initially, {\sc SolveEL} is called with the entire winning region as the sub-arena, and hence, $E(v)\cap \Win_0\neq \emptyset$.
        It is also the case that for every $v'\in E(v)\cap \Win_0$ and every odd sequence $\pi'$ we have $r(v',\pi')\neq \infty$.
        It follows that the empty prefix of both $r_1$ and $r_2$ satisfy our requirements.

        Consider a call of {\sc SolveEL} handling the $i$-th element of $o_1,\ldots, o_d$. 
        This call corresponds to node $t_i=o_1e_1\cdots o_{i-1}e_{i-1}$ in the Zielonka tree and iterates over all children of $t_i$. Let $o_i$ be some child of $t_i$. 
        Let $W_{2j}^p$ be the set according to which $v$ is included in the sub-arena according to $o_i$. 
        
        In the case that $j=0$ and $p=0$, we have $\gamma(v)\cap (\tau(o_1e_1\cdots e_{l-1})\setminus \tau(o_1e_1\cdots o_i))\neq \emptyset$ and $\incr_v(\pi,\psi)$ is updated.
        Iterating over $i'<i$ we established that for every $v'\in E(v)$ that is in the sub-arena sent to {\sc SolveEL} we have $r(v,\pi) \geq_{i-1} r(v,\pi')$.
        Clearly, $\pi|_{i-1}=\incr_v(\pi,\psi)|_{i-1}$ and this is sufficient to establish that $r(v,\pi)>_{v,\pi}r(v',\incr_v(\pi,\psi))$. We stop the iteration. 

        In case that $p>0$, by the definition of attractors, there is some edge $(v,v')\in E$ such that $v'$ is contained in $W_{2j}^{p-1}$.
        Furthermore, $v\notin W_0^0$ so $\gamma(v)\cap (\tau(o_1e_1\cdots e_{l-1})\setminus \tau(o_1e_1\cdots o_l))=\emptyset$. 
        The shared prefix between $\pi$ and $\incr_v(\pi,\psi)$ includes $i$.
        Then, for every $\pi'$ that has the same $i$-prefix as $\pi$ we have that $r(v,\pi)$ and $r(v',\pi')$ are set using the same signature in such a way that $r(v,\pi)>_i r(v',\pi')$.
        This is particularly the case for $\incr_v(\pi,\psi)$ implying that $r(v,\pi)>_{v,\pi} r(v',\incr_v(\pi,\psi))$.
        We stop the iteration. 

        Consider the remaining case that $v\in W_{2j}^0$ for $j>0$. It again follows that $\gamma(v)\cap (\tau(o_1e_1\cdots e_{l-1})\setminus \tau(o_1e_1\cdots o_l))=\emptyset$.
        By construction $e_l \leq |\child(o_1e_1\cdots o_l)|$ is the child according to which $v$ is returned by the recursive call to {\sc SolveEL}.
        It follows that $\gamma(v)\cap (\tau(o_l)\setminus \tau(e_l))=\emptyset$.
        As $v$ is not contained in the set {\tt Avoid}, $v$ is not a dead end in the sub-arena.
        Consider now the recursive call of {\sc SolveEL}.
        For all nodes $v'$ that are computed by {\sc SolveEL} we have that their rank $r_1(v',\cdot)$ is in the same subtree of $U(n,k)$ such that $r_1(v,\pi)=_i r_1(v',\pi')$ and $r_2(v,\pi)=_i r_2(v',\pi')$ for all $\pi'$ such that $\pi|_i=\pi'|_i$. In particular, this holds for $\incr_v(\pi,\psi)$.
        Hence, in order to establish that $r(v,\pi)>_{v,\pi} r(v',\incr_v(\pi,\psi))$ we need to check the rest of the sequences. 

        The requirement for the iteration $i+1$ is that $r(v,\pi)=_i r(v',\pi')$, which we have just established.
        We continue in this iteration until either we get to a value of $i$ where the iteration stops and establishes that $r(v,\pi)>_{v,\pi} r(v',\incr_v(\pi,\psi))$ or we get to the end of the sequence.
        In this final case, $\gamma(v)\subseteq \tau(o_1e_1\cdots o_i)$ and $\incr_v(\pi)=\pi$.
        We have established that for every $v'\in E(v)\cap \subarena$ we have $r(v,\pi)\geq_{i-1} r(v',\pi)$, which shows that $r(v,\pi)>_{v,\pi} r(v,\incr_v(\pi,\psi))$.

        The case of $v\in V_1$, below, requires to consider also the case of choice of transitions outside the currently considered sub-arena.
        \item 
        We show that for every $v\in V_1\cap \Win_0$ and every edge $(v,v')$ we have $r(v,\pi)>_{v,\pi}r(v',\incr_v(\pi,\psi))$.

        Let $\pi=o_1,\ldots, o_d$ and let $\psi=e_1,\ldots, e_{d'}$.
        We iterate $i$ over $\{1,\ldots, d\}$.
        We establish, as we iterate over $i$, that the sub-arena $\subarena$ handled at stage $i$ ensures that for every odd sequence $\pi'$ such that $\pi'|_{i-1}=\pi|_{i-1}$ and all successors $v'$ of $v$ we have either (a) $v'$ is outside the sub-arena $SA$ and $r(v,\pi)>_{i-1} r(v',\pi')$ or (b) $v'$ is inside the sub-arena and $r(v,\pi)\geq_{i-1} r(v',\pi')$. 
        Initially, {\sc SolveEL} is called with the entire winning region as the sub-arena and the root of $Z_\alpha$ as the node.
        Hence, $E(v)\subseteq \Win_0$ and for every $v'\in E(v)$ and every $\pi'$ we have $r(v',\pi')\neq \infty$.
        It follows that the empty prefix for both $r_1$ and $r_2$ satisfy our requirements. 

        Consider a call of {\sc SolveEL} handling the $i$-th element of $o_1,\ldots, o_d$.
        This call corresponds to the node $t_i=o_1e_1\cdots o_{i-1}e_{i-1}$ in the Zielonka tree and iterates over all the children of $t_i$. Let $o_i$ be some child of $t_i$. Let $W_{2j}^p$ be the set according to which $v$ is included in the sub-arena according to $o_i$.

        Iterating over $i'<i$ we established that for every $v'\in E(v)$ that is not in the sub-arena we have that $r(v,\pi)>_{i-1} r(v',\incr_v(\pi,\psi))$, which is sufficient.

        So, in the rest of this proof, we care only about the case that $v'$ is in the sub-arena. In this case, we established that for every $v'\in E(v)$ that is in the sub-arena we have $r(v,\pi)\geq_{i-1} r(v',\pi')$ for every $\pi'$ such that $\pi'|_{i-1}=\pi|_{i-1}$.

        In case that $j=0$ and $p=0$, then $\gamma(v)\cap (\tau(o_1e_1\cdots e_{l-1})\setminus \tau(o_1e_1\cdots o_l))\neq \emptyset$. Then, $\incr_v(\pi,\psi)$ is updated.
        However, $r(v,\pi)\geq_{i-1}r(v',\incr_v(\pi,\psi))$ is sufficient to establish that $r(v,\pi)>_{v,\pi}r(v',\incr_v(\pi,\psi))$ and we are done. We stop the iteration. 

        In case that $p>0$, by the definition of attractors, for every edge $(v,v')\in E$ we have that $v'$ is contained in $W_{2j'}^{p'}$ for some $j'<j$ or $j=j'$ and $p'<p$.
        In both cases, $r(v,\pi)>_i r(v',\pi')$ for every $\pi'$ that agrees with $\pi$ on the $i$ prefix. 
        As $\gamma(v)\cap \tau(o_1e_1\cdots e_{l-1})\setminus \tau(o_1e_1\cdots o_l)=\emptyset$, we have that $\incr_v(\pi,\psi)$ agrees with $\pi$ on the $i$ prefix. 
        This implies that $r(v,\pi)>_{v,\pi} r(v',\incr_v(\pi,\psi))$. We stop the iteration.

        Consider the case that $v\in W_{2j}^0$ for $j>0$.
        It again follows that $\gamma(v)\cap (\tau(o_1e_1\cdots e_{l-1})\setminus \tau(o_1e_1\cdots o_l))=\emptyset$. 
        By construction $e_l \leq |\child(o_1e_1\cdots o_l)|$ is the child of $o_1e_1\cdots o_l$ according to which $v$ is returned by the recursive call to {\sc SolveEL}.
        It follows that $\gamma(v)\cap (\tau(o_1e_1\cdots o_l)\setminus \tau(o_1e_1\cdots e_l))=\emptyset$.
        Consider an edge $(v,v')\in E$ such that $v'$ is in the sub-arena.
        It can be the case that $v'\in W_{2j'}^p$ for $j'<j$ and some $p$.
        In this case $r_1(v,\pi)>_i r_1(v',\pi')$ for all $\pi'$ that agree with $\pi$ on the $i$ prefix. 
        This in particular holds for $\incr_v(\pi,\psi)$, which implies $r(v,\pi)>_{v,\pi} r(v',\incr_v(\pi,\psi))$.
        This includes all nodes $v'\in E(v)\cap \subarena$ such that $v'\notin W_{2j}^0$. 
        It follows that when we call {\sc SolveEL} on $W_{2j}^0$ for every successor $v'$ of $v$ that is not in $W_{2j}^0$ we just established that $r(v,\pi)>_i r(v',\pi')$.
        For every successor $v'$ of $v$ that is in $W_{2j}^0$ and all $\pi'$ agreeing with $\pi$ on the $i$ prefix, they are allocated a rank in the same subtree $U_{2j}^0$.
        It follows that they have $r_1(v,\pi)=_i r_1(v',\pi')$ and $r_2(v,\pi)=_i r_2(v',\pi')$. This establishes that when we iterate to $i+1$, the conditions required by the inductive proof indeed hold.
        We also note that $\pi|_i=\incr_v(\pi,\psi)|_i$ as colors from the set $(\tau(o_1e_1\cdots e_{l-1})\setminus \tau(o_1e_1\cdots o_l))$ are not visited in this region.
        Furthermore, in order to establish that $r(v,\pi)>_{v,\pi} r(v',\incr_v(\pi,\psi))$, we need to continue with the iteration. 

        The requirements for iteration $i+1$ were just established. We continue the iteration until either we get to a value of $i$ where the iteration stops and establishes that $r(v,\pi)>_{v,\pi}r(v',\incr_v(\pi,\psi))$ or we get to a node $t_i$ such that $t_i$ is a leaf. The case of leaves $t_i$ such that $t_i\not\models \alpha$ is handled through the case of $W_0^0$. It follows that $t_i\models\alpha$ and as $\gamma(v) \subseteq \tau(t_i)$ we have $\incr_v(\pi,\psi)=\pi$. 
        Furthermore, we established that for every $v'\in \subarena$ we have $r(v,\pi)\geq_i r(v',\incr_v(\pi,\psi))$, which shows that $r(v,\pi)>_{v,\pi}r(v',\incr_v(\pi,\psi))$.
    \end{itemize} 
\end{proof}

\subsubsection*{Rank-lifting algorithm for EL games.}
As in the case of Streett games, Emerson-Lei games can be solved by computing good rankings using a direct rank-lifting algorithm.

\begin{theoremrep}\label{thm:elqp}
    EL games with $n$ nodes, $m$ edges and $c$ colors 
    can be solved in time $O(mc(\log(c)+\log(n)) c! u(n,c/2))$ and space $O(nc!\min\{\log(n)\log(c),\log(n)+c\})$.     
\end{theoremrep}

To prove Theorem~\ref{thm:elqp} we propose a rank-lifting algorithm for EL games, show it to be correct and analyze its complexity. 

%    We note that $c\geq d_\alpha$. We use $U(n,d_\alpha/2)$ in the proof and replace $d_\alpha$ by $c$ in the upper bound. 
    We need to update the definition of the lift operator for Streett games so that it produces Zielonka-tree consistent rankings by the rank associated with multiple odd sequences simultaneously:
    when lifting the rank of a pair $(v,\pi)$, we need to ensure that all permutations $\pi'$ that share a prefix with $\pi$ are updated simultaneously. 
    We do this by updating the rank of affected odd sequences to the minimum acceptable that is larger than the prefix of the rank imposed by $\pi$. 
    For $v'\in V$:
    $$
    \begin{array}{l}
      \mathsf{update}(r,v,\pi,v') =       \textstyle\min_\leq (\{ \infty\} \quad \cup \\[4pt]\qquad
      \left \{ (u,\psi)%\in U(n,c/2)\times\Psi_\alpha 
      \left  |
      ~(u,\psi)>_{v,\pi} r(v',\incr_{v}(\pi,\psi))
      \right \} \right )
    \end{array}
    $$

    The rank-lifting algorithm for EL games then starts by initializing $r$ to be the minimal ranking $r_{\min}$ 
    which assigns, to each a pair $(v,\pi) \in V \times \Pi_\alpha$, the minimal leaf in $U(n,d_\alpha/2)$ and the minimal $\psi^\pi_{\min}\in\Psi_\alpha$ such that $\pi$ and $\psi^\pi_{\min}$ combine to a branch.
    That is, $\psi=1\cdots 1$ of an appropriate length. 
    It then repeatedly applies the lifting function to $r$ for some vertex $v$ and odd sequence $\pi$ until the ranking stabilizes.

Recall that $c\geq d_\alpha$. We use $U(n,d_\alpha/2)$ in the proof and replace $d_\alpha$ by $c$ in the upper bound. 
The proof of Theorem~\ref{thm:elqp} follows from the following Lemmata.

\begin{lemma}
    \label{lemma:rank-lifting-correct-el}
    Algorithm~\ref{alg:rank lifting streett} can compute a good EL ranking such that for every state $v$ winning for player~0 there is some sequence $\pi$ such that $r_1(v,\pi)\neq \infty$.
\end{lemma}

\begin{proof}
    Just like in the Streett case, the order on $U(n,d_\alpha/2)^\infty\times \Psi_\alpha$ induces a complete lattice on the space of ranking functions. 
   
   We use $\leq$ to denote the lexicographic order on $\Psi_\alpha$, and let $\psi^\pi_{\min}$ denote the minimal element of $\Psi_\alpha$ according to this order that combines to a branch with $\pi$. 
   That is, $\psi^\pi_{\min}=1\cdots 1$ of an appropriate length. 
   We use $\geq$ and $>$ also for pairs $(u,\psi)\in U(n,d_\alpha/2)\times \Psi_\alpha$ as defined earlier.
   This is the lexicographic order on the interleaving of $u$ and $\psi$, starting with elements from $u$. 

   Given two EL rankings $r,r'$, we then define 
    \begin{itemize}
      \item $r\sqsubseteq r'$ iff for all $v\in V$ and $\pi\in \Pi_\alpha$ we have $r(v,\pi')\leq_{v,\pi'} r'(v,\pi'')$,
      \item $(r\sqcap r')(v,\pi)=\min_\leq (r(v,\pi),r'(v,\pi))$, 
      \item $(r\sqcup r')(v,\pi)=\max_\leq(r(v,\pi),r'(v,\pi))$, 
    \end{itemize}

    It is possible to see that $r\sqcap r'$ and $r\sqcup r'$ are Zielonka-tree consistent.
    As $\psi=(r\sqcap r')_2(v,\pi)$ is either $r_2(v,\pi)$ or $r'_2(v,\pi)$ it follows that $\psi$ completes $\pi$ to a branch in $T$. 
    Consider two sequences $\pi_1$ and $\pi_2$ that agree up to level $l$. 
    Then, $r_1(v,\pi_1)$ and $r_1(v,\pi_2)$ agree up to level $l$ and $r'_1(v,\pi_1)$ and $r'_1(v,\pi_2)$ agree up to level $l$.
    Hence, if $r_1(v,\pi_1)\leq r'_1(v,\pi_1)$ and the difference is before level $l$, then $r_1(v,\pi_2)\leq r'_1(v,\pi_2)$ as well. 
    If the difference occurs after level $l$ then it is still the case that up to level $l$ all ranks agree and the result is consistent. 
    The same holds for $r_2$ and $r'_2$ and for $\sqcup$ as well. 

    Then the set of EL rankings, partially ordered by $\sqsubseteq$, forms a complete lattice with join $\sqcup$ and meet $\sqcap$
    and the according minimal and maximal elements $r_{\min}$ and $r_{\max}$, respectively.

    Next, we define a lifting function that updates EL rankings.
    For two sequences $\pi$ and $\pi'$, let $\mathsf{max\_pref}(\pi,\pi')$ denote the length of the joint prefix between $\pi$ and $\pi'$, namely, the maximal index $l$ such that $\pi|_l=\pi'|_l$.
    For an EL ranking $r$, $v\in V$, and $\pi\in \Pi_\alpha$, we define
    $\mathsf{Lift}(r,v,\pi)$ to be the function that for $u\neq v$ and $\pi'$ sets $(\mathsf{Lift}(r,v,\pi))(u,\pi')=r(u,\pi')$, and for $\pi'$ sets $(\mathsf{Lift}(r,v,\pi))(v,\pi')$ as defined in Figure~\ref{fig:EL lift}.
    This time, a $\mathsf{Lift}$ according to $v$ and $\pi$ could change the ranking for $v$ and other permutations $\pi'$ (in the second line). 
    This is required when the rank of $v$ and $\pi$ changes in a location that is shorter than the joint prefix of $\pi$ and $\pi'$ ($\mathsf{max\_pref}(\pi,\pi')$).
    In this case, we need to ensure that the new rank of $v$ and $\pi'$ is at least as large as the relevant prefix of the rank of $v$ and $\pi$ ($\mathsf{up}(\cdot)$).
    \begin{figure*}[bt]
    $$
    \begin{array}{l}
      (\mathsf{Lift}(r,v,\pi))(v,\pi)  = \textstyle\max_\leq \{r(v,\pi),\mathsf{lift}(r,v,\pi)\}, \mbox{ and}\\
      (\mathsf{Lift}(r,v,\pi))(v,\pi')  =
      \textstyle\max_\leq \left \{ r(v,\pi'), \mathsf{up}\Big (\mathsf{lift}(r,v,\pi),\mathsf{max\_pref}(\pi,\pi')\Big ) \right \}, \pi'\neq \pi, \mbox{ where}\\
      \mathsf{up}((u,\psi),l) = 
        \textstyle\min_\leq \left \{
        (u',\psi') \left | u|_l=u'|_l \mbox{ and } \psi|_l=\psi'|_l
\right. \right \} 
    \\
    \mathsf{lift}(r,v,\pi) = 
    \left \{
    \begin{array}{l l}
    \textstyle\min_\leq \big \{ \mathsf{update}(r,v,\pi,v') \big \} & v \in V_0 \\
    \textstyle\max_\leq \big \{ \mathsf{update}(r,v,\pi,v') \big \} & v \in V_1 
    \end{array}
    \right .\\
    \mathsf{update}(r,v,\pi,v') = \\
    \multicolumn{1}{c}{
    \textstyle\min_\leq \left ( 
    \{ \infty\} \cup 
      \left \{ (u,\psi)\in U(n,c/2)\times\Psi_\alpha ~\Big |~
      (u,\psi)>_{v,\pi} r(v',\incr_{v}(\pi,\psi)) \right \}  \right )}
    \end{array}
    $$
    \caption{\label{fig:EL lift}The functions $\mathsf{Lift}$ and $\mathsf{lift}$ for $v\in V$ and $\pi\in \Pi_\alpha$.}
    \end{figure*}

    Consider a Zielonka-tree compatible ranking $r$. Let $r'=\mathsf{Lift}(r,v,\pi)$. We show that $r'$ is also Zielonka-tree compatible.
    For $u\neq v$ and for every $\pi'$ we have $r(u,\pi')=r'(u,\pi')$ so that compatibility is obvious for $u$. 
    We concentrate on $r'(v,\cdot)$.
    If $r'(v,\pi)=r(v,\pi)$ then compatibility is obvious again.
    Consider the case that $r'(v,\pi)\neq r(v,\pi)$ and let $l'$ be the index such that $r'(v,\pi)$ is different from $r(v,\pi)$ in index $l'$.
    Consider a sequence $\pi'$ that agrees with $\pi$ for a shorter prefix than $l'$, then $r(v,\pi')$ does not change.
    If $\pi'$ agrees with $\pi$ on a prefix of length $l'$ or longer, then due to the increase of $r'(v,\pi)$ compared to $r(v,\pi)$ we have that $r'(v,\pi')$ is the minimal value that agrees with $r'(v,\pi)$ up to length $l$ and is larger than $r(v,\pi')$. 
    Consider now two sequences $\pi'$ and $\pi''$. 
    If the joint prefix of $\pi'$ and $\pi''$ extends their joint prefix with $\pi$ then both are increased in the same way in $r'$ and are still compatible.
    If the joint prefix of $\pi'$ and $\pi''$ is shorter than their joint prefix with $\pi$, then the joint prefix of $\pi'$ and $\pi''$ does not change and they are still compatible. 

    Thus, $\mathsf{Lift}$ with respect to $v$ and $\pi$ is a monotone expansive operator. That is, we have $\mathsf{Lift}(r)\sqsubseteq\mathsf{Lift}(r')$ for any two
    EL rankings $r,r'$ such that $r\sqsubseteq r'$ and $r \sqsubseteq \mathsf{Lift}(r,v,\pi)$. 
    Let $\mathsf{LIFT}(r)=\displaystyle\bigsqcup_{v,\pi} \mathsf{Lift}(r,v,\pi)$, which is also monotone and hence has a least fixpoint. 

    Crucially, the definition of $\mathsf{LIFT}$ encodes the property that is required for an EL ranking to be good. 
    Hence any EL ranking $r$ that is a fixpoint of $\mathsf{LIFT}$ (that is, for which we have $\mathsf{Lift}(r,v,\pi)=r$ for every $v$ and $\pi$) is a good EL ranking: 
    Let $r$ be an EL ranking with $\mathsf{Lift}(r,v,\pi)=r$ for every $v$ and $\pi$, and let $v\in V_0$ and $\pi\in \Pi_\alpha$
    such that $r(v,\pi)\neq\infty$; also let $\psi$ denote $r_2(v,\pi)$. 
    Then there is a suitable $v'\in E(v)$ such that $r(v,\pi)=\textstyle\min_\leq\{(u,\psi)\in U(n, c/2)\times\Psi_\alpha\mid (u,\psi)>_{v,\pi}r(v',\pi')\}$ for $\pi'=\incr_{v}(\pi,\psi)$.
    It follows that $r(v,\pi)>_{v,\pi}r(v',\pi')$, as required. The case for $v\in V_1$ is analogous.

    The least fixpoint of $\mathsf{LIFT}$ can be computed by arbitrary iteration of the different $\mathsf{Lift}(r,v,\pi)$ operators:
    Clearly, every fixpoint of $\mathsf{LIFT}$ is a fixpoint for $\mathsf{Lift}(\cdot,v,\pi)$ for every $v$ and $\pi$ and vice versa. 
    Furthermore, the height of the lattice is finite and every $\mathsf{Lift}$ is monotone and expansive.
    Hence, every arbitrary iteration of the different operators terminates.
    Let $r_{min}=r_0,\ldots$ denote such arbitrary iteration order with $r_\infty$ its limit 
    and let $r_\mu$ denote the least fixpoint of $\mathsf{LIFT}$.
    Clearly, $r_{min}\sqsubseteq r_\mu$. 
    By induction, if $r_i\sqsubseteq r_\mu$, then, for every $v$ and $\pi$ we have
    $\mathsf{Lift}(r_i,v,\pi) \sqsubseteq \mathsf{Lift}(r_\mu,v,\pi)=r_\mu$. 
    Hence, $r_\infty\sqsubseteq r_\mu$ and as both are fixpoints of $\mathsf{LIFT}$ they are equivalent. 

    The rank-lifting algorithm for EL games then starts by initializing $r$ to be the minimal ranking $r_{\min}$, assigning the minimal leaf in $U(n,d_\alpha/2)$ and $\psi^\pi_{\min}\in\Psi_\alpha$
    to all pairs 
    $(v,\pi)\in V\times\Pi_\alpha$, and then repeatedly applies the lifting function to $r$ until the ranking stabilizes.
    The algorithm itself is identical to Algorithm~\ref{alg:rank lifting streett} with ranking $r$ ranging over EL-rankings, $\pi$ ranging over odd sequences rather than permutations, and $r_{min}$ as explained here. 

    The rank lifting algorithm computes the least fixpoint of the join of all the different $\mathsf{Lift}$ functions. 
    In particular, for every good EL ranking $r$, and every $v$ and $\pi$ we have
    $\mathsf{Lift}(r,v,\pi)\sqsubseteq r$. 
    Hence the result of the algorithm is a good EL ranking such that for every state $v$, $v$ is winning for player $0$ if and only if there is some
    permutation $\pi$ such that the computed ranking assigns $(u,\psi)$ to the pair $(v,\pi)$, where $u$ is a leaf in $U(n,d_\alpha/2)$ (and not $\infty$). 
\end{proof}

    This shows the correctness of the algorithm. 
    We now turn to an efficient implementation of the algorithm and its complexity.

\begin{lemma}
\label{lem:rank-lifting-complexity-el}
Algorithm~\ref{alg:rank lifting streett} runs in time $O(m c(\log(c)+\log(n)) z_\alpha \cdot u(n,d_\alpha/2))$ and in space $O(nz_\alpha \min\{\log(n)\log(d_\alpha), \log(n)+d_\alpha\})$.%\np{Replace $c$ by $d_\alpha$.}
\end{lemma}

\begin{proof}
    In order to perform the rank lifting efficiently, we proceed as for the Streett ranking.
    For every pair $(v,\pi)$ we keep track of the rank of the ``best'' successors as well as the number of successors with this rank (the number is required only for vertices in $V_0$).
    We keep a queue of pairs $(v,\pi)$ that need to be lifted.
    Initially, we add to this queue all vertices and odd subsequences $(v,\pi)$ such that $r(v,\pi)$ needs lifting as $\gamma(v)\cap (\tau(o_1e_1\cdots e_{l}) \setminus \tau(o_1e_1\cdots o_{l+1}))\neq\emptyset$ for some index $l$. 
    This can be analyzed in time proportional to $O(nc\log(c)z_\alpha)$ ($c\log(c)$ is required to explore the colors in $c$ according to the order dictated by the leaf, $z_\alpha$ is the number of leaves of $Z(\alpha)$).
    We maintain in the queue only vertices that definitely need lifting. 
    When lifting a pair $(v,\pi)$, we first explore all successors of $v$, compute the new rank $r(v,\pi)$, and update the number of successors it depends on.
    Suppose that $r(v,\pi)$ was updated on the $l$-th point in one of the sequences. Then, when updating $r(v,\pi)$, we affect all odd sequences that agree with $\pi$ for longer than the $l$ prefix. 
    In terms of complexity, we consider lifting all of them as ``work'' performed by the algorithm and count them separately. 
    We then explore all predecessors $^\backprime v$ of $v$ and all odd sequences $^\backprime \pi$ such that $r(^\backprime v,^\backprime \pi)$ depends on $r(v,\pi)$ in order to remain the same. 
    For all such pairs, we update their counts and if they need a lift due to their reliance on $(v,\pi)$, they are added to the queue. 
    As in the case of Streett games, the number of such dependencies is potentially large.
    We analyze the total number of such updates corresponding to the edge $(^\backprime v,v)$.
    We show that the total number of increments per node in the Zielonka tree is $U(n,d_\alpha/2)$ and that the amortized cost of an increment is proportional to the number of edges connecting the vertex $v$.
    Thus, the total complexity is as stated. 

    We now turn to consider the details.
    Recall the notations $|C|=c$, $Z(\alpha)=\pair{T,r}$, $z_\alpha$ the number of leaves of $T$. Let $b$ be the maximum branching width in $T$.
    Given a vertex $v$, we suggest an encoding of the ranking $r(v,\cdot)$ as a map $r_v:T\rightarrow (U(n,d_\alpha/2) \times ([1..b]\cup \set{\bot})) \cup \{\downarrow,\infty\}$.
    We show that using this encoding it is possible to update efficiently the rank as it is incremented and that the amortized cost per update is proportional to the number of edges that are connected to a vertex.
    The function $r_v$ satisfies the following conventions.
    Consider a leaf $t_k$ in $T$ composed from the odd and even sequences $o_1,\ldots, o_d$ and $e_1,\ldots, e_{d'}$. 

    We use $r_v(\epsilon)=\infty$ to signify that for every odd sequence $\pi$ we have $r(v,\pi)=\infty$.
    Otherwise, $r_v(t)\neq \infty$ for every $t\in T$ and $r_v(t)=\downarrow$ for every node $t$ in even depth of $T$.

    Consider a node in odd depth $t=o_1e_1\cdots e_{l-1}o_l$ (i.e, $\tau(t)\not\models \alpha$).
    We have $r_v(t) =(u,\mathsf{dir}) \in U(n,d_\alpha/2) \times ([1..b]\cup \set{\bot})$ such that 
    $|u|=l$ and $\mathsf{dir}$ is either $\bot$ or $\mathsf{dir}\leq |\child(t)|$.
    That is, $r_v$ associates $t$ with a node in $U(n,d_\alpha/2)$ in the same depth as the length of $o_1,\ldots, o_l$ and with the direction of a child of $t$ or with $\bot$ (no child). 

    Furthermore, we have the following:
    \begin{compactitem}
        \item 
        If $r_v(t)=(u,\bot)$ for $u\in U(n,d_\alpha/2)$, then 
        this represents the situation that for all extensions $\pi'$ of the odd sequence $o_1,\ldots, o_{l}$ and the even sequence $\psi'=e_1,\ldots, e_{l-1},1,1,1,\ldots$, we have $r(v,\pi')=(u',\psi')$, where $u'$ is the least leaf under $u$. 

        In this case, the values that $r_v$ associates with nodes that appear below $t$ in the tree are not relevant. 
        \item 
        If $r_v(t)=(u,e_l)$ for $u\in U(n, d_\alpha/2)$ and $e_l \leq |\child(t)|$, then  
        this represents the situation that for every child $te_lo$ of $te_l$, the rank of extensions of $o_1,\ldots, o_l,o$ are determined by the value that $r_v$ associates with nodes further down in the tree. 

        In particular, in this case we maintain that $r_v(te_lo)=(u',\mathsf{dir})$ for a child $u'$ of $u$ in $U(n,d_\alpha/2)$. 
    \end{compactitem}
    Notice that in the case that $t$ has no children, it is impossible to  recursively find winning regions according to $t$ in line $11$ of Algorithm~\ref{alg:Zielonka EL}.
    In such a case, the entire computation is the partition of a region to $W_0$, meaning that encoding a node in $U(n,d_\alpha/2)$ is sufficient for the ranking.
    In such a case, $r_v$ associates with $t$ a value tagged by $\bot$. 

    In the case that $te_l$ does not have children, the entire relevant region of the game is returned immediately in line $18$ of Algorithm~\ref{alg:Zielonka EL}.
    In such a case, $r_v$ associates a value with the grandfather of $t$, which determines the ranking.

    We notice that for every odd sequence $o_1,\ldots, o_d$, the above representation allows exactly one even sequence $e_1,\ldots, e_{d'}$ to be associated with it. 
    Notice also that $r_{\min}$ under this notation is the function $r_v$ that associates every child of the root $\epsilon$ of $T$ with the pair $(u,\bot)$, where $u$ is the minimal child $u$ of the root of $U(n,d_\alpha/2)$.
    Indeed, this associates every odd sequence $\pi$ with the least leaf in $U(n,d_\alpha/2)$ and the even sequence $1,1,1,\ldots$ that completes $\pi$ to a branch in $T$. 

    For every leaf of the tree $T$ there can be at most $U(n,d_\alpha/2)$ lift operations.
    Indeed, once a certain node $t$ at odd depth points to a certain child $e$ (in even depth) and this changes to the next child $e'$, the only way to get back and increment in the subtree below $e$ is if the node of $U(n,d_\alpha/2)$ associated with $t$ changes or something changes in one of the ancestors of $t$.
    It follows that the total number of increments over all nodes is $O(nz_\alpha u(n,d_\alpha/2))$.
    Notice that we consider each leaf as lifted separately, even though a single lift can affect multiple leaves together. 
    We now compute the amortized cost of one increment.

    Consider the pair $(v,\pi)$ such that $r(v,\pi)=(u,\psi)$ needs to be incremented. 
    For every successor $v'\in E(v)$ we need to compute $\mathsf{update}(r,v,\pi,v')$.
    This essentially amounts to checking $(u,\psi) >_{v,\pi} (u',\psi')$, where $r(v',\incr_v(\pi,\psi))=(u',\psi')$.
    However, the representation $r_{v'}$ includes exactly one candidate for $\psi'$.
    This candidate can be found by traversing $r_{v'}$ from the root of $T$ gradually building $\pi'$ and $\psi'$.
    This search is therefore linear in $d_\alpha$.

    Once we determine the new value of $r(v,\pi)$, we need to consider the predecessors $^\backprime v\in E^{-1}(v)$ and branches of the tree (corresponding to a combination $^\backprime \pi$ and $^\backprime \psi$) such that $\pi$ is $\incr_{^\backprime v}(^\backprime \pi,^\backprime \psi)$.

    Consider the edge $(^\backprime v,v)\in E$. 
    Consider a branch $t_0,\ldots, t_k$ in the Zielonka tree.
    The vertex $^\backprime v$ identifies a level $l$ that is the minimal such that $\gamma(^\backprime v) \cap (\tau(t_l) \setminus \tau(t_{l+1}))$. 
    If $l$ is even and $t'_{l+1}$ is the next sibling of $t_{l+1}$, then every branch/leaf that starts with $t_{l+1}$ depends on the branch $t_0,\ldots,t_l,t'_{l+1},1,1,1,1,\ldots$.
    If $l$ is odd, then every branch/leaf that starts with $t_0,\ldots, t_{l+1}$ depends only on itself. 
    The nodes in the tree that correspond to the minimal levels $l$ with respect to $^\backprime v$ as defined above form a cut in the tree. 
    A node in the cut in even level $l'$ in the tree identifies a dependency of all its descendants on the same branch in the tree.
    The total number of such dependencies is bounded by the number of leaves of the subtree under $t_l$. 
    We show this by induction on the depth of the subtree. 
    For a subtree of depth $1$ or depth $0$, clearly the maximal number of dependencies is 1. 
    Consider a subtree of depth $p+1$ and a cut in it. 
    If the cut is at the root of the subtree $t$, then the cut contains exactly the root and can have at most all leaves as dependents. 
    Otherwise, the cut defines a cut of each of the subtrees rooted at the children of the root of our subtree.
    By induction, the number of dependencies is bounded by the sum of dependencies in each subtree, which corresponds to the total number of leaves. 
    Overall, the number of dependencies that correspond to the edge $(^\backprime v,v)$ is at most $z_\alpha$ for all possible rank updates.
    It follows that a certain edge $(^\backprime v,v)$ is affected by at most $O(z_\alpha u(n,d_\alpha/2))$ lifting operations.
    Overall, the amortized cost of an edge per lift is hence constant. 
    
    To summarize, given a vertex $v$ and an odd sequence $\pi$, a lift operation applied to $(v,\pi)$ explores all the outgoing edges of $v$ and all its incoming edges and the odd sequences that depend on $\pi$. 
    Overall, when iterating over all lifting operations applied to $v$, number of ``edge crossings'' is proportional to $O((|E(v)|+ |E^{-1}(v)|)z_\alpha u(n,d_\alpha/2))$. 
    Every comparison of the ranking takes time proportional to the depth of the Zielonka tree, scanning the colors according to their order, and the length of the rank, $O(c(\log(c)+\log(n)))$. 
    It follows that the total time complexity is $O(mc(\log(c)+\log(n))z_\alpha u(n,d_\alpha/2))$.

    The space requirement follows from the need to store $n z_\alpha$ values in $U(n,d_\alpha/2)$.
    As before, space $\min{\{\log(n)\log(\lceil d_\alpha/2\rceil),\log(n)+\lceil d_\alpha/2\rceil\}})$ suffices to store a node in the universal tree.
    Therefore, we obtain the following overall space requirement: $O(n z_\alpha \min{\{\log(n)\log(\lceil d_\alpha/2\rceil),\log(n)+\lceil d_\alpha/2\rceil\}})$.
\end{proof}

\begin{remark}\label{rem:complexity}
The statement in Theorem~\ref{thm:elqp} uses Lemma~\ref{lem:zielonkaSize} to bound the number of leaves of the Zielonka trees by $c!$;
this factor dominates the overall solution complexity.
Hence the complexity decreases for Emerson-Lei conditions with significantly smaller Zielonka trees.
In particular, the proposed rank-lifting algorithm is quasi-polynomial for classes of games in which $|Z_\alpha|$ grows polynomially with $c$
(such as parity games).
We note that this is also the case in reductions that use the full structure of the Zielonka tree to 
reduce EL games to parity games \cite{DBLP:journals/theoretics/CasaresCFL24}. However, it is \emph{not}
the case if one uses the standard LAR-construction \cite{DBLP:conf/stoc/GurevichH82} to reduce EL to parity games
(cf. Theorem~\ref{theorem:existing results}).
In both reductions to parity games, the memory usage is not optimal as it incorporates the memory of both players into the parity game.

\end{remark}

\section{Symbolic Solution using Universal Trees}\label{sec:symbolic}

We now show how hierarchical (hence universal) trees can be used to bound the recursive descent in Zielonka-McNaughton's algorithm
for Emerson-Lei games.
This generalizes previous quasi-polynomial Zielonka-style algorithms for parity games \cite{DBLP:conf/mfcs/Parys19,DBLP:journals/lmcs/LehtinenPSW22}
to EL games
by adding additional recursive calls according to the structure of the condition's Zielonka tree.

Consider an EL game $G=\pair{A,\alpha}$ with set $V$ of vertices, where $\alpha$ is over set $C$ of colors and $\gamma:V\to 2^C$ is the coloring function.
Let $Z(\alpha)=\langle T,\tau\rangle$ be the Zielonka tree of $\alpha$.
The mutually recursive procedures in Algorithm~\ref{alg:ELQPZ0} take as input a (sub)game (also denoted $G$), a node $t\in T$ and two hierarchical trees $U^0$ and $U^1$.
We prove that, for hierarchical trees with suitable parameters, these procedures correctly compute the winning regions of the respective players (Corollary~\ref{cor:symbolic algorithms correct}).
To this end, we first recall standard notions about closed sets and dominions in games and specialize them to EL games. 

Let $i\in \set{0,1}$ denote one of the players in $G$.
A set $X\subseteq V'\subseteq V$ is \emph{$i$-closed} in $V'$ if for all $v\in X\cap V_i$, we have $E(v)\cap X\neq\emptyset$ and
for all $v\in X\cap V_{1-i}$, we have $(E(v)\cap V')\subseteq X$; that is, when restricting the arena to $V'$, player $1-i$ cannot move outside of $X$ and player $i$ has a (memoryless) strategy to stay within $X$.
An $i$-\emph{dominion} in $V'\subseteq V$ is a set $D\subseteq V'$ of vertices 
such that for all $v\in D$, player $i$ has a (general) strategy to win according to $\alpha$ while staying in $D$ (that is, a strategy with which every play starting at $v$ and played on $V'$ stays within $D$, and furthermore is winning for player $i$); in particular, $i$-dominions in $V'$ are $i$-closed in $V'$.

\begin{algorithm}[bt]
	\caption{\label{alg:ELQPZ0}Symbolic solution of EL games with Universal Trees}
	\DontPrintSemicolon

    \textbf{procedure }\textsc{SolveELEven}($G,t,U^0,U^1$):
    
    \textbf{let }$U^1=\langle U^1_1,\ldots, U^1_k\rangle$
    
    $G_1\gets G$
    
    \For{$i\gets 1$\upshape\textbf{ to }$k$}{
    
    $S_i\gets \emptyset$

    \ForAll{$(s \in \child(t))$}{
		$N_{s} \gets \{v\in G_i \mid \gamma(v)\cap (\tau(t)\setminus \tau(s))\neq \emptyset\}$ \;
		
		$G'_i \gets G_i \setminus\attr^{G_i}_0(N_s)$\;

        $W_s \gets \textsc{SolveELOdd}(G'_i,s,U^0,U_i^1)$\;
        
        $S_i\gets S_i\cup W_s$

	}
    $G_{i+1} \gets G_i \setminus \attr^{G_i}_{1}(S_i)$\; 

    }
	\Return{$G_{k+1}$}\medskip
    
    \textbf{procedure }\textsc{SolveELOdd}($G,t,U^0,U^1$):
    
    \textbf{let }$U_0=\langle U^0_1,\ldots, U^0_k\rangle$
    
    $G_1\gets G$
    
    \For{$i\gets 1$\upshape\textbf{ to }$k$}{
    
    $S_i\gets \emptyset$

    \ForAll{$(s \in \child(t))$}{
		$N_{s} \gets \{v\in G_i \mid \gamma(v)\cap (\tau(t)\setminus \tau(s))\neq \emptyset\} $\;
		
		$G'_i \gets G_i \setminus\attr^{G_i}_1(N_s)$\;
        $W_s \gets \textsc{SolveELEven}(G'_i,s,U_i^0,U^1)$\;
        
        $S_i\gets S_i\cup W_s$

	}
    $G_{i+1} \gets G_i \setminus \attr^{G_i}_{0}(S_i)$\;

    }
	\Return{$G_{k+1}$}
    
\end{algorithm}

We collect various properties of dominions.

\begin{lemmarep}\label{lemma:domclosdom}
Let $X\subseteq V'\subseteq V$, let $D\subseteq V'$ be an $i$-dominion in $V'$, and let $X$ be $(1-i)$-closed in $V'$. Then $D\cap X$ is an $i$-dominion in $X$.
\end{lemmarep}
\begin{proof}
Let $v\in D\cap X$. We have to show that 
there is a strategy for player $i$ such that every play played on $X$ that starts at $v$ and follows that strategy stays within $D\cap X$ and is won by player $i$.
Since $D$ is an $i$-dominion in $V'$, there is a strategy $\sigma$ for player $i$ such that
every play played on $V'$ that starts at $v$ and follows $\sigma$ stays within $D$ and is won by player $i$.
Play on $X$ according to $\sigma$, starting from $v$. 
Since $D$ is $i$-closed in $V'$, any move 
by player $1-i$ stays within $D$ (hence in $D\cap X$).
Since $X$ is $(1-i)$-closed in $V'$,
all moves by player $i$, in particular the ones prescribed by $\sigma$, also stay within $D\cap X$. 
Every play on $X\cap D$ resulting from this interaction is compatible with strategy $\sigma$ and hence
won by player $i$. Hence $D\cap X$ indeed is an $i$-dominion in $X$.
\end{proof}

\begin{lemmarep}\label{lemma:domclosdomi}
Let $V'\subseteq V$, let $D\subseteq V'$ be an $i$-dominion in $V'$, and let $X\subseteq D$ be $i$-closed in $V'$. Then $D\setminus \attr^{V'}_{i}(X)$ is an $i$-dominion in $V'\setminus \attr^{V'}_{i}(X)$.
\end{lemmarep}
\begin{proof}
Let $v\in D'=D\setminus \attr^{V'}_{i}(X)$. It suffices to show that 
there is a strategy for player $i$ such that every play played 
on $V'\setminus \attr^{V'}_{i}(X)$ that starts at $v$ and follows that strategy stays within $D'$ and is won by player $i$.
Since $D$ is an $i$-dominion in $V'$, there is a strategy $\sigma$ for player $i$ such that
every play played on $V'$ that starts at $v$ and follows $\sigma$ stays within $D$ and is won by player $i$.
To see that $D'$ is an $i$-dominion, play on $D'$ according to $\sigma$, starting from $v$ and note that any move by player $1-i$ that stays
within $V'\setminus \attr^{V'}_{i}(X)$ also stays within $D'=D\setminus \attr^{V'}_{i}(X)$ since
$D$ is an $i$-dominion. Every move by player $i$ stays within 
$D\setminus \attr^{V'}_{i}(X)$ as well: since any node that belongs to player $i$ and has an edge to $\attr^{V'}_{i}(X)$ is contained in $\attr^{V'}_{i}(X)$,
no node in $D'$ that belongs to player $i$ has an edge to $\attr^{V'}_{i}(X)$. In particular,
the moves prescribed by $\sigma$ remain within $D'$.
Hence every play starting in $D'$ and played according to $\sigma$ stays within $D'$ and is won by player $i$ (since
$\sigma$ is a winning strategy for player $i$). Thus $D'$ is an $i$-dominion in $V'\setminus \attr^{V'}_{i}(X)$.
\end{proof}

Next, we show that non-empty $i$-dominions whose sets of occurring colors do not satisfy the objective of
player $i$ contain non-empty sub-dominions that do not contain at least one of these colors. In more detail,
given a non-empty dominion $D$ of player $i$ and a Zielonka tree node $t$ such that $t$ is not winning for that player (that is, 
$\tau(t)\models \alpha$ iff $i=1$) and such that the label $\tau(t)$ of $t$ contains all colors that occur in the dominion,
there is a non-empty sub-dominion $D'\subseteq D$ for player $i$ such that $D'$ does not contain at least one of the colors from $\tau(t)$.
The latter can be formalized by stating that there is a child $s$ of the Zielonka node $t$ such that all colors that occur in 
$D'$ are contained in $\tau(s)$ (which is a strict subset of $\tau(t)$).
Intuitively, this means that within their dominions, players have winning strategies that eventually avoid dangerous colors.

\begin{lemmarep}\label{lem:subdominion}
Let $V'\subseteq V$, let $D\subseteq V'$ be a non-empty $i$-dominion in $V'$ and let $t\in T$ be a node in $Z(\alpha)$ such that $\tau(t)\models \alpha$
iff $i=1$ and such that all colors that occur in $D$ are contained in $\tau(t)$. Then there is $s\in \child(t)$ such that 
$D$ contains a non-empty $i$-dominion $D'\subseteq D$ (in $V'$) and the colors that occur in $D'$ are contained in $\tau(s)$. 
\end{lemmarep}
\begin{proof}
Let $v\in D$. As $D$ is an $i$-dominion, there is a strategy for player $i$ such that
every play that starts at $v$ and follows that strategy stays within $D$ and is won by player $i$.
Fix this strategy.
Pick a vertex $v'\in D$ such that $v'$ can be reached from $v$ by a finite play $\rho_{v'}$ with the strategy and 
such that $A_{v'}$ is winning for player $i$, where $A_{v'}$ is the set of colors that can be visited from
$v'$ by extending $\rho_{v'}$ following the strategy. 
We point that the strategy fixes only the moves of player $i$.
Hence there may be many ways to extend $\rho_{v'}$ according to the fixed strategy.
We will show that $A_{v'}$ then is contained in the label of some $s\in \child(t)$.

To see that such a vertex exists, assume towards a contradiction that for every vertex 
$v'\in D$ such that $v'$ can be reached from $v$ by a finite play $\rho_{v'}$ with the strategy, the set $A_{v'}$ of colors is winning for player $1-i$.
Then we can construct a play $\rho$ in $D$ that follows the fixed strategy and is winning for player $1-i$, in contradiction to the fixed
strategy being a winning strategy for player $i$. We construct $\rho$ by following the strategy at vertices belonging to 
player $i$ and by fairly exploring all edges at vertices belonging to player $1-i$. For every vertex $v''$ visited infinitely often by $\rho$,
every color from $A_{v''}$ is visited infinitely often as well, since $\rho$ visits every vertex that can be visited from
$v''$ by extending $\rho_{v''}$ following the strategy.

Fix some such vertex $v'$ and
define $D'$ to be the set of game vertices that can be reached from $v'$ by extending $\rho_{v'}$ using the fixed strategy.
By definition, $D'$ is $i$-closed and furthermore an $i$-dominion. We have shown that the set $A_{v'}$ of colors that
occur in $D'$ is winning for player $i$.
Thus there is a child $s$ of $t$ such that $A_{v'}\subseteq \tau(s)$, as required.
\end{proof}

\subsubsection*{Partition of dominions}
We show that the partition of winning regions described in Section~\ref{sec:el} also partitions dominions.
This further refines the insight of Lemma~\ref{lem:subdominion} as follows.
Let $D$ be a non-empty $i$-dominion in a subgame $G$, let $n_D$ denote the number of vertices in $D$, 
let $t\in T$ be a node in the Zielonka tree such that $\tau(t)\models\alpha$ 
iff $i=1$ (so that player $i$ loses plays that visit all colors from $\tau(t)$), and
assume that all colors that occur in $D$ are contained in $\tau(t)$. 
Let $\sigma$ denote a strategy with which player $i$ wins from every vertex in $D$.

Intuitively, player $i$ does not have complete control over the colors visited by playing according to their winning strategy
$\sigma$ in $D$. 
However, at each point, they may pick a subobjective $s\in \child(t)$
such that $\sigma$ plays optimally for that subobjective, attempting to satisfy it. Player $1-i$ may accept 
that and lose by staying in the region for that subobjective forever, or they may be able to spoil the subobjective $s$ by forcing a visit to some vertex 
that sees some color from $\tau(t)$ that is not contained in $\tau(s)$. However, the latter situation may happen only finitely often since otherwise,
player $1-i$ would have a strategy to win in $D$.

We inductively define subsets $D_j\subseteq D$ (for $0\leq j \leq n_D$). 
Put $D_{0}=\emptyset$.
For $j>0$, let $A_j$ denote the set of vertices $v\in D$ such that $v\notin D_{j'}$ for any $j'<j$,
and such that there is $s\in \child(t)$ such
that for every play $\rho$ starting at $v$ and following the strategy $\sigma$,
either all colors visited by $\rho$ are contained in $\tau(s)$, or
$\rho$ visits a color from $\tau(t)\setminus \tau(s)$ and the first vertex in $\rho$ that does so is contained in $D_{j'}$ for some $j'<j$. 
Put $D_{j}=\attr^{D}_{i}(A_j)$. For $v\in D_j$, we say that $v$ has \emph{degree} $j$.

We point out that for $v\in V_i\cap A_{j}$, $E(v)\cap D_{j'}=\emptyset$
for all $j'<j$. 
This holds since $v\in A_j$ implies $v\notin D_{j'}=\attr^{D}_{i}(A_{j'})$ for all $j'<j$.

Crucially, the sets $D_j$ (with $1\leq j\leq n_D$) partition $D$:

\begin{lemma}\label{lemma:dominiondegrees}
Let $D$ be an $i$-dominion of size $n_D$ and let $t\in T$ such that $\tau(t)\models \alpha$ iff $i=1$
and such that the colors that occur in $D$ are contained in $\tau(t)$. Then $D=D_1\cup \ldots \cup D_{n_D}$.
\end{lemma}
\begin{proof}
We point out that $D$ is contained in the winning region of player $i$. 
If $i=0$,
then the partition of $D$ according
to the degree of vertices corresponds directly to the partition of the winning region that is considered in
Section~\ref{sec:el}: the vertices $v\in D\cap W_{2j+1}$ 
correspond to the set $A_{j+1}$ while vertices $v\in D\cap W_{2j+2}$ correspond 
to $\attr^{D}_{0}(A_{j+1})=D_{j+1}$, where $0\leq j\leq n_D$; here, we do not make use of the additional partitions of the attractors $W_{2j'}$
($1\leq j'\leq n_D$) described in Section~\ref{sec:el}. 
This partition satisfies the requirements regarding
degrees by Lemma~\ref{lemma:partition}, that is, we indeed have $D_{j+1}=D\cap W_{2j+2}$.

If $i=1$, then we relate $D$ to the partition of the winning region according to the dual condition $\neg\alpha$ (which is an Emerson-Lei condition as well),
and again use Lemma~\ref{lemma:partition} to show that this gives rise to the partition $D_1\cup \ldots \cup D_{n_D}$.
\end{proof}

\begin{theoremrep}\label{theorem:dominions}
Let $G$ be an Emerson-Lei game, $t$ a node in the Zielonka tree such that all colors that are in $G$ are 
contained in $\tau(t)$, $U^0$ an $(n_0,k)$-hierarchical tree and $U^1$ an $(n_1,k)$-hierarchical tree,
where $|\tau(t)|\leq 2k$. Then 
\begin{itemize}
\item[--] the set computed by $\textsc{SolvELEven}(G,t,U^0,U^1)$ is $0$-closed in $G$, contains all $0$-dominions of size up to $n_0$ in $G$ and does not intersect
with any $1$-dominion of size up to $n_1$ in $G$;
\item[--] the set computed by $\textsc{SolvELOdd}(G,t,U^0,U^1)$ is $1$-closed in $G$, contains all $1$-dominions of size up $n_1$ in $G$ and does not intersect
with any $0$-dominion of size up to $n_0$ in $G$.
\end{itemize}
\end{theoremrep}
\begin{proof}
The proof is by induction on the sum $\mathsf{level}(t)+n_0+n_1$.
We consider just the first item, the proof of the second item is dual. 

We first show that $\textsc{SolvELEven}(G,t,U^0,U^1)$ is $0$-closed in $G$.
To this end, we show by induction that for all $0\leq i \leq k$, the set $G_{i+1}$ is $0$-closed in $G$.
For $i=0$, $G_1=G$ trivially is $0$-closed in $G$.
For $i>1$, first consider the case of a vertex $v\in G_{i+1}\cap V_1$. We have to show that $E(v)\cap G\subseteq G_{i+1}$. This follows from the induction hypothesis since
whenever a vertex from $E(v)$ is contained both in $G_i$ and in $\attr^{G_i}_{1}(S_i)$ (meaning that the vertex is removed in the current iteration of the loop), 
then $v$ is also contained in $\attr^{G_i}_{1}(S_i)$
(and removed in the same iteration of the loop). Next, consider $v\in G_{i+1}\cap V_0$. We have to show that $E(v)\cap G_{i+1}\neq\emptyset$. 
By the inductive hypothesis, $G_i$ is $0$-closed, that is, $E(v)\cap G_{i}\neq\emptyset$.
It follows from $v\in G_{i+1}$ that $E(v)$ is not a subset of $\attr^{G_i}_{1}(S_i)$. Hence there is some $w\in E(v)\cap G_{i+1}$ that is not 
contained in $\attr^{G_i}_{1}(S_i)$ so that $w\in E(v)\cap G_{i+1}\neq\emptyset$, as required.

Next we show that $\textsc{SolvELEven}(G,t,U^0,U^1)$ contains all $0$-dominions in $G$ that have size at most
$n_0$. Let $D$ be such a $0$-dominion. Initially, we have $D\subseteq G_1 = G$. It suffices to show that $D$ does not intersect with
$\attr^{G_i}_{1}(S_i)$ for any $i$ since then $D$ is contained in $G_i$ for all $i$; in particular, $D$ then
is contained in the computed result $G_{k+1}$. We continue to show the claim that for all $i$, $D\cap\attr^{G_i}_{1}(S_i)=\emptyset$. 
So let $1\leq i\leq k$. By the induction hypothesis, for each $s\in \child(t)$, the set $W$ 
(computed in line $8$) does
not intersect with any $0$-dominion in $G'_i$ that has size at most $n_0$. Assume towards a contradiction
that there is some $v\in D\cap \attr^{G_i}_{1}(S_i)$. 
We note that $\attr^{G_i}_{1}(S_i)$ is $1$-closed. Consider the play that is obtained from
letting the winning strategy for player $0$ to win within the dominion compete against the player $1$ strategy
with which they attract to $S_i$. Clearly, this play stays within $D$ but eventually reaches $S_i$.
It follows that $S_i$ intersects with $D$, which in turn implies that 
there is some $s\in \child(t)$ for which the set $W$ (computed in line $8$) does intersect with $D$.
The corresponding set $G'_i$ is $1$-closed in $G_i$ so that
by Lemma~\ref{lemma:domclosdom}, $D\cap G'_i$ is a $0$-dominion in $G'_i$,
so that $W$ does intersect with a $0$-dominion in $G'_i$. The contradiction follows since $|D\cap G'_i|\leq|D|\leq n_0$.

It remains to show that $\textsc{SolvELEven}(G,t,U^0,U^1)$ does not intersect with any $1$-dominion in $G$ that has size at most
$n_1$. Let $D$ be such a $1$-dominion in $G$. We note that $\tau(t)\models \alpha$ and that all colors that occur in $G$, hence also all colors that occur in $D$,
are contained in $\tau(t)$. By Lemma~\ref{lemma:dominiondegrees}, every vertex in $D$ has degree at most $n_1$ (in $D$)
and there is some $r\leq n_1$ such that $D$ is partitioned into pairwise disjoint, non-empty sets $D_1, \ldots ,D_{r}$,
where, for $1\leq j\leq r$, all vertices in the set $D_j$ have degree $j$. 
Here, we can pick $r\leq n_1$ since while we have $D_1\neq\emptyset$ by Lemma~\ref{lem:subdominion}, 
there may be some $j'$ such that $D_{j'}=\emptyset$ (and then $D_{j''}=\emptyset$ for all $j''\geq j'$); pick $r=j'$.

Put $d_j:=|D_{j}|$. We point out that $\sum_i d_i\leq n_1$. 
Recall that $U^1$ is a hierarchical $(n_1,k)$-universal tree. Hence 
there are indices $i_1<\ldots <i_r$ such that $U^1_{i_j}$ is a hierarchical $(d_j,k-1)$-universal tree.
To conclude the proof, we show by induction on $j$ that at the end of the $i_j$-th iteration of the loop at line $3$ in the procedure $\textsc{SolvELEven}(G,t,U^0,U^1)$,
the remaining game $G_{i_j+1}$ does not intersect with any set $D_{j'}$ such that $1\leq j'\leq j$, that is, all vertices from $D$ that have degree at most $j$ have
been removed in $G_{i_j+1}$. As $n_1\leq k$, this will imply that $G_{k+1}$ does not intersect with $D$ as required.

So let $1\leq j \leq r$. Since $G_{i_j+1}=G_{i_j}\setminus \attr^{G_{i_j}}_{1}(S_{i_j})$, 
it suffices to show that every vertex from 
$D_{j}\cap G_{i_j}$ is contained in $\attr^{G_{i_j}}_{1}(S_{i_j})$.
So let $v\in D_{j}\cap G_{i_j}$. As $v$ has degree $j$, 
we have $v\in\attr^{D}_{1}(A_j)$ where $A_j$ is the set of vertices
$v\in D$ such that there is $s\in R(t)$ such that
strategy $\sigma$ ensures that from $v$ on, only colors from $\tau(s)$ are visited.
We point out that $G_{i_j}$ is obtained from $G$ by repeatedly removing
sets of the form $\attr^{G_h}_{i}(S_i)$ for some $h$ and $S_i$. 
Also, the set $D''=A_j\cup D_{j-1}\cup \ldots \cup D_0$ is a $1$-dominion in
$G$ since $D$ is a $1$-dominion in $G$.
By repeated application
of Lemma~\ref{lemma:domclosdomi}, starting with the dominion $D''$,
the set $D'=A_j\cap G_{i_j}$ forms a $1$-dominion in $G_{i_j}$.
Intuitively, $D'$ is the part of $A_j$ that still remains in $G_{i_j}$ (that is, the part that has not been removed during 
the first ${i_j}$ iterations of the loop).
Also, $|A_j|\leq d_j$ so that $D'$ has size at most $d_j$.
Next, we show that $D'\subseteq S_{i_j}$, where $S_{i_j}$ is the set that has been
computed at the end of the $i_j$-the iteration
of the loop. So consider $w\in D'$ and let $A^w_j$ denote the set of all
vertices reachable from $w$ by following the strategy $\sigma$. Since $D'$ is an
$i$-dominion in $G_{i_j}$ that has size at most $d_j$, $A^w_j\subseteq D'$ is an
$i$-dominion in $G_{i_j}$ of size at most $d_j$.
In the induction base case, we have $j=1$ so that 
no vertex with degree less than $j$ is reachable from any vertex in $A_j$.
If $j>1$, then $G_{i_j}$ does not contain a vertex of degree less than $j$
by the inner inductive hypothesis
so that again
no vertex with degree less than $j$ is reachable from any vertex in $A_j$.
As $w\in A_j$ and $A_j$ does not contain any vertices
of degree less than $j$, there is $s\in \child(t)$ such that
all colors that occur in $A^w_j$ are contained in $\tau(s)$, that is,
no color from $N_s=\tau(t)\setminus \tau(s)$ occurs in $A^w_j$.
It follows that $A^w_j$ also is an
$i$-dominion in $G'_{i_j}=G_{i_j}\setminus \attr^{G_{i_j}}_{0}(N_s)$ (line $7$ of the algorithm).
As $U^1_{i_j}$ is a hierarchical $(d_j,k-1)$-universal tree, the outer induction
hypothesis gives us that the dominion $A^w_j$ is contained in $\textsc{SolvELOdd}(G'_{i_j},s,U^0,U^1_{i_j})$. Hence 
$A^w_j$ is also contained in $S_{i_j}$ at the end of the loop that starts at line $5$. As $w\in A^w_j$, $w\in S_{i_j}$, as claimed. Hence $D'\subseteq S$.
As player $1$ can attract from $v$ to $A_j$ in $D$ by assumption, they can attract from $v$ to $D'=A_j\cap G_{i_j}$ in $G_{i_j}$. Thus $v\in\attr^{G_{i_j}}_{1}(D')\subseteq \attr^{G_{i_j}}_{1}(S_{i_j})$, as required.
\end{proof}

\begin{corollaryrep}
Let $G$ be an EL game with $n$ vertices and objective $\alpha$ using $c$ colors, let $t\in T$ be the root of $Z(\alpha)$ and assume $\tau(t)\models\alpha$. Then 
$\textsc{SolvELEven}(G,t,U(n,c/2),U(n,c/2))$ computes the winning region of player~$0$ in $G$.
\label{cor:symbolic algorithms correct}
\end{corollaryrep}
\begin{proof}
The result follows directly from Theorem~\ref{theorem:dominions} as the winning region of player $0$ is the union of all $0$-dominions in $G$.
\end{proof}

Algorithm~\ref{alg:ELQPZ0} computes just the winning regions and does not directly yield winning strategies. This entails
improved space complexity in comparison to the rank-lifting algorithm from Section~\ref{sec:elranking} (that also computes strategies):

\begin{lemma} Algorithm~\ref{alg:ELQPZ0} decides Emerson-Lei games with $n$ vertices, $m$ edges, $c$ colors and
Zielonka trees of size $z_\alpha$ in time
$O(mnz_\alpha u(n,c/2)^2)$ and space $O(m+nc)$. 
\end{lemma}
\begin{proof}
For each call, the algorithm computes at most $k+1$ attractors, one immediately before the call, and $k$
within the outer loop. This yields the factor $mn$ in the
time complexity, as $k\leq n$ and a single attractor can be computed in time $O(m)$.
Furthermore, each call has $k\cdot |\child(t)|$ recursive calls with reduced parameters.
The sum of the sizes of the Zielonka trees considered in the recursive calls
is exactly the size of the tree below $t$ minus $1$, yielding the overall factor $z_\alpha$.
For each $s\in \child(t)$, the sum of the sizes of the universal
trees in the recursive calls with parameter $s$ is just $|U_1|-1$ ($|U_0|-1$, respectively), yielding the overall factor $u(n,c/2)^2$.

Regarding space complexity, we note that the depth of the recursion is at most $c$, and at each level of the recursion, it suffices to keep
a constant number of sets of size at most $n$ in memory; in particular, the Zielonka tree can be computed on the fly
by computing, in each recursive call, the children $s$ of the current node $t$ as well as their labels.
Thus no additional memory is required to store the full Zielonka tree.
\end{proof}

\section{Conclusion}
We have introduced notions for formally ranking progress in Streett games and, more generally, Emerson-Lei games,
and we have shown that the existence of a finite ranking characterizes winning in such games.
The proposed rankings use Zielonka trees to incorporate memory that winning strategies may require, yet crucially they use universal trees 
as domain for the ranking. Therefore, the rankings combine Zielonka trees and universal trees,
lifting the breakthrough quasi-polynomial method devised for parity games to the more general Streett and Emerson-Lei games in a natural way.

Technically, these results hinge on the employed universal trees being hierarchical in a precise sense; however, we show that
standard universal trees, including the quasi-polynomial sized ones, are hierarchical (see Lemma~\ref{lemma:hierarchical}).

We propose both an asymmetric rank-lifting algorithm for computing winning regions and strategies in Streett and Emerson-Lei games, as well as a symmetric (and symbolic) attractor-based algorithm for the computation of winning regions in Emerson-Lei games, all building on the notions of universal-tree-based rankings introduced before.
By using quasi-polynomial sized universal trees as ranking domain, the proposed algorithms improve both the time and the space complexity for solving Streett and Emerson-Lei games 
in algorithms that allow for the extraction of winning strategies.

For future work, we mark the usage of alternating-cycle decompositions~\cite{DBLP:journals/theoretics/CasaresCFL24} instead of Zielonka trees within the proposed algorithms, and the reduction of Emerson-Lei alternating automata to weak alternating automata, based on Emerson-Lei rankings that use universal trees.

\clearpage
\bibliographystyle{plain}
\bibliography{bib}

\end{document}